\documentclass[smallextended]{svjour3}       
\smartqed  
\usepackage{graphicx}
\usepackage{amsfonts,amsmath,amssymb,proof,mathrsfs}
\usepackage[all,cmtip]{xy}

\usepackage{moresize}

\usepackage{natbib,hyperref}

\newcommand{\Lbx}{\mathbf{Lb}^{\boldsymbol{*}}}
\newcommand{\EL}{\boldsymbol{!} \mathbf{MALC}^{\boldsymbol{*}}}
\newcommand{\ELM}{\boldsymbol{!} \mathbf{L}^{\boldsymbol{*}}}

\newcommand{\LLsM}{\boldsymbol{!}^{\mathbf{r}} \mathbf{L}^{\boldsymbol{*}}}
\newcommand{\LLs}{\boldsymbol{!}^{\mathbf{r}} \mathbf{MALC}^{\boldsymbol{*}}}
\newcommand{\PMod}{\langle\rangle}
\newcommand{\NMod}{[]^{-1}}
\newcommand{\BS}{\mathop{\backslash}}
\newcommand{\SL}{\mathop{/}}
\newcommand{\U}{\mathbf{1}}
\newcommand{\One}{\U}
\newcommand{\LU}{\mathbf{L}_{\U}}
\newcommand{\Ld}{\mathbf{L}^{\boldsymbol{*}}(/)}

\newcommand{\yd}{\to}

\newcommand{\sysA}{\boldsymbol{!}_{\mathbf{b}}^{\mathbf{2015}} \mathbf{MALC^{\boldsymbol{*}}b}\mathrm{(st)}}
\newcommand{\sysB}{\boldsymbol{!}_{\mathbf{b}}^{\mathbf{2018}} \mathbf{MALC^{\boldsymbol{*}}b}\mathrm{(st)}}
\newcommand{\sysAa}{\sysA'}

\newcommand{\sysBa}{\sysB'}
\newcommand{\sysBb}{\sysB''}
\newcommand{\sysBr}{\boldsymbol{!}_{\mathbf{b}}^{\mathbf{2018}} \mathbf{MALCb}\mathrm{(st)}}
\newcommand{\sysBar}{\sysBr'}

\newcommand{\sysAfl}{\boldsymbol{!}_{\mathbf{b}}^{\mathbf{2015}} \mathbf{MALC^{\boldsymbol{*}}b}}
\newcommand{\sysBfl}{\boldsymbol{!}_{\mathbf{b}}^{\mathbf{2018}} \mathbf{MALC^{\boldsymbol{*}}b}}
\newcommand{\sysBrfl}{\boldsymbol{!}_{\mathbf{b}}^{\mathbf{2018}} \mathbf{MALCb}}

\newcommand{\LsysAfl}{\boldsymbol{!}_{\mathbf{b}}^{\mathbf{2015}} \mathbf{L^{\boldsymbol{*}}b}}
\newcommand{\LsysBfl}{\boldsymbol{!}_{\mathbf{b}}^{\mathbf{2018}} \mathbf{L^{\boldsymbol{*}}b}}
\newcommand{\LsysBrfl}{\boldsymbol{!}_{\mathbf{b}}^{\mathbf{2018}} \mathbf{Lb}}

\newcommand{\LsysA}{\boldsymbol{!}_{\mathbf{b}}^{\mathbf{2015}} \mathbf{L^{\boldsymbol{*}}b}\mathrm{(st)}}
\newcommand{\LsysB}{\boldsymbol{!}_{\mathbf{b}}^{\mathbf{2018}} \mathbf{L^{\boldsymbol{*}}b}\mathrm{(st)}}

\newcommand{\LsysAa}{\LsysA'}
\newcommand{\LsysBa}{\LsysB'}

\newcommand{\LsysBr}{\boldsymbol{!}_{\mathbf{b}}^{\mathbf{2018}} \mathbf{Lb}\mathrm{(st)}}
\newcommand{\LsysBar}{\LsysBr'}

\newcommand{\MALCsb}{\mathbf{MALC}^{\boldsymbol{*}}\mathbf{b}}

\newcommand{\mconj}{\cdot}

\newcommand{\Lc}{\mathcal{L}}
\newcommand{\Ac}{\mathcal{A}}

\newcommand{\pf}{\mathfrak{p}}
\newcommand{\twp}{\overline{w}^{\pf}}
\newcommand{\typ}{\overline{y}^{\pf}}
\newcommand{\tzp}{\overline{z}^{\pf}}

\newcommand{\bap}{\mathbf{a}^{\pf}}
\newcommand{\bbp}{\mathbf{b}^{\pf}}
\newcommand{\bcp}{\mathbf{c}^{\pf}}
\newcommand{\bdp}{\mathbf{d}^{\pf}}
\newcommand{\bep}{\mathbf{e}^{\pf}}
\newcommand{\bfp}{\mathbf{f}^{\pf}}

\newcommand{\CONTR}{\mathrm{contr}}

\newcommand{\CUT}{\mathrm{cut}}

\newcommand{\SMALC}{\mathbf{SMALC}_\Sigma}
\newcommand{\MALC}{\mathbf{MALC}^{\boldsymbol{*}}}

\newcommand{\Var}{\mathrm{Var}}

\newcommand{\Gc}{\mathcal{G}}

\newcommand{\Af}{\mathfrak{A}}
\newcommand{\Bc}{\mathcal{B}}

\newcommand{\STS}{\mathcal{S}}

\newcommand{\Der}{\mathscr{D}}
\newcommand{\DerL}{\Der_{\mathrm{left}}}
\newcommand{\DerR}{\Der_{\mathrm{right}}}

\begin{document}

\title{The Multiplicative-Additive Lambek Calculus with Subexponential and Bracket
Modalities}

\titlerunning{MALC with Subexponential and Brackets}

\author{Max Kanovich          \and        Stepan Kuznetsov \and         Andre Scedrov
}


\institute{
           M. Kanovich \at
              University College London,\\\hspace*{1em}
              Gowers St., London, U.K.\\
              and National Research University Higher School of Economics,\\\hspace*{1em}
              11 Pokrovsky Blvd., Moscow, Russia\\
              \email{m.kanovich@ucl.ac.uk}
           \and
           S. Kuznetsov \at
              Steklov Mathematical Institute of the Russian Academy of Sciences,\\\hspace*{1em}
              8 Gubkina St., Moscow, Russia\\
              and National Research University Higher School of Economics,\\\hspace*{1em}
              11 Pokrovsky Blvd., Moscow, Russia\\
              \email{sk@mi-ras.ru}
           \and
           A. Scedrov \at
              University of Pennsylvania, \\\hspace*{1em}
              209 South 33rd St., Philadelphia, PA, U.S.A. \\
              and National Research University Higher School of Economics (until July 2020),\\\hspace*{1em}
              11 Pokrovsky Blvd., Moscow, Russia\\
              \email{scedrov@math.upenn.edu}
}

\date{Received: date / Accepted: date}

\maketitle

\begin{abstract}
We give a proof-theoretic and algorithmic complexity analysis for systems introduced
by Morrill to serve as the core of the CatLog categorial grammar parser. We consider two recent
versions of Morrill's calculi, and focus on their fragments including multiplicative (Lambek)
connectives, additive conjunction and disjunction, brackets and bracket modalities, and the
${!}$ subexponential modality. For both systems, we resolve issues connected with the cut rule
and provide necessary modifications, after which we prove admissibility of cut (cut elimination theorem).
We also prove algorithmic undecidability for both calculi, and show that categorial grammars based
on them can generate arbitrary recursively enumerable languages.


\end{abstract}

\section{Linguistic Introduction}

The Lambek calculus~\citep{Lambek58} was introduced for mathematical modelling of natural language syntax via {\em categorial grammars.}
The concept of categorial grammar goes back to ideas of~\citet{Ajdukiewicz} and~\citet{BarHillel}. 
The framework of categorial grammars aims to describe natural language by means of logical derivability 
(see \citet{Buszkowski2003,Carpenter,MorrillBook,MootRetore} {\em etc}). 
From the modern logical point of view, the calculus of Lambek grammars (the Lambek calculus) is a variant of Girard's linear logic~\citep{Girard} in its
non-commutative intuitionistic version~\citep{Abrusci}.
Nowadays Lambek-style categorial grammars form one framework in a family of closely related formalisms, including combinatory categorial grammars~\citep{Steedman},
categorial dependency grammars~\citep{DikovskyDekhtyar}, and others.

A categorial grammar assigns logical formulae to lexemes (words) of the language. These formulae are syntactic categories, or {\em types,} of these words.
In Lambek grammars, types are constructed using three binary connectives, namely two divisions, $\BS$ and $\SL$, and the product, $\cdot$.  

Following the usual introduction into categorial grammars, we start with the standard example: {\sl ``John loves Mary.''}  Here {\sl ``John''} and {\sl ``Mary''} receive syntactic type $N$ (noun); {\sl ``loves,''} as a transitive
verb, is of type $(N \BS S) \SL N$. Here $S$ is the syntactic category of grammatically valid sentences. Thus, a transitive verb is handled as something
that needs a noun phrase on the left and a noun phrase on the right to become a complete sentence. In the Lambek calculus, $A, A \BS B$ yields $B$, and so does
$B \SL A, A$ (the complete formulation of the Lambek calculus is presented in Section~\ref{S:MALC}). Thus, $N, (N \BS S) \SL N, N \to S$ is a theorem 
of the Lambek calculus, which validates {\sl ``John loves Mary''} as a correct sentence.

Lambek grammars are also capable of handling more sophisticated syntactic constructions, in particular, coordination ({\sl ``and,''} {\sl ``or''}) and
some cases of dependent clauses. These cases include examples like {\sl ``the girl whom John loves''} (parsed as $N$). Here the most interesting syntactic type is 
the one for {\sl ``whom'':} $(CN \BS CN) \SL (S \SL N)$. The type $CN$ stands for ``common noun,'' {\em i.e.,} a noun without article. {\sl ``Whom''} takes,
as its right argument, an incomplete sentence {\sl ``John loves,''} which lacks a noun phrase on the right to become a complete sentence
(like {\sl ``John loves Mary''}) and  is therefore of type $S \SL N$. The complete analysis of {\sl ``the girl whom John loves''} corresponds to the following
theorem of the Lambek calculus: $$N \SL CN, CN, (CN \BS CN) \SL (S \SL N), N, (N \BS S) \SL N \to N.$$
Coordination between two sentences ({\sl ``John loves Mary and Pete loves Ann''}) is handled by assigning $(S \BS S) \SL S$ to {\sl ``and.''} 

There are, however, serious limitations of the expressive power of Lambek grammars. Namely, the famous result of~\citet{PentusCF} states that
any language described by a Lambek grammar is necessarily context-free. 
On the other hand, context-freeness of real natural language syntax had been a disputed question in the linguistic community, see~\citet{PullumGazdar}. Finally,~\citet{Shieber} demonstrated a non-context-free construction in Swiss German. Though examples like Shieber's one may seem exotic, constructing context-free grammars for sophisticated natural phenomena, even if such grammars exist, is practically quite hard.
This discrepancy motivates extending and modifying the Lambek calculus in order to obtain more
powerful categorial grammar formalisms.

In this paper we consider some of these extensions. In the analysis of linguistic examples, we generally follow~\citet{MorrillBook} and
later papers by Morrill and his co-authors. 

The first extension handles the syntactic phenomenon called {\em medial extraction} by means of a subexponential modality allowing permutation.
To make it clear what medial extraction is, recall the {\sl ``the girl whom John loves''} example. In this example, the dependent clause 
{\sl ``John loves''} is a sentence which lacks a noun phrase. Let us call the place where this noun phrase is omitted a {\em gap} and denote it by 
$[]$. A sentence with a gap in the end ({\sl ``John loves $[]$,''} cf. {\sl ``John loves Mary''}) is of type $S \SL N$. Symmetrically, a gap in the
beginning yields type $N \BS S$, like for {\sl ``$[]$ loves Mary''} in {\sl ``the boy who loves Mary.''}  Here {\sl ``who''} receives type
$(CN \BS CN) \SL (N \BS S)$. Unfortunately, this does not cover dependent clauses in which the gap is located in the middle of the sentence, {\em i.e.,}
examples like {\sl ``the girl whom John met $[]$ yesterday.''} This dependent clause is neither of type $S \SL N$, nor of type $N \BS S$.

Medial extraction can be handled by adding a subexponential modality (cf.~\citet{KanKuzNigSce2018Dale}), denoted by ${!}$, which allows permutation.
In general, the Lambek calculus is non-commutative, thus, the order of the words in a sentence matters. For formulae of the form ${!}A$, however,
permutation is allowed, and they can be freely moved. Now the gap gets type ${!}N$ and can be relocated to an arbitrary place of the dependent clause;
the clause in whole receives type $S \SL {!}N$.

Another issue connected to dependent clauses is overgeneration ({\em i.e.,} wrong judgement of incorrect syntactic structures as valid ones), which arises
when dependent clauses and {\sl ``and''}-coordination appear together. An example is $^*${\sl ``the girl whom John loves Mary and Pete loves.''}  This is 
not a correct noun phrase (which is denoted by the asterisk put before it). Unfortunately, {\sl ``John loves Mary and Pete loves $[]$''} is still
of type $S \SL N$ (cf. {\sl ``John loves Mary and Pete loves Ann''} being of type $S$), which incorrectly validates our example as a noun phrase.
Another example is $^*${\sl ``the paper that John saw the person who wrote''} (again, we have {\sl ``John saw the person who wrote $[]$''} is of type
$S \SL N$).

These wrong derivations can be cut off using the mechanism of {\em brackets}~\citep{Morrill1992,MoortgatMultimodal}, which introduces controlled
non-associa\-ti\-vi\-ty. Brackets are instantiated by special bracket modalities (see Section~\ref{S:calculi} for details) and embrace certain parts of the sentence
into {\em islands.} Islands typically include {\sl and}-coordinated sentences, {\sl that}-clauses, gerund clauses, {\em etc.} Brackets (borders of islands)
cannot be penetrated by the permutation rules for ${!}N$. Thus, the dependent clause with brackets inside is no longer of type $S \SL {!}N$, and the whole wrong derivation
gets invalidated.

Finally, we consider
a more rare syntactic phenomenon called {\em parasitic extraction,} a typical example of which is given by the following noun phrase:
{\sl ``the paper that John signed without reading.''}  In this example we have {\em two} gaps: 
{\sl ``John signed $[]$ without reading $[]$,''} and both gaps should be filled with {\em the same} object of type $N$:
{\sl ``John signed \underline{the paper} without reading \underline{the paper}.''}
Of course, one can think of examples with three and more gaps, like {\sl ``the paper that the author of $[]$ signed $[]$ without reading $[]$,''} and so on.
 In a series of papers~\citep{Morrill2014,MorrillValentin,MorrillLACompLing,MorrillPhilosophy,Morrill2018JLM,Morrill2019}, Morrill, with his co-author Valent\'{\i}n, uses several different 
calculi for handling parasitic extraction. All these approaches, however, use a subexponential modality the {\em contraction rule,} 
which makes proof search problematic and
often yields algorithmic undecidability. Generally, contraction is a rule of the form
$$
\infer{\ldots, {!}A, \ldots \to C}{\ldots, {!}A, \ldots, {!}A, \ldots \to C}
$$ 
Morrill and his co-authors, however, suggest more sophisticated versions of contraction, which involve brackets. The general idea of their approaches
is as follows: in the situation of parasitic extraction, only one gap lies plainly in the dependent clause; other gaps, which are called
parasitic, reside in bracketed subislands of the clause. Moreover, they can get nested.  Thus, the contraction rule becomes highly non-standard.

Morrill's systems differ one from another in the rules for ${!}$. 
In this article, we give a logical analysis for two of these systems. One is presented in~\citet{MorrillValentin,MorrillLACompLing} and is closely related to the
one of~\citet{MorrillPhilosophy}. The other one is from~\citet{Morrill2019,Morrill2018JLM}. For both systems, we discuss issues connected to cut elimination, and then
prove cut elimination for modified versions of these systems. Next, we provide a generic method of encoding semi-Thue systems in extensions of the Lambek calculus with
subexponential modalities, and use this method to prove undecidability of the derivability problems for Morrill's systems. We also show that categorial grammars 
based on these systems generate all recursively enumerable languages. Finally, using methods of~\citet{BuszkoZML}, we strengthen these algorithmic results by
restricting ourselves to smallest reasonable fragments, which includes only one division, subexponential, brackets, and bracket modalities.

This journal article extends our conference papers in the 21st International Symposium on Fundamentals of Computation Theory, FCT 2017, held
in Bordeaux in September 2017~\citep{KanKuzSceFCT},
and  in the 24th Conference on Formal Grammar, FG 2019, held in Riga in August 2019~\citep{KanKuzSceFG19}. However, here we provide a significant refinement of the results presented in the FCT~'17 and FG~'19 papers. First, here we consider the system
with additive connectives. This makes cut elimination results stronger. 
 Second, besides undecidability, we also show that categorial grammars based on each of the calculi in question generate
{\em all} recursively enumerable languages, not just one $\Sigma_1^0$-hard one (Section~\ref{S:grammar}). Third, using a variant of Buszkowski's 
translation~\citep{BuszkoZML}, we establish undecidability even for the one-division fragments of the calculi in question (Section~\ref{S:Buszko}).

In comparison with a series of our papers on the Lambek calculus and non-commutative linear logic with
subexponential modalities~\citep{KanKuzSceLFCS,KanKuzSceJLC,KanKuzSceFG,KanKuzNigSce2018Dale,KanKuzNigSce2018IJCAR},
the principal difference of this paper 
is the presence of brackets and bracket modalities. Contraction rules used by Morrill in the bracketed calculi
essentially interact with brackets and become disfunctional in the bracket-free fragment. 
Undecidability results, on their turn,
rely on contraction. Thus, they should be proved for calculi with brackets independently from the bracket-free case.

On the other side, two papers on the bracketed Lambek calculus~\citep{KanKuzMorSceFSCD,MorKuzKanSce2018FG}
 do not deal with the subexponential modality ($!$), and feature effective
algorithms instead of undecidability results.

\section{The Multiplicative-Additive Lambek Calculus with Exponential/Relevant Modality}\label{S:MALC}

We start with more traditional calculi without brackets and bracket modalities, namely, the multiplicative-additive Lambek calculus
extended with a (sub)expo\-nen\-tial modality. 

Formulae of the calculi we are going to define in this section are constructed from a countable set $\Var$ of variables and 
the unit constant $\U$ using five binary connectives: $\cdot$ (product, or multiplicative conjunction), $\BS$ (left division), $\SL$ (right division),
$\wedge$ (additive conjunction), and $\vee$ (additive disjuncton), and one unary connective, ${!}$ (exponential).
Sequents are expressions of the form $\Pi \to A$, where $A$ is a formula, and $\Pi$ is a finite linearly ordered sequence of formulae.
Notice that these calculi are in general non-commutative, $\Pi$ is a sequence, not a set or multiset.

The first calculus we consider is $\LLs$,
the multiplicative-additive Lambek calculus extended
with a relevant subexponential modality ($\mathbf{r}$ stands for ``relevant,'' see below). The axioms of $\LLs$ are sequents of the
form $A \to A$ and $\Lambda \to \U$, and the rules of inference are as follows:
$$
\infer[\BS R]{\Pi \to A\BS B}{A, \Pi \to B}
\qquad
\infer[\BS L]{\Delta_1, \Pi, A \BS B, \Delta_2 \to C}{\Pi \to A & \Delta_1, B, \Delta_2 \to C}
$$
$$
\infer[\SL R]{\Pi \to B \SL A}{\Pi, A \to B}
\qquad
\infer[\SL L]{\Delta_1, B \SL A, \Pi, \Delta_2 \to C}{\Pi \to A & \Delta_1, B, \Delta_2 \to C}
$$
$$
\infer[\cdot R]{\Gamma, \Delta \to A \cdot B}{\Gamma \to A & \Delta \to B}
\qquad
\infer[\cdot L]{\Delta_1, A \cdot B, \Delta_2 \to C}{\Delta_1, A, B, \Delta_2 \to C}
\qquad
\infer[\U L]{\Delta_1, \U, \Delta_2 \to C}{\Delta_1, \Delta_2 \to C}
$$
$$
\infer[\wedge R]{\Pi \to A_1 \wedge A_2}{\Pi \to A_1 & \Pi \to A_2}
\qquad
\infer[\wedge L_i\mbox{, $i = 1,2$}]{\Delta_1, A_1 \wedge A_2, \Delta_2 \to C}{\Delta_1, A_i, \Delta_2 \to C}
$$
$$
\infer[\vee R_i\mbox{, $i = 1,2$}]{\Pi \to A_1 \vee A_2}{\Pi \to A_i}
\qquad
\infer[\vee L]{\Delta_1, A_1 \vee A_2, \Delta_2 \to C}{\Delta_1, A_1, \Delta_2 \to C & \Delta_1, A_2, \Delta_2 \to C}
$$
$$
\infer[{!} R]{{!}A_1, \dots, {!}A_n \to {!}B}{{!}A_1, \dots, {!}A_n \to B}
\qquad
\infer[{!} L]{\Delta_1, {!}A, \Delta_2 \to C}{\Delta_1, A, \Delta_2 \to C}
$$
$$
\infer[{!}P_1]{\Delta_1, \Phi, {!}A, \Delta_2 \to C}{\Delta_1, {!}A, \Phi, \Delta_2 \to C}\qquad
\infer[{!}P_2]{\Delta_1, {!}A, \Phi, \Delta_2 \to C}{\Delta_1, \Phi, {!}A, \Delta_2 \to C}
$$ $$
\infer[{!}C]{\Delta_1, {!}A, \Delta_2 \to C}{\Delta_1, {!}A, {!}A, \Delta_2 \to C}
$$
$$
\infer[\CUT]{\Delta_1, \Pi, \Delta_2 \to C}{\Pi \to A & \Delta_1, A, \Delta_2 \to C}
$$

Notice that, from the proof-theoretic point of view, it is better to use, instead of $(\CONTR)$, the following non-local contraction rules~\citep{KanKuzNigSce2018Dale}:
$$
\infer[{!}NC_1]{\Delta_1, \Phi, {!}A, \Delta_2 \to C}{\Delta_1, {!}A, \Phi, {!}A, \Delta_2 \to C}\qquad
\infer[{!}NC_2]{\Delta_1, {!}A, \Phi, \Delta_2 \to C}{\Delta_1, {!}A, \Phi, {!}A, \Delta_2 \to C}
$$
In the presence of ${!}P_{1,2}$, however, ${!}C$ has the same power as ${!}NC_{1,2}$.

The ${!}$ modality here is called ``relevant,'' since it allows contraction and permutation, but not
weakening, like in relevant logic.

The second system without brackets is  $\EL$, the multiplicative-addi\-ti\-ve Lambek calculus extended with a full-power exponential modality.
It is obtained from $\LLs$ by adding the lacking structural rule for ${!}$, namely weakening:
$$
\infer[{!}W]{\Gamma, {!}A, \Delta \to C}{\Gamma, \Delta \to C}
$$

Both $\LLs$ and $\EL$ are particular cases of $\SMALC$, the multi\-plicative-additive Lambek calculus extended with an arbitrary
 family of subexponentials $\Sigma$, considered by~\cite{KanKuzNigSce2018Dale}. In that paper it is shown that these calculi
 enjoy cut elimination and that the derivability problems for these calculi are undecidable.
  
 Cut elimination yields the subformula property (each formula occurring in the cut-free derivation  is a subformula of the
 goal sequent) and thus conservativity of elementary fragments. Namely, if one wants to derive only sequents that include
 formulae with a restricted set of connectives, it is sufficient just to restrict the set of rules of the calculus to this set
 of connectives. 
For convenience, we use a shorter notation, $\LLsM$ and $\ELM$, for the fragments without additive connectives
($\vee$ and $\wedge$) of $\LLs$ and $\EL$ respectively.

Let us formally define the notion of categorial grammar based on a non-commutative intuitionistic-style sequent calculus $\Lc$ without brackets, like the systems
$\LLs$ and $\EL$ defined above. 

\begin{definition}
An $\Lc$-grammar is a triple $\Gc = \langle \Sigma, \rhd, H \rangle$, where $\Sigma$ is a finite alphabet,
$H$ is a formula, and $\rhd$ is a finite binary correspondence between letters of $\Sigma$ and formulae (called {\em lexicon}).
A word $w = a_1 \ldots a_n$ over $\Sigma$ is accepted by $\Gc$ if there exist formulae $A_1, \ldots, A_n$
such that $a_i \rhd A_i$ ($i = 1, \ldots, n$) and the sequent $A_1, \ldots, A_n \to H$ is derivable in $\Lc$.
The language generated, or recognised, by $\Gc$ consists of all words recognised by $\Gc$.
\end{definition}

The system with weakening, $\EL$, has not so much to do with linguistic applications, but is interesting from the logical point of view.
In particular, we use it as an intermediate calculus in our undecidability proofs (Sections~\ref{S:undec} and~\ref{S:Buszko}). 

The system with a relevant modality, $\LLs$, supports analysis of many cases of extraction from dependent clauses, including parasitic extraction.
For example, {\sl ``the paper that John signed without reading''} is analysed as follows. First, we define the necessary fragment of the lexicon:

\vspace*{-10pt}
{\small \begin{align*}
\mbox{\sl the} & {} \rhd N \SL CN   && & \mbox{\sl John} & {} \rhd N \\
\mbox{\sl paper} & {} \rhd CN  &&  & \mbox{\sl signed, reading} & {} \rhd (N \BS S) \SL N \\
\mbox{\sl that} & {} \rhd (CN \BS CN) \SL (S \SL {!}N) && & \mbox{\sl without} & {} \rhd ((N \BS S) \BS (N \BS S)) \SL (N \BS S)
\end{align*}}
Here $N$ stands for ``noun phrase,'' $CN$ states for ``common noun'' (without an article), and $S$ stands for ``sentence.''
Next, we derive the sequent
\begin{multline*}
N \SL CN, CN, (CN \BS CN) \SL (S \SL {!}N), N, (N \BS S) \SL N, \\
((N \BS S) \BS (N \BS S)) \SL (N \BS S), (N \BS S) \SL N \to N
\end{multline*}
in $\LLs$, as shown on Figure~\ref{Fig:exampleA}.

\begin{figure}
\centerline{\rotatebox{90}{
$$
\infer[\SL L]{N \SL CN, CN, (CN \BS CN) \SL (S \SL {!}N), N, (N \BS S) \SL N, 
((N \BS S) \BS (N \BS S)) \SL (N \BS S), (N \BS S) \SL N \to N}
{\infer[\SL R]{N, (N \BS S) \SL N, ((N \BS S) \BS (N \BS S)) \SL (N \BS S), (N \BS S) \SL N \to S \SL {!}N}
{\infer[{!}NC_1]{N, (N \BS S) \SL N, ((N \BS S) \BS (N \BS S)) \SL (N \BS S), (N \BS S) \SL N, {!}N \to S}
{\infer[{!}L]{N, (N \BS S) \SL N, {!}N, ((N \BS S) \BS (N \BS S)) \SL (N \BS S), (N \BS S) \SL N, {!}N \to S}
{\infer[{!}L]{N, (N \BS S) \SL N, N, ((N \BS S) \BS (N \BS S)) \SL (N \BS S), (N \BS S) \SL N, {!}N \to S}
{\infer[\SL L]{N, (N \BS S) \SL N, N, ((N \BS S) \BS (N \BS S)) \SL (N \BS S), (N \BS S) \SL N, N \to S}
{N \to N & \infer[\SL L]{N, N \BS S, ((N \BS S) \BS (N \BS S)) \SL (N \BS S), (N \BS S) \SL N, N \to S}
{N \to N & \infer[\SL L]{N, N \BS S, ((N \BS S) \BS (N \BS S)) \SL (N \BS S), N \BS S \to S}
{N \BS S \to N \BS S & \infer[\BS L]{N, N \BS S, (N \BS S) \BS (N \BS S) \to S}
{N \BS S \to N \BS S & \infer[\BS L]{N, N \BS S \to S}{N \to N & S \to S}}}}}}}}}
&
\infer[\SL L]{N \SL CN, CN, CN \BS CN \to N}
{CN \to CN & \infer[\SL L]{N \SL CN, CN \to N}{CN \to CN & N \to N}}}
$$
}}
\caption{Derivation for {\sl ``the paper that John signed without reading''} in $\LLs$ (like \cite[Fig.~24]{Morrill2019}, but with
brackets and bracket modalities removed)}\label{Fig:exampleA}
\end{figure}

Without brackets, however, categorial grammars based on $\LLs$ suffer from overgeneration, parsing ungrammatical phrases like
{\sl *``the girl whom John loves Mary and Pete loves''} (see Introduction). In the next section, we introduce systems including
both brackets and a restricted subexponential, developed by Morrill in a series of papers.

\section{Morrill's Calculi with Brackets and Subexponential}\label{S:calculi}

In this section we describe extensions of the Lambek calculus, which include both brackets (and bracket modalities which control them) and
a subexponential, which interacts with brackets in an intricate way.

In his papers, Morrill (sometimes with his co-author Valent\'{\i}n) introduces different variants of his calculus---the difference is in the most
interesting rule, contraction. We consider two of Morrill's calculi, and denote these calculi by $\sysA$ and $\sysB$, by the year of first publication.
In this notation, ``(st)'' means the presence of stoups, the $\mathbf{b}$ on the right stands for ``brackets,'' and 
$\boldsymbol{!}_{\mathbf{b}}$ means that the subexponential ${!}$ interacts with the bracketing structure.

The $\sysA$ system, in its version without stoups, appears in~\citet{MorrillValentin}, and then in~\cite{MorrillLACompLing}
(\cite{MorrillPhilosophy} features a slightly different version of this system). The $\sysB$ system appears in Morrill's recent papers~\citep{Morrill2018JLM,Morrill2019};
however, essentially here Morrill returns to an older formulation of the bracket-aware contraction rule~\citep{MorrillBook,Morrill2014}.

Morrill's systems are quite involved, including up to 45 connectives. In this article we consider their simpler fragments, including multiplicative and additive Lambek
connectives ($\BS, \SL,\cdot, \U, \vee, \wedge$), brackets and bracket modalities ($\PMod$ and $\NMod$), and the subexponential ${!}$. Since all Morrill's systems
do not include cut as a rule, these fragments are conservative inside the bigger systems, and our undecidability results also work for the latter. (The question of {\em admissibility} of cut in Morrill's system is more subtle, and we discuss it later on.)

Before going forward, let us notice that full Morrill's systems also include Kleene star, axiomatised by means of an $\omega$-rule (Morrill calls it
``existential exponential'' and denotes by ``?''). In the presence of Kleene star, the Lambek calculus is known to be at least 
$\Pi_1^0$-hard~\citep{BuszkoPalka,Kuzn2017WoLLIC}, if the $\omega$-rule is used, and at least $\Sigma_1^0$-hard~\citep{KuznLICS19}, if the Kleene star
is axiomatised by means of induction axioms. In both cases, this means undecidability.
 Moreover, in the view of Kozen's results on complexity of Horn theories
of Kleene algebras~\citep{Kozen2002}, the complexity of a system with both Kleene star (with an $\omega$-rule) and a subexponential modality allowing contraction is likely
to rise up to $\Pi_1^1$-completeness. Morrill, however, emphasizes the fact that in formulae used in categorial grammars designed for real
languages the Kleene star never occurs with positive polarity. Thus, the $\omega$-rule
is never used, and the Kleene star does not incur problems with decidability.
Thus, the only possible source of undecidability is the specific contraction rule
for the subexponential. We consider  fragments of Morrill's systems with this
rule, which are sufficient to show undecidability.

The syntax and metasyntax of sequents in Morrill's systems (in particular, their fragments considered throughout this article) is more involved, if compared to the
calculi without brackets. First, in the antecedents we now have brackets which operate along with the structural comma
(a metasyntactic correspondent of the product connective), introducing partial non-associativity. Second, in order to
avoid superfluous usage of permutation rules for ${!}$-formulae and to facilitate proof search, in his systems Morrill groups the
${!}$-formulae to specifically designated commutative areas in the sequent. Using the terminology of ~\citet{GirardStoupMSCS,GirardStoupAPAL}, Morrill calles these areas {\em stoups.} Morrill's calculi, both technically and ideologically, are close to the sequent system by~\cite{hodas94ic}. In that system, antecedents are split into two zones, $\zeta;\Delta$, where $\zeta$ is the intuitionistic zone (formulae there are allowed to contract and weaken) and $\Delta$ is the linear one. In Morrill's terms, $\zeta$ is the stoup. Morrill's rules are more complicated, because of non-commutativity of the system in general, and also partial non-associativity introduced by brackets.
Introducing
the stoups, in fact, is the first step towards a focused proof system~\citep{Andreoli,Morrill2015Fiji,KanKuzNigSce2018IJCAR}.
Since permutations for ${!}$-formulae cannot penetrate brackets, each pair of brackets has its own stoup.

Let us define the syntax formally.
Formulae will be built from variables (primitive types) $p,q,\ldots$ and the multiplicative unit constant $\One$ using 
three binary operations: $\BS$ (left division), $\SL$ (right division), $\mconj$ (product), and three unary operations:
$\PMod$ and $\NMod$ (bracket modalities) and ${!}$ (subexponential). Sequents (in Morrill's terminology, {\em h-sequents)} are expressions of the form
$\Xi \Rightarrow A$, where $A$ is a formula and $\Xi$ is a complex metasyntactic structure which we call {\em meta-formula} (Morrill calls them {\em zones).}
Meta-formulae are built from formulae using comma and brackets; also formulae which are intended to be marked by the subexponential ${!}$, which
allows permutation, are placed into stoups. 
Following~\citet{Morrill2019}, we define the notion of meta-formula along with two auxiliary notions, stoup and {\em tree term}, simultaneously.
\begin{itemize}
\item A stoup is a multiset of formulae: $\zeta = \{ A_1, \ldots, A_n \}$. A stoup could be empty, the empty stoup is denoted by $\varnothing$.
\item A tree term is either a formula or a bracketed expression of the form $[\Xi]$, where $\Xi$ is a meta-formula.
\item A meta-formula is an expression of the form $\zeta; \Gamma$, where $\zeta$ is a stoup and $\Gamma$ is a linearly
ordered sequence of tree terms. Here $\Gamma$ could also be empty; the empty sequence is denoted by $\Lambda$.
\end{itemize}
We use comma both for concatenation of tree term sequences and for multiset union of stoups (Morrill uses $\uplus$ for the latter).
Moreover, for adding one formula into a stoup we write $\zeta, A$ instead of $\zeta, \{A\}$. Empty stoups are omitted: instead of
$\varnothing; \Gamma$ we write just $\Gamma$.

Let us first formulate the rules which do not operate ${!}$, since these rules are the same in all Morrill's systems.

$$
\infer[\mathrm{id}]{A \yd A}{}
$$

$$
\infer[{\SL} L]{\Xi (\zeta_1, \zeta_2 ; \Delta_1, C \SL B, \Gamma, \Delta_2) \yd D}
{\zeta_1; \Gamma \yd B & \Xi(\zeta_2 ; \Delta_1, C, \Delta_2 ) \yd D}
\qquad
\infer[{\SL} R]{\zeta; \Gamma \yd C \SL B}{\zeta; \Gamma, B \yd C}
$$

$$
\infer[{\BS} L]{\Xi (\zeta_1, \zeta_2 ; \Delta_1, \Gamma, A \BS C, \Delta_2) \yd D}
{\zeta_1; \Gamma \yd A & \Xi(\zeta_2 ; \Delta_1, C, \Delta_2 ) \yd D}
\qquad
\infer[{\BS} R]{\zeta; \Gamma \yd A \BS C}{\zeta; A, \Gamma \yd C}
$$

$$
\infer[{\mconj} L]{\Xi (\zeta; \Delta_1, A \mconj B, \Delta_2) \yd D}
{\Xi (\zeta; \Delta_1, A, B, \Delta_2) \yd D}
\qquad
\infer[{\mconj} R]{\zeta_1, \zeta_2 ; \Delta, \Gamma \yd A \mconj B}
{\zeta_1; \Delta \yd A & \zeta_2; \Gamma \yd B}
$$

$$
\infer[\vee R_i\ i=1,2]{\Xi \yd A_1 \vee A_2}{\Xi \yd A_i}
\qquad
\infer[{\One} L]{\Xi(\zeta; \Delta_1, \One, \Delta_2) \yd A}{\Xi(\zeta;\Delta_1,\Delta_2) \yd A}
$$

$$
\infer[\vee L]{\Xi(\zeta; \Delta_1, A_1 \vee A_2, \Delta_2) \yd C}
{\Xi(\zeta; \Delta_1, A_1, \Delta_2) \yd C & \Xi(\zeta; \Delta_1, A_2, \Delta_2) \yd C}
\qquad
\infer[{\One} R]{\Lambda \yd \One}{}
$$

$$
\infer[\wedge L_j\ j=1,2]{\Xi(\zeta; \Delta_1, A_1 \wedge A_2, \Delta_2) \yd C}
{\Xi(\zeta; \Delta_1, A_j, \Delta_2) \yd C}
\qquad
\infer[\wedge R]{\Xi \yd A_1 \wedge A_2}{\Xi \yd A_1 & \Xi \wedge A_2}
$$

$$
\infer[{\NMod} L]{\Xi(\zeta; \Delta_1, [\NMod A], \Delta_2) \yd B}
{\Xi (\zeta; \Delta_1, A, \Delta_2) \yd B}
\qquad
\infer[{\NMod} R]{\Xi \yd \NMod A}{[ \Xi ] \yd A}
$$

$$
\infer[{\PMod} L]{\Xi(\zeta; \Delta_1, \PMod A, \Delta_2) \yd B}
{\Xi(\zeta; \Delta_1, [ A], \Delta_2) \yd B}
\qquad
\infer[{\PMod} R]{ [\Xi] \yd \PMod A}{\Xi \yd A}
$$

The two calculi, $\sysA$ and $\sysB$, also share two rules for ${!}$:
$$
\infer[{!} L]{\Xi(\zeta; \Gamma_1, {!}A, \Gamma_2) \yd B}{\Xi(\zeta, A ; \Gamma_1, \Gamma_2) \yd B}
\qquad
\infer[{!} P]{\Xi(\zeta, A; \Gamma_1, \Gamma_2) \yd B}{\Xi (\zeta; \Gamma_1, A, \Gamma_2) \yd B}
$$

However, the ${!}R$ rule and, most importantly, the contraction rule ${!}C$ are different.
In the ``older'' system $\sysA$ they are formulated as follows:
$$
\infer[{!}R]{\zeta; \Lambda \yd {!}B}{\zeta; \Lambda \yd B}
\qquad
\infer[{!}C,\ \zeta_2 \ne \varnothing]{\Xi(\zeta_1, \zeta_2; \Gamma_1, \Gamma_2, \Gamma_3) \yd B}
{\Xi(\zeta_1, \zeta_2; \Gamma_1, [\zeta_2; \Gamma_2], \Gamma_3) \yd B}
$$
The ``newer'' system $\sysB$ uses the following formulation of ${!}R$ and ${!}C$:
$$
\infer[{!} R]{{!}A \yd {!}B}{{!}A \yd B}
\qquad
\infer[{!} C]{\Xi(\zeta, A; \Gamma_1, [[\Gamma_2]], \Gamma_3) \yd B}
{\Xi(\zeta, A; \Gamma_1, [A; \Gamma_2], \Gamma_3) \yd B}
$$

As noticed above, in the absence of cut we can easily formulate fragments of $\sysA$ and $\sysB$ without additive connectives: one just removes the corresponding
rules ($\vee L$, $\vee R_{1,2}$, $\wedge L_{1,2}$, $\wedge R$). In the notations, we just replace ``$\mathbf{MALC}$'' with ``$\mathbf{L}$'':
$\LsysA$, $\LsysB$.  In the following sections we use the same naming convention: if a calculus' name includes ``$\mathbf{MALC}$,'' then
replacing it with ``$\mathbf{L}$'' gives a name for the fragment of this calculus without additive connectives.

For calculi with brackets, defining recognition of words in categorial grammars is trickier. One can keep the definition
from Section~\ref{S:MALC} and say that $w = a_1 \dots a_n$ is accepted
by the grammar if $A_1, \dots, A_n \to H$ is derivable, for some $A_i$ such that $a_i \rhd A_i$ ($i = 1, \dots, n$). Notice that this
sequent does not include brackets, but may include bracket modalities, $\PMod$ and $\NMod$. Thus, brackets could appear inside the derivation.
This notion of recognition is called {\em s-recognition}~\citep{Jaeger2003}. 

Linguistic applications, however, suggest another notion of recognition for Lambek grammars with brackets, called {\em t-recognition.}
A word $w = a_1 \dots a_n$ is t-accepted by an grammar $\Gc$ if the sequent $\Pi \to H$ is derivable for some $\Pi$ such that if one
removes all brackets (but not bracket modalities!) from $\Pi$, it yields $A_1, \dots, A_n$, where $a_i \rhd A_i$ ($i = 1,\dots,n$). In other words,
a word is accepted if there corresponding sequent is derivable {\em for some bracketing $\Pi$.}

In the implementation of Morrill's bracketed calculi in the CatLog parser, the bracket structure on $A_1, \ldots, A_n$ is requested from the user
as part of input data~\citep{CatLog3tech}.  There is an ongoing project of implementing automatic guessing of the correct bracket structure (so-called {\em bracket induction});
at the present time, there exists such an algorithm for the fragment with only multiplicative connectives and bracket modalities, without 
subexponential~\citep{MorKuzKanSce2018FG}.

As an example, we analyse the phrase {\sl ``the paper that John signed without reading''} using $\sysB$. Our analysis is a simplification
of the one of~\citet{Morrill2019}. In comparison with the analysis in Section~\ref{S:MALC} (Figure~\ref{Fig:exampleA}), here we take care
of the bracketed domains, which cannot be penetrated by associativity of product or permutations of ${!}$-formulae. Also notice that the
contraction rule here implements parasitic extension in the following sense: applying contraction to ${!}N$ (actually, to $N$ located in
the stoup) instantiates a secondary (parasitic) copy of ${!}N$ into an island. In order to prevent reusage of islands for parasitic extraction,
the island transforms from a strong (double-bracketed) to a weak (single-bracketed) one. The lexicon now is as follows (if compared to 
the one in Section~\ref{S:MALC}, the types here are augmented with bracket modalities):

{\small
\begin{align*}
\mbox{\sl the} &\triangleright N \SL CN &&				&			\mbox{\sl likes, signed} &\triangleright (\PMod N \BS S) \SL N \\
\mbox{\sl man, paper} &\triangleright CN && 				&			\mbox{\sl without} &\triangleright (\NMod ((\PMod N \BS S) \BS (\PMod N \BS S))) \SL (\PMod N \BS S) \\
\mbox{\sl reading} &\triangleright (\PMod N \BS S) \SL N && &  \mbox{\sl who, that} &\triangleright (\NMod\NMod (CN \BS CN)) \SL (S \SL {!}N)\\
\mbox{\sl John} &\triangleright \PMod N
\end{align*}}
Before parsing, we have to impose the right bracket structure on our phrase. This is done as follows:
{\sl ``the paper {\rm [[}that {\rm [}John{\rm]} signed {\rm [[}without reading{\rm]] ]]}.''}
Indeed, in Morrill's CatLog categorial grammar the subject group and the {\sl without}-clause form islands, and the {\sl that}-clause forms a strong island, embraced by
double brackets. Moreover, we also have to double-bracket our without-clause
(make it a ``strong island''), since it will be used for parasitic extraction.

Now the sequent we have to derive in $\sysB$ is as follows:
\begin{multline*}
N \SL CN, [[\, (\NMod\NMod (CN \BS CN)) \SL (S \SL {!}N), [N], (\PMod N \BS S) \SL N, \\
[[\, (\NMod ((\PMod N \BS S) \BS (\PMod N \BS S))) \SL (\PMod N \BS S), (\PMod N \BS S) \SL N\, ]]\;]] \to N
\end{multline*}
The derivation is presented on Figure~\ref{Fig:Johnsigned}.

\begin{figure}
\centerline{\rotatebox{90}{$
\infer[\SL L]{N \SL CN,  CN, [[\, (\NMod\NMod (CN \BS CN)) \SL (S \SL {!}N), [N], (\PMod N \BS S) \SL N, 
[[\,(\NMod ((\PMod N \BS S) \BS (\PMod N \BS S))) \SL (\PMod N \BS S), (\PMod N \BS S) \SL N\,]]\;]] \yd N}
{\infer[\SL L]{CN, [[\, (\NMod\NMod (CN \BS CN)) \SL (S \SL {!}N), [N], (\PMod N \BS S) \SL N, 
[[\, (\NMod ((\PMod N \BS S) \BS (\PMod N \BS S))) \SL (\PMod N \BS S), (\PMod N \BS S) \SL N\,]]\:]] \yd CN}
{\infer[\SL R]{[N], (\PMod N \BS S) \SL N, 
[[\,(\NMod ((\PMod N \BS S) \BS (\PMod N \BS S))) \SL (\PMod N \BS S), (\PMod N \BS S) \SL N\,]] \yd S \SL {!}N}
{\infer[{!} L]{[N], (\PMod N \BS S) \SL N, 
[[\, (\NMod ((\PMod N \BS S) \BS (\PMod N \BS S))) \SL (\PMod N \BS S), (\PMod N \BS S) \SL N\,]], {!}N \yd S}
{\infer[{!} C]{N; [N], (\PMod N \BS S) \SL N, 
[[\, (\NMod ((\PMod N \BS S) \BS (\PMod N \BS S))) \SL (\PMod N \BS S), (\PMod N \BS S) \SL N\,]] \yd S}
{\infer[{!} P]{N; [N], (\PMod N \BS S) \SL N, 
[\,N;  (\NMod ((\PMod N \BS S) \BS (\PMod N \BS S))) \SL (\PMod N \BS S), (\PMod N \BS S) \SL N\,] \yd S}
{\infer[{!} P]{[N], (\PMod N \BS S) \SL N, N,
[\,N;  (\NMod ((\PMod N \BS S) \BS (\PMod N \BS S))) \SL (\PMod N \BS S), (\PMod N \BS S) \SL N\,] \yd S}
{\infer[\SL L]{[N], (\PMod N \BS S) \SL N, N,
[\, (\NMod ((\PMod N \BS S) \BS (\PMod N \BS S))) \SL (\PMod N \BS S), (\PMod N \BS S) \SL N, N\,] \yd S}
{N \yd N & 
\infer[\SL L]{[N], (\PMod N \BS S) \SL N, N, [\,(\NMod ((\PMod N \BS S) \BS (\PMod N \BS S))) \SL (\PMod N \BS S), \PMod N \BS S\,] \yd S}
{\PMod N \BS S \yd \PMod N \BS S & 
\infer[\NMod L]{[N], (\PMod N \BS S) \SL N, N, [\,\NMod ((\PMod N \BS S) \BS (\PMod N \BS S))\,] \yd S}
{\infer[\BS L]{[N], (\PMod N \BS S) \SL N, N, (\PMod N \BS S) \BS (\PMod N \BS S)) \yd S}
{N \yd N & \infer[\BS L]{[N], \PMod N \BS S, (\PMod N \BS S) \BS (\PMod N \BS S)) \yd S}
{\PMod N \BS S \yd \PMod N \BS S & \infer[\BS L]{[N], \PMod N \BS S \yd S}
{\infer[\PMod R]{[N] \yd \PMod N}{N \yd N} & S \yd S}}}}}}}}}}} &
\infer[\NMod L]{CN, [[\, \NMod\NMod (CN \BS CN)\, ]] \yd CN}
{\infer[\NMod L]{CN, [\,\NMod (CN \BS CN)\,] \yd CN}
{\infer[\BS L]{CN, CN \BS CN \yd CN}{ CN \yd CN &  CN \yd CN}}}}
& N \yd N}
$
}}
\caption{Derivation for {\sl ``the paper that John signed without reading''}  in $\sysB$ (cf. \citet[Fig.~24]{Morrill2019})}\label{Fig:Johnsigned}
\end{figure}

\section{Issues with Cut Elimination}\label{S:issues}

Cut elimination is one of the standard logical properties which is expected from a reasonable Gentzen-style sequent calculus.
Since in the systems discussed in this article cut is not included as an official rule, the question of cut elimination 
appears as the question of the {\em admissibility} of cut. 
From the linguistic perspective, cut supports the principle of compositionality: once we have proved that a phrase has
syntactic type, say, $NP$, we can use it at any place where a noun phrase is allowed. 

\citet{Morrill2019} mentions a semantic approach to prove admissibility of cut in $\sysB$ as an ongoing work by O.~Valent\'{\i}n. In this
paper, we wish to pursue the more traditional syntactic approach for cut elimination, both in $\sysA$ and $\sysB$.

Unfortunately, Morrill's systems, as formulated above (Section~\ref{S:calculi}), fail to enjoy cut elimination (cut admissibility).
For $\sysB$, the counter-example is ${!}p, q \to q \cdot {!}p$. This sequent expresses the natural property that ${!}$-formulae
commute with arbitrary formulae, and it is derivable using cut:
$$
\infer[\CUT]{{!}p, q \to q \cdot {!}p}
{\infer[!R]{{!}p \to {!}{!}p}{{!}p \to {!}p}
& \infer[!L]{{!}{!}p, q \to q \cdot {!}p}
{\infer[!P]{{!}p; q \to q \cdot {!}p}
{\infer[{\cdot} R]{q, {!}p \to q \cdot {!}p}{q \to q & {!}p \to {!}p}}}}
$$
However, no cut-free derivation is available. Indeed, the lowermost rule of such a derivation should be either ${\cdot}R$ or $!L$. 
The former is impossible, since neither $\Lambda \to q$, nor ${!}p \to q$, nor ${!}p,q \to q$ is derivable. In the latter case,
we get $p; q \to q \cdot {!}p$ and again have two possibilities: ${\cdot} R$ or $!P$. For ${\cdot R}$, the only possible way
of splitting could be $q \to q$ and $p; \Lambda \to {!}p$. The latter, however, is counter-intuitively not derivable (though $p$
in stoup should mean $!p$): one cannot immediately apply $!R$, and applying $!P$ gives $p \to {!}p$. Applying $!P$ would give
either $p,q \to q \cdot {!}p$ or $q,p \to q \cdot {!}p$, none of which is derivable.
Notice that the proof search here is finite, since the contraction rule could not be used in the absence of brackets.

Thus, we have to modify $\sysB$ in order to restore the cut elimination property.
We do this by replacing ${!}R$ and ${!}C$ with the following rules: 
$$
\infer[!R']{A; \Lambda \to {!}B}{A; \Lambda \to B}
\qquad
\infer[!C']
{\Xi(\zeta, A; \Gamma_1, [[ \zeta'; \Gamma_2 ]], \Gamma_3) \to B}
{\Xi(\zeta, A; \Gamma_1, [ \zeta', A; \Gamma_2 ], \Gamma_3) \to B}
$$
Notice that ${!}R'$ corresponds to the $!R$ rule of Morrill's $\sysA$. The only difference is that here
the stoup should include exactly one formula.

We denote the modified calculi by $\sysBa$.

In what follows, we shall show (Theorem~\ref{Th:cutelim}) that  $\sysBa$ admits the cut rule in the following stoup-aware form:
$$
\infer[\CUT]{\Xi(\xi,\zeta; \Gamma_1, \Pi, \Gamma_2) \to C}
{\xi;\Pi \to A & \Xi(\zeta; \Gamma_1, A, \Gamma_2) \to C}
$$
Using cut and the left rules for ${!}$, one can derive the old ${!}R$ rule from the new ${!}R'$ one:
$$
\infer[{!}L]{{!}A \to {!}B}
{\infer[{!}R']{A; \Lambda \to {!}B}{\infer[\CUT]{A; \Lambda \to B}
{\infer[{!}R']{A; \Lambda \to {!}A}{\infer[{!}P]{A; \Lambda \to A}{A \to A}} & {!}A \to B}}}
$$
As for ${!}C'$, the old rule ${!}C$ is just its particular case, for $\zeta' = \varnothing$. Thus,
$\sysBa$ is an extension of $\sysB$.

For $\sysA$, problems come from the non-emptiness restriction imposed on the contraction rule.
The ${!}C$ rule in $\sysA$ is formulated in the ``multi-contraction'' form, allowing to contract
several formulae in the stoup at once. However, it should contract {\em at least one} formula.
This constraint can be easily violated by cut with $\Lambda \to {!}\U$ (which is derivable in $\sysA$).
In systems without brackets this would not be an issue, since in such systems contraction of zero
formulae does nothing. In $\sysA$, however, ${!}C$ operates brackets, so such a ``zero-contraction'' would violate
bracket discipline. 

The concrete counter-example is $q \to \PMod q$. This sequent clearly has no cut-free derivation, but can be derived using cut:
$$
\infer[\CUT]{q \to \PMod q}{\infer[{!}R]{\Lambda \to {!}\U}{\Lambda \to \U} & 
\infer[{!}L]{{!}\U, q \to \PMod q}{\infer[{!}C]{\U; q \to \PMod q}
{\infer[{!}P]{\U; [\U; q] \to \PMod q}
{\infer[{!}P]{\U, [\U; q] \to \PMod q}
{\infer[\U L]{\U, [\U, q] \to \PMod q}
{\infer[\U L]{[\U,q] \to \PMod q}
{\infer[\PMod R]{[q] \to \PMod q}{q\to q}}}}}}}}
$$

We modify $\sysA$ in the following way, yielding the system $\sysAa$:
$$
\infer[{!}R',\ \zeta \ne \varnothing]{\zeta; \Lambda \to {!}B}{\zeta; \Lambda \to B}
\qquad
\infer[{!}C',\ \zeta_2 \ne \varnothing]{\Xi(\zeta_1, \zeta_2, \zeta'; \Gamma_1, \Gamma_2, \Gamma_3) \to C}
{\Xi(\zeta_1, \zeta_2; \Gamma_1, [\zeta', \zeta_2; \Gamma_2], \Gamma_3) \to C}
$$
Theorem \ref{Th:cutelimA} establishes cut admissiblity in $\sysAa$.

\section{Lambek's Restriction}\label{S:restriction}

The original Lambek calculus~\citep{Lambek58} has an important difference from the systems discussed above, 
namely {\em Lambek's non-emptiness restriction.} 
Let us start with a linguistic example~\citep[Sect. 2.5]{MootRetore}. In the calculi defined above,
one can derive $(N \SL N) \SL (N \SL N), N \to N$. This sequent validates ``very book'' as an object of type $N$
(common noun), which is incorrect. Indeed, the type $(N \SL N) \SL (N \SL N)$ for ``very'' is a left modifier for
adjective, cf. ``very interesting book,'' analyzed as $(N \SL N) \SL (N \SL N), N \SL N, N \to N$.

This example motivates the following constraint: {\em  left-hand sides of all sequents 
 are required to be non-empty.} This constraint existed in the original Lambek calculus~\citep{Lambek58}. It is quite strange
 from the logical point of view, but is natural from the linguistic side and also in the view of algebraic interpretations
 (considering residuated semigroups instead of monoids).

In the presence of a full-power exponential modality, however, imposing Lambek's restriction is quite a subtle matter~\citep{KanKuzSceLFCS,KanKuzSceJLC}.
Actually, there is no way of doing it without losing at least one of the desired properties of a good logical system---cut elimination and substitution.
Similar issues arise with reconciling Lambek's restriction with the relevant modality.

The subexponential modalities used by Morrill, however, are not that powerful, and their behaviour is constrained by brackets. This makes it possible
to impose Lambek's restriction in a linguistically consistent manner. 
In this section, we present $\sysBar$, a version of $\sysBa$ with Lambek's restriction imposed.

Before going into the formalism, let us consider one more linguistic example~\citep{Morrill2018JLM}.
This example features an incorrect noun phrase, {\sl *``man who likes.''} The dependent clause here
is analysed with two gaps, {\sl *``man who {\rm []} likes {\rm []}.''} The intended semantics (and the correct
version of the phrase) is {\sl ``man who likes himself,''}  that is, both gaps should be filled with
the same $N$, using the parasitic extraction mechanism. The lexicon here is the same as in the example in
Section~\ref{S:calculi}.

Since the dependent clause forms a strong (double-bracketed) island, the brackets are imposed as follows:
{\sl ``man {\rm [[}who likes{\rm ]]}.''} Next, we recall that the subject should form a weak (single-bracketed) island, and
in these brackets can be generated in $\sysA$ by the contraction rule. This allows $\sysA$ to parse  (incorrectly) {\sl ``likes''} as
a dependent clause with two gaps, a host one for the object and a parasitic one for the subject:
$$
\infer[\SL R]{(\PMod N \BS S) \SL N \to S \SL {!}N}
{\infer[{!}L]{(\PMod N \BS S) \SL N, {!}N \to S}
{\infer[{!}C]{N; (\PMod N \BS S) \SL N \to S}
{\infer[{!}P]{N; [N; \Lambda], (\PMod N \BS S) \SL N \to S}
{\infer[{!}P]{[N; \Lambda], (\PMod N \BS S) \SL N, N \to S}
{\infer[\SL L]{[N], (\PMod N \BS S) \SL N, N \to S}
{N \to N & \infer[\BS L]{[N], \PMod N \BS S \to S}
{\infer[\PMod R]{[N] \to \PMod N}{N \to N} & S \to S}}}}}}}
$$
The complete derivation for {\sl *``man who likes''} as a common noun group ($CN$) in $\sysA$ is given on Figure~\ref{Fig:manwholikes}

\begin{figure}
\centerline{ 
{
$$
\infer[\SL L]
{CN, [[ (\NMod\NMod (CN \BS CN)) \SL (S \SL {!}N), (\PMod N \BS S) \SL N ]]  \to CN}
{\infer[\SL R]{(\PMod N \BS S) \SL N \to S \SL {!}N}
{\infer[{!}L]{(\PMod N \BS S) \SL N, {!}N \to S}
{\infer[{!}C]{N; (\PMod N \BS S) \SL N \to S}
{\infer[{!}P]{N; [N; \Lambda], (\PMod N \BS S) \SL N \to S}
{\infer[{!}P]{[N; \Lambda], (\PMod N \BS S) \SL N, N \to S}
{\infer[\SL L]{[N], (\PMod N \BS S) \SL N, N \to S}
{N \to N & \infer[\BS L]{[N], \PMod N \BS S \to S}
{\infer[\PMod R]{[N] \to \PMod N}{N \to N} & S \to S}}}}}}}
&
\infer=[\NMod L]{CN, [[ \NMod\NMod (CN \BS CN) ]] \to CN}
{\infer[\BS L]{CN, CN \BS CN \to CN}{CN \to CN & CN \to CN}}
}
$$
}
}
\caption{Derivation for {\sl *``man {\rm [[}who likes{\rm ]]}''} in $\sysA$ (cf.~\cite{Morrill2018JLM})}
\label{Fig:manwholikes}
\end{figure}

The problem here is the empty island (subject of the dependent clause) generated by the ${!}C$ rule.
This issue was one of the motivations for~\citet{Morrill2018JLM} to introduce the new system~$\sysB$, which features
another version of ${!}C$.

With this new version, the island for parasitic extraction should be given in the bracketing of the goal sequent.
Moreover, it should be declared as a strong (double-bracketed) island, and then the ${!}C$ rule will transform it
into a weak one. The erroneous phrase {\sl *``man who likes,''} however, can still be parsed by $\sysB$, but requires
the empty subject island to be explicitly introduced in the bracketing, see Figure~\ref{Fig:manwholikesB}.

\begin{figure}
\centerline{{
$$
\infer[\SL L]
{CN, [[ (\NMod\NMod (CN \BS CN)) \SL (S \SL {!}N), [[ \Lambda ]], (\PMod N \BS S) \SL N ]]  \to CN}
{\infer[\SL R]{[[\Lambda]], (\PMod N \BS S) \SL N \to S \SL {!}N}
{\infer[{!}L]{[[\Lambda]],(\PMod N \BS S) \SL N, {!}N \to S}
{\infer[{!}C]{N; [[\Lambda]], (\PMod N \BS S) \SL N \to S}
{\infer[{!}P]{N; [N; \Lambda], (\PMod N \BS S) \SL N \to S}
{\infer[{!}P]{[N; \Lambda], (\PMod N \BS S) \SL N, N \to S}
{\infer[\SL L]{[N], (\PMod N \BS S) \SL N, N \to S}
{N \to N & \infer[\BS L]{[N], \PMod N \BS S \to S}
{\infer[\PMod R]{[N] \to \PMod N}{N \to N} & S \to S}}}}}}}
&
\infer=[\NMod L]{CN, [[ \NMod\NMod (CN \BS CN) ]] \to CN}
{\infer[\BS L]{CN, CN \BS CN \to CN}{CN \to CN & CN \to CN}}
}
$$
}
}
\caption{Derivation of {\sl *``man {\rm [[}who {\rm [[\ ]]} likes{\rm ]]}''} in $\sysB$ (notice the empty
subject island)}\label{Fig:manwholikesB}
\end{figure}

An easy way of making phrases like {\sl *``man who likes''} invalid in grammars based on $\sysB$ is to forbid the user (or an automated
bracket-inducing system) to put empty bracket domains on the original sentence. This is essentially the idea which motivates
the usage of $\sysB$ in favour of $\sysA$: in $\sysA$, the brackets embracing an empty island appeared only inside the derivation, while
in the newer system $\sysB$ they should be provided as an input, which could be disallowed externally.

 A more logically consistent approach, however, requires imposing non-emptiness restriction
systematically for all sequents in derivations. The restriction is formulated as follows: 
\begin{center}
{\em every meta-formula, both the whole antecedent and each bracketed domain, should be non-empty.}
\end{center}
Non-emptiness of a meta-formula means that it is not equal to the empty one, $\varnothing;\Lambda$. In other words, it should
{\em either} include a non-empty sequence of formulae, {\em or} have a non-empty stoup.

The first thing one has to do in order to maintain Lambek's restriction is to remove the unit constant $\U$. The unit essentially
means ``empty,'' and there is no consistent way of reconciling it with Lambek's restriction. Indeed, having the unit, we can put it
into any meta-formula, thus making it formally non-empty. (Unfortunately, it seems that Morrill needs the unit for handling discontinuity,
that is why he does not impose Lambek's restriction on his systems.)

Most of the $\sysBa$ rules keep this restriction, {\em i.e.,} if in the premises all meta-formulae are non-empty, then the same holds for the conclusion.
Only three rules need specifically imposed restrictions:
\begin{itemize}
\item for $\BS R$ and $\SL R$, we require that either $\Gamma \ne \Lambda$ or $\zeta \ne \varnothing$ (this is the original Lambek's restriction);
\item for the contraction rule, ${!}C'$, we require that either $\Gamma_2 \ne \Lambda$ or $\zeta' \ne \varnothing$.
\end{itemize}
The latter constraint exactly captures the idea that parasitic gapping into an empty bracketed island is ungrammatical (cf. the ``man that likes'' example above).

We denote the version of $\sysBa$ with Lambek's restriction by $\sysBar$.

\section{Cut Elimination in Modified Systems}\label{S:cutelim}

In this section we prove that the cut rule in the following form
$$
\infer[\CUT]
{\Xi(\xi, \zeta; \Gamma_1, \Pi, \Gamma_2) \to C}
{\xi; \Pi \to A & \Xi(\zeta; \Gamma_1, A, \Gamma_2) \to C}
$$
is admissible in the following calculi: $\sysBa$, $\sysBar$, and $\sysAa$. We show this by a 
single inductive argument for $\sysBa$ and $\sysBar$, and then make necessary changes for $\sysAa$.

\begin{theorem}\label{Th:cutelim}
Let sequents $\xi; \Pi \to A$ and $\Xi(\zeta; \Gamma_1, A, \Gamma_2) \to C$ be derivable
in $\sysBa$ or $\sysBar$.
Then 
$\Xi(\xi,\zeta; \Gamma_1, \Pi, \Gamma_2) \to C$ is also derivable
in $\sysBa$ or, respectively, $\sysBar$.
\end{theorem}

The proof of cut elimination traditionally goes by nested induction: on the complexity of the formula being cut,
and on the depth of the cut, that is, the number of rules applied in the derivation above the cut.

For the original Lambek calculus, cut elimination was shown by~\citet{Lambek58}. \citet{MoortgatMultimodal}
extended Lambek's proof to the Lambek calculus with brackets. 
The presence of ${!}$ and
stoups, however, makes cut elimination more involved. Namely, the principal case for ${!}$ moves the active formula
being cut to the stoup:
$$
\infer[\CUT]
{\Xi(\zeta, B; \Gamma_1, \Gamma_2) \to C}
{\infer[{!}R']{B; \Lambda \to {!}A}{B; \Lambda \to A} & 
\infer[{!}L]{\Xi(\zeta; \Gamma_1, {!}A, \Gamma_2) \to C}
{\Xi(\zeta, A; \Gamma_1, \Gamma_2) \to C}}
$$
Propagating the cut upwards in this situation would require a specific version of cut for formulae inside the stoup,
and eliminate it together with the usual cut rule by simultaneous induction.
Contraction, however, raises yet another issue with propagating cut. Namely, if we contract the formula $A$ being cut,
then after propagation we get two cut applications, one under another. For the lower cut, we fail to maintain the decrease
of induction parameters, see~\citet{KanKuzNigSce2018Dale}. 

The standard strategy, going back to~\citet{Gentzen} and applied to linear logic with exponentials by~\citet{Girard} and~\citet{LMSS},
replaces the cut rule with a more general rule called mix. Mix is a combination of cut and contractions, and this more general rule 
is then eliminated by a straightforward inductive argument. In the presence of brackets and stoups, however, formulating mix becomes
an extremely tedious job. In the view of that, we follow another strategy, ``deep cut elimination'' by~\citet{BraunerBRICS,dePaiva};
see also~\cite{Brauner2000IGPL,EadesPaiva}.

\begin{proof}
Let $\xi; \Pi \to A$ and $\Xi(\zeta; \Gamma_1, A, \Gamma_2 \to C)$
 have cut-free derivations $\DerL$ and $\DerR$ respectively.
We proceed by nested induction on two parameters: $\kappa$, the complexity of the formula $A$ being cut;
$\sigma$, the total number of rule applications in the derivations of $\DerL$ and $\DerR$.
In each case either $\kappa$ gets reduces, or $\sigma$ gets reduced with
the same $\kappa$.

We consider the lowermost rules of $\DerL$ and $\DerR$.

We call ${!}P$ and ${!}C$ {\em structural} rules; all other rules (excluding cut, which is not allowed
in our derivations) are {\em logical} ones. Being the lowermost rule of $\DerL$ or $\DerR$, a logical rule is called
{\em principal,} if it introduces the formula $A$ being cut. 
The axiom $\Lambda\to\U$ in this proof is considered a principal rule (with no premises) introducing $\U$.
Structural rules are never principal, since they operate only inside the stoup, while the formula $A$
being cut is not in the stoup.

First we list all possible cases, with short comments, and then accurately consider each of them:

\begin{enumerate}
\item The lowermost rule in $\DerL$ is $!R$ and the lowermost rule in $\DerR$ is $!L$ ({\em i.e.,} the principal case with ${!}$). 
This is actually the most interesting case, in which deep cut elimination
differs from traditional cut elimination schemes. In this case, we are going to perform a non-local transformation
of the $\DerR$ tree, as shown below.
\item Both lowermost rules of $\DerL$ and $\DerR$ are principal, and $A$ is not of the form ${!}A'$ (if it is, we are in Case~1).
This is the standard principal case for cut elimination in the Lambek calculus: the tricky part with ${!}$ is considered in
Case~1, not here.
\item The lowermost rule in $\DerL$ is a non-principal one. In this case we propagate cut to the left.
\item The lowermost rule in $\DerR$ is a non-principal one. Propagate cut to the right.
\item One of the premises of cut is an axiom of the form $A \to A$. Cut disappears. 
\end{enumerate}

{\bf Case 1 (deep: principal for ${!}$):} the lowermost rule in $\DerL$ is 
${!}R$ and the lowermost rule in $\DerR$ is ${!}L$. Cut is applied as follows:
$$
\infer[\CUT]{\Xi(\zeta,B; \Gamma', \Gamma'') \to C}
{\infer[!R]{B; \Lambda \to {!}A}{B; \Lambda \to A} & 
\infer[!L]{\Xi(\zeta; \Gamma', {!}A, \Gamma'') \to C}
{\Xi(\zeta, A; \Gamma', \Gamma'') \to C)}}
$$

Let us trace the designated occurrence of $A$ inside the stoup upwards along $\DerR$. 
Each principal $!C'$ application branches the trace.
The trace also branches on applications of $\wedge R$ and $\vee L$.
Each branch ends at a principal application of ${!}P$ (see Figure~\ref{Fig:deep1}).

\begin{figure}
\includegraphics[scale=.9]{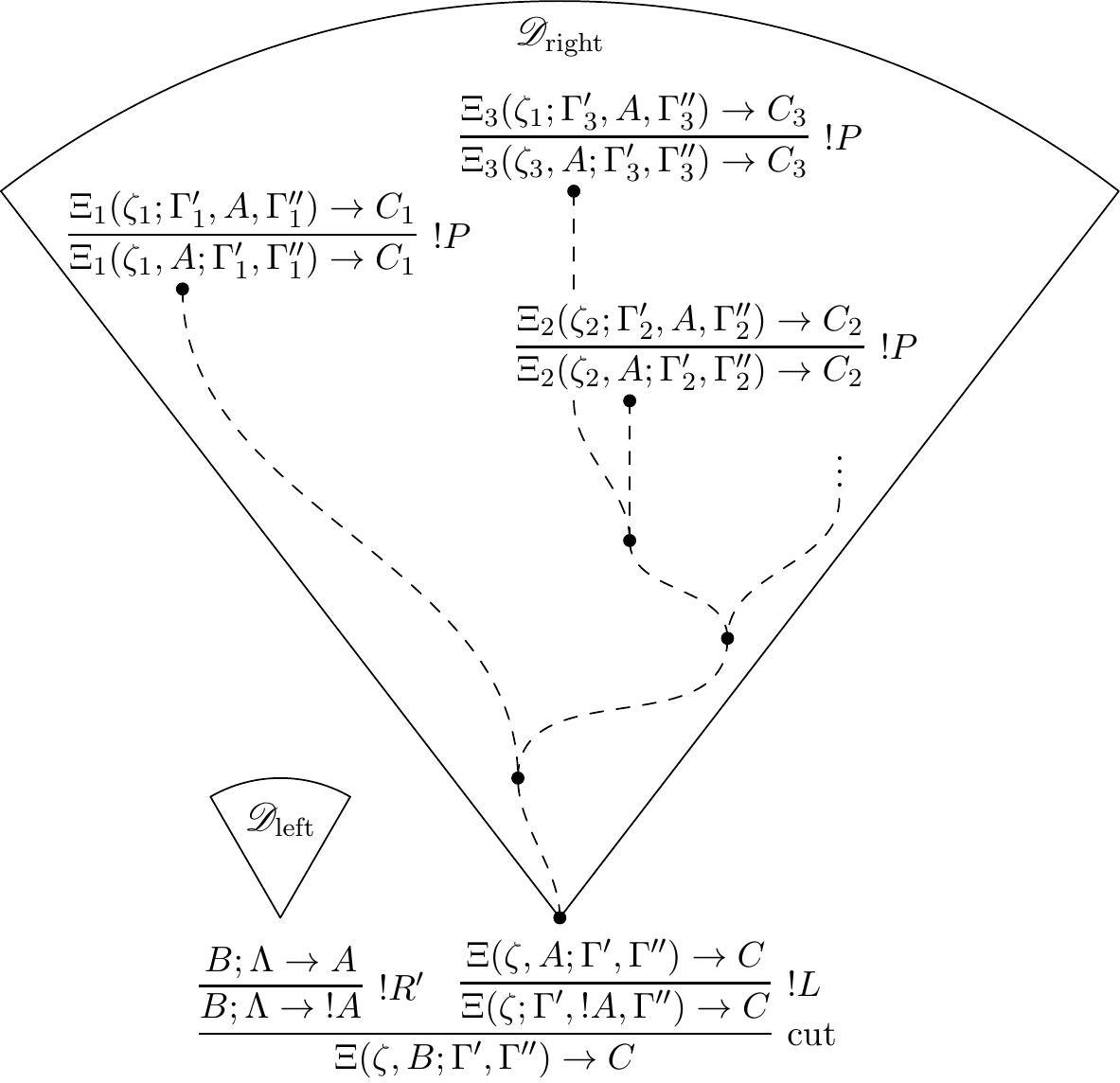}
\caption{Tracing $A$ in the stoup up to ${!}P$}\label{Fig:deep1}
\end{figure}

Now we perform 
the deep cut elimination step. In $\DerR$, we replace the designated occurrences of $A$ in the stoup with $B$.
The applications of $!C'$ remain valid. Other
rules do not operate $A$ in the stoup and therefore remain intact. After this replacement
applications of ${!}P$ transform into applications of cut with $B;\Lambda \to A$ as the left
premise (Figure~\ref{Fig:deep2}). One trace could go through several instances of ${!}P$ with
the active $A$, like $\Xi_2$ and $\Xi_3$ in the example; in this case we go from top to
bottom.

\begin{figure}
\includegraphics[scale=.85]{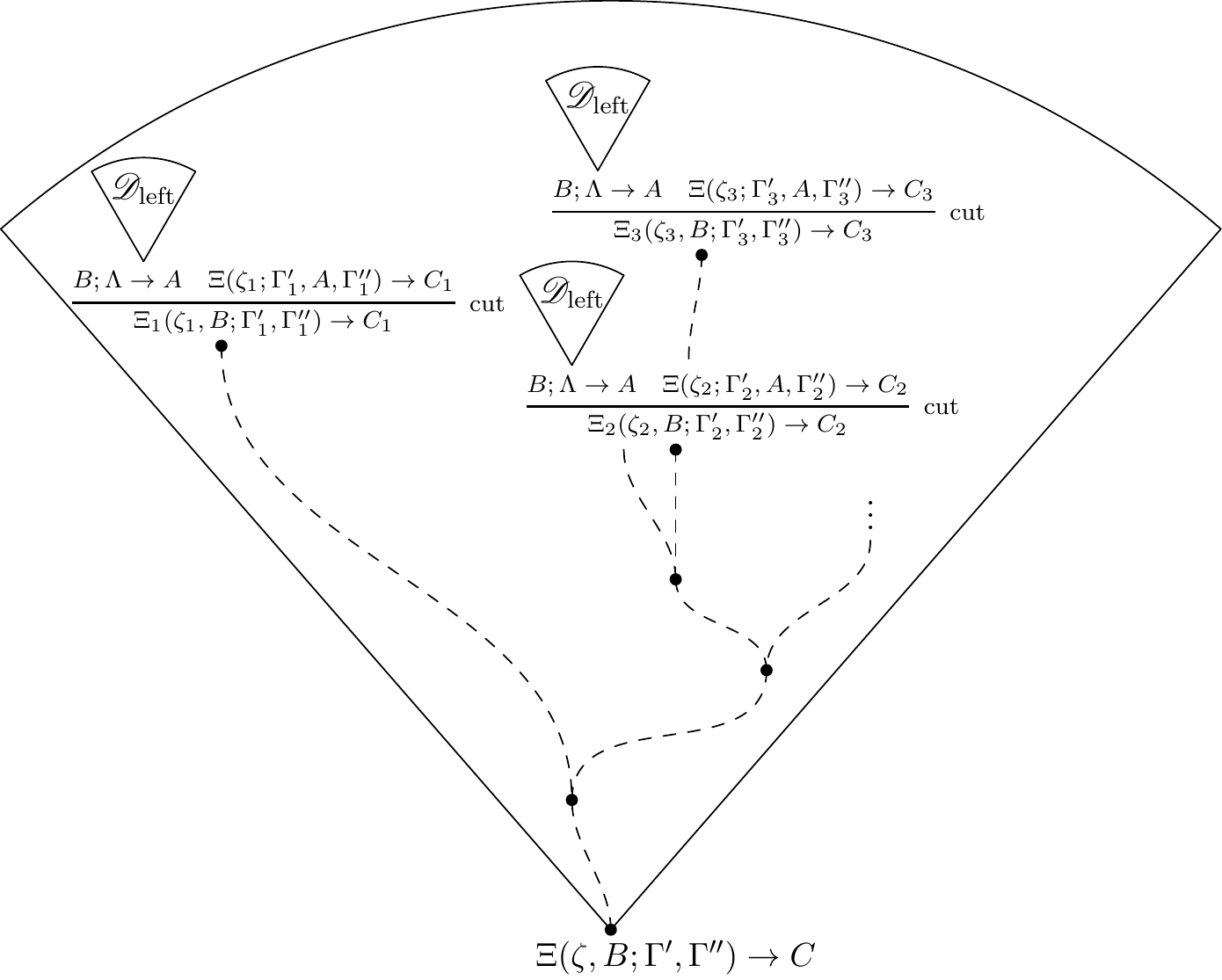}
\caption{Deep cut elimination}\label{Fig:deep2}
\end{figure}

The new cuts have lower $\kappa$ (the cut formula is $A$ instead of ${!}A$), and therefore
they are eliminable by induction hypothesis.

For the case with Lambek's restriction, notice that in the deep cut elimination step we just
changed $A$ to $B$ in the stoups, so Lambek's restriction could not get violated.

The figures illustrating deep cut elimination (Figure~\ref{Fig:deep1} and Figure~\ref{Fig:deep2}) are taken
from~\citet{KanKuzSceFCT}, with the necessary changes for calculi with stoups.

{\bf Case 2 (principal, but not ${!}$).} As said before, this case is the standard principal case for cut elimination
in the multiplicative-additive Lambek calculus with brackets~\citep{MoortgatMultimodal}, since ${!}$ does not appear in this case.
Adding the stoups makes only a minor difference.
All the interesting things about ${!}$ have already happened in the ``deep'' Case~1.

So, the main connective of $A$ is not ${!}$. Consider other possible cases.

{\em Subcase 2.a.} $A = A_1 \BS A_2$ or $A = A_2 \SL A_1$ (the latter is of course symmetric to the former).
In this case
$$\small
\infer[\CUT]{\Xi(\xi, \zeta_1, \zeta_2; \Delta_1, \Gamma, \Pi, \Delta_2) \to C}
{\infer[\BS R]{\xi; \Pi \to A_1 \BS A_2}{\xi; A_1, \Pi \to A_2} & 
\infer[\BS L]{\Xi(\zeta_1, \zeta_2; \Delta_1, \Gamma, A_1 \BS A_2, \Delta_2) \to C}{\zeta_1; \Gamma \to A_1 & \Xi(\zeta_2; \Delta_1, A_2, \Delta_2) \to C}}
$$
transforms into
$$\small
\infer[\CUT]{\Xi (\xi, \zeta_1, \zeta_2; \Delta_1, \Gamma, \Pi, \Delta_2) \to C}
{\zeta_1; \Gamma \to A_1 & \infer[\CUT]{\Xi(\xi, \zeta_2; \Delta_1, A_1, \Pi, \Delta_2) \to C}{\xi; A_1, \Pi \to A_2 & \Xi(\zeta_2; \Delta_1, A_2, \Delta_2) \to C}}
$$
Both new cuts have a smaller $\kappa$ parameter (and then we do not care for $\sigma$).

Here in the new derivation no new $\BS R$ instance was added, so Lambek's restriction is observed.

{\em Subcase 2.b.} $A = A_1 \cdot A_2$.
In this case
$$\small
\infer[\CUT]{\Xi (\xi_1, \xi_2, \zeta; \Delta_1, \Pi_1, \Pi_2, \Delta_2) \to C}
{\infer[\cdot R]{\xi_1, \xi_2; \Pi_1, \Pi_2 \to A_1 \cdot A_2}{\xi_1; \Pi_1 \to A_1 & \xi_2; \Pi_2 \to A_2} & 
\infer[\cdot L]{\Xi(\zeta; \Delta_1, A_1 \cdot A_2, \Delta_2) \to C}{\Xi(\zeta; \Delta_1, A_1, A_2, \Delta_2) \to C}}
$$
transforms into
$$\small
\infer[\CUT]{\Xi (\xi_1, \xi_2, \zeta; \Delta_1, \Pi_1, \Pi_2, \Delta_2) \to C}
{\xi_1; \Pi_1 \to A_1 & \infer[\CUT]{\Xi(\xi_2, \zeta; \Delta_1, A_1, \Pi_2, \Delta_2) \to C}{\xi_2; \Pi_2 \to A_2 & 
\Xi(\zeta; \Delta_1, A_1, A_2, \Delta_2) \to C}}
$$
(again $\kappa$ decreases).

{\em Subcase 2.c.} $A = \U$:
$$\small
\infer[\CUT]{\Xi(\zeta; \Delta_1, \Delta_2) \to C}{\Lambda \to \U & \infer[\U L]{\Xi(\zeta; \Delta_1, \U, \Delta_2) \to C}
{\Xi(\zeta; \Delta_1, \Delta_2) \to C}}
$$
The goal coincides with the right premise, so we just remove this detour and arrive at a cut-free proof of
$\Xi(\zeta; \Delta_1, \Delta_2) \to C$.

{\em Subcase 2.d.} $A = A_1 \vee A_2$:
$$\small
\infer[\CUT]{\Xi(\xi,\zeta; \Delta_1, \Pi, \Delta_2) \to C}
{\infer[\vee R_i]{\xi; \Pi \to A_1 \vee A_2}{\xi; \Pi \to A_i} & 
\infer[\vee L]{\Xi(\zeta; \Delta_1, A_1 \vee A_2, \Delta_2) \to C}
{\Xi(\zeta; \Delta_1, A_1, \Delta_2) \to C & \Xi(\zeta; \Delta_1, A_2, \Delta_2) \to C}}
$$
transforms into
$$\small
\infer[\CUT]{\Xi(\xi,\zeta; \Delta_1, \Pi, \Delta_2) \to C}
{\xi; \Pi \to A_i & \Xi(\zeta; \Delta_1, A_i, \Delta_2) \to C}
$$
($\kappa$ gets decreased, and the derivation of $\Pi \to A_j$ for $j \ne i$ gets forgotten).

{\em Subcase 2.e.} $A = A_1 \wedge A_2$:
$$\small
\infer[\CUT]{\Xi(\xi,\zeta; \Delta_1, \Pi, \Delta_2) \to C}
{\infer[\wedge R]{\xi; \Pi \to A_1 \wedge A_2}{\xi; \Pi \to A_1 & \xi; \Pi \to A_2} 
& \infer[\wedge L_j]{\Xi(\xi,\zeta; \Delta_1, A_1 \wedge A_2, \Delta_2) \to C}
{\Xi(\xi,\zeta; \Delta_1, A_j, \Delta_2) \to C}}
$$
transforms into
$$\small
\infer[\CUT]{\Xi(\xi,\zeta; \Delta_1, \Pi, \Delta_2) \to C}
{\xi; \Pi \to A_j & \Xi(\xi,\zeta; \Delta_1, A_j, \Delta_2) \to C}
$$
($\kappa$ gets decreased, and the derivation of $\xi; \Pi \to A_i$ for $i \ne j$ gets forgotten).

{\em Subcase 2.f.} $A = \PMod A'$. Notice that here $\xi$ is empty, otherwise $\PMod R$ could not be applied.
$$\small
\infer[\CUT]{\Xi(\zeta; \Delta_1, [\Pi], \Delta_2) \to C}
{\infer[\PMod R]{[\Pi] \to \PMod A'}{\Pi \to A'} & 
\infer[\PMod L]{\Xi(\zeta; \Delta_1, \PMod A', \Delta_2 \to C}{\Xi(\zeta; \Delta_1, [A'], \Delta_2) \to C}}
$$
transforms into
$$\small
\infer[\CUT]{\Xi(\zeta; \Delta_1, [\Pi], \Delta_2) \to C}{\Pi \to A' & \Xi(\zeta; \Delta_1, [A'], \Delta_2) \to C}
$$
($\kappa$ decreases).

{\em Subcase 2.g.} $A = \NMod A'$. In this subcase, notice that the stoup of the meta-formula including the active $\NMod A'$ is empty, by the $\NMod L$ rule;
$\zeta'$ below is the stoup of a {\em different} bracketed domain.
$$\small
\infer[\CUT]{\Xi(\zeta'; \Delta_1, [\xi; \Pi], \Delta_2) \to C}
{\infer[\NMod R]{\xi;\Pi \to \NMod A'}{[\xi;\Pi] \to A'} & 
\infer[\NMod L]{\Xi(\zeta'; \Delta_1, [\NMod A'], \Delta_2) \to C}{\Xi(\zeta'; \Delta_1, A', \Delta_2) \to C}}
$$
transforms into
$$\small
\infer[\CUT]{\Xi(\zeta'; \Delta_1, [\xi; \Pi], \Delta_2) \to C}{[\xi; \Pi] \to A' & \Xi(\zeta'; \Delta_1, A', \Delta_2) \to C}
$$
($\kappa$ decreases).

{\bf Case 3 (left non-principal).}
The lowermost rule of $\DerL$ is non-principal if and only if it is a left rule, {\em i.e.,}
operates in the antecedent. 

{\em Subcase 3.a.} The lowermost rule in $\DerL$ is one of the one-premise rules,
$\cdot L$, $\U L$, $\vee L$, $\PMod L$, $\NMod L$, ${!}L$, ${!}P$, ${!} C'$. 
Such rules are called {\em easy rules} in~\cite{KanKuzNigSce2018Dale}, therefore we denote this rule by $ER$:
$$\small
\infer[\CUT]{\Xi(\xi,\zeta; \Gamma_1, \Pi, \Gamma_2) \to C}{\infer[ER]{\xi;\Pi \to A}{\xi'; \Pi' \to A} & 
\Xi(\zeta; \Gamma_1, A, \Gamma_2) \to C}
$$
The easy rule is still valid in a larger context, where $\Pi$ is put between $\Gamma_1$ and $\Gamma_2$ and $\xi$ is 
added to $\zeta$, therefore one can reconstruct the
derivation as follows:
$$\small
\infer[ER]{\Xi(\xi,\zeta; \Gamma_1, \Pi, \Gamma_2) \to C}
{\infer[\CUT]{\Xi(\xi',\zeta; \Gamma_1, \Pi', \Gamma_2) \to C}{\xi';\Pi' \to A & \Xi(\zeta; \Gamma_1, A, \Gamma_2) \to C}}
$$
The new cut has a smaller $\sigma$ with the same $\kappa$, and can be eliminated by induction.
Lambek's restriction on ${!}C'$, if it was imposed, is kept.

{\em Subcase 3.b.} The lowermost rule in $\DerL$ is $\BS L$ or $\SL L$.
Here
$$\small
\infer[\CUT]
{\Xi(\xi,\zeta; \Gamma_1, \Pi(\sigma,\chi; \Delta_1, \Psi, E \BS F, \Delta_2), \Gamma_2) \to C}
{\infer[\BS L]{\xi; \Pi(\sigma,\chi; \Delta_1, \Psi, E \BS F, \Delta_2) \to A}
{\chi; \Psi \to E & \xi; \Pi(\sigma; \Delta_1, F, \Delta_2) \to A} & 
\Xi(\zeta; \Gamma_1, A, \Gamma_2) \to C}
$$
transforms into
$$\small
\infer[\SL L]
{\Xi(\xi,\zeta; \Gamma_1, \Pi(\sigma,\chi; \Delta_1, \Psi, E \BS F, \Delta_2), \Gamma_2) \to C}
{\chi; \Psi \to E & 
\infer[\CUT]{\Xi(\xi,\zeta; \Gamma_1, \Pi(\sigma; \Delta_1, F, \Delta_2), \Gamma_2) \to C}
{\xi; \Pi(\sigma; \Delta_1, F, \Delta_2) \to A & \Xi(\zeta; \Gamma_1, A, \Gamma_2) \to C}}
$$
In particular, $\Pi$ could just coincide with the designated meta-formula inside. In this case $\sigma,\chi$ gets
merged with $\xi,\zeta$ into one stoup.

The new cut has a smaller $\sigma$ with the same $\kappa$. The $\SL$ case is symmetric.

{\em Subcase 3.c.} The lowermost rule in $\DerL$ is $\vee L$. Then a derivation of the sequent $\Xi(\xi,\zeta; \Gamma_1, \Pi(\sigma; \Delta_1, B_1 \vee B_2, \Delta_2), \Gamma_2) \to C$ by cut from $\xi; \Pi(\sigma; \Delta_1, B_1 \vee B_2, \Delta_2) \to A$ and $\Xi(\zeta; \Gamma_1, A, \Gamma_2) \to C$ is transformed into a derivation by $\vee L$, whose premises are derived as follows ($i=1,2$):
$$\small
\infer[\CUT]{\Xi(\xi,\zeta; \Gamma_1, \Pi(\sigma; \Delta_1, B_i, \Delta_2), \Gamma_2) \to C}
{\xi; \Pi(\sigma; \Delta_1, B_i, \Delta_2) \to A & 
\Xi(\zeta; \Gamma_1, A, \Gamma_2) \to C}
$$
In particular, $\Pi$ could just coincide with the designated meta-formula inside. In this case $\sigma$ gets
merged with $\xi,\zeta$ into one stoup.

Now we have two cuts, but each of it has a smaller $\sigma$ parameter with the same $\kappa$, and they are independent, {\em i.e.,} derivations of the premises of
these cuts are already cut-free. Thus, we proceed by induction.

{\bf Case 4 (right non-principal).} 

{\em Subcase 4.a.}
In the right case we also have the notion of {\em easy rule,} which is a one-premise rule that does something
in the context ($\Xi$, $\zeta$, $\Gamma_1$, $\Gamma_2$, $C$) while keeping the active occurrence of $A$ intact. These rules are non-principal instances of 
$\BS R$, $\SL R$, $\cdot L$, $\U L$, $\wedge L_i$, $\vee R_i$, $\PMod L$, $\PMod R$, $\NMod L$, $\NMod R$, ${!}L$, ${!}P$. Contraction,
${!}C'$, is also essentially an easy rule, but we consider it more  accurately below (Subcase~3.a$'$).
For easy rules,
the transformation is as follows:
$$\small
\infer[\CUT]{\Xi(\xi,\zeta; \Gamma_1, \Pi, \Gamma_2) \to C}
{\xi; \Pi \to A & \infer[ER]{\Xi(\zeta; \Gamma_1, A, \Gamma_2) \to C}{\Xi'(\zeta'; \Gamma'_1, A, \Gamma'_2) \to C'}}
$$
transforms into
$$\small
\infer[ER]{\Xi(\xi,\zeta; \Gamma_1, \Pi, \Gamma_2) \to C}
{\infer[\CUT]{\Xi'(\xi,\zeta'; \Gamma'_1, \Pi, \Gamma'_2) \to C'}{\xi; \Pi \to A & \Xi'(\zeta'; \Gamma'_1, A, \Gamma'_2) \to C'}}
$$
which is legal, since the easy rule is still valid with $\Pi$ substituted for $A$ (recall that $A$ was not the active
formula in the easy rule) and $\xi$ added to the stoup. Moreover, if the original application of the easy rule, $\BS R$ or $\SL R$,
obeyed Lambek's restriction, so will the new one. The $\sigma$ parameter decreases, with the same $\kappa$.

{\em Subcase 4.a$'$.} The lowermost rule in $\DerR$ is ${!}C'$. The interesting situation here is when contraction and cut are
performed in the same meta-formula; otherwise ${!}C'$ acts as an easy rule, considered in the previous subcase. Since ${!}C'$ operates two
bracketed domains (the outer domain and the island), we have two subsituations.

Namely,
$$\small
\infer[\CUT]{\Xi(\xi,\zeta,B; \Gamma'_1, \Pi, \Gamma''_1, [[\zeta'; \Gamma_2]], \Gamma_3) \to C}
{\xi;\Pi \to A &  \infer[{!}C']{\Xi(\zeta, B; \Gamma'_1, A, \Gamma''_1, [[\zeta'; \Gamma_2]], \Gamma_3) \to C}
{\Xi(\zeta, B; \Gamma'_1, A, \Gamma''_1, [\zeta', B; \Gamma_2], \Gamma_3) \to C}}
$$
transforms into
$$\small
\infer[{!}C']{\Xi(\xi,\zeta,B; \Gamma'_1, \Pi, \Gamma''_1, [[\zeta'; \Gamma_2]], \Gamma_3) \to C}
{\infer[\CUT]{\Xi(\xi,\zeta,B; \Gamma'_1, \Pi, \Gamma''_1, [\zeta',B; \Gamma_2], \Gamma_3) \to C}
{\xi;\Pi \to A & \Xi(\zeta, B; \Gamma'_1, A, \Gamma''_1, [\zeta', B; \Gamma_2], \Gamma_3) \to C}}
$$
(the case with $A$  in $\Gamma_3$ is considered symmetrically), and
$$\small
\infer[\CUT]{\Xi(\zeta,B; \Gamma_1, [[\xi,\zeta'; \Gamma'_2, \Pi, \Gamma''_2]], \Gamma_3) \to C}
{\xi;\Pi \to A & \infer[{!}C']{\Xi(\zeta,B; \Gamma_1, [[\zeta'; \Gamma'_2, A, \Gamma''_2]], \Gamma_3) \to C}
{\Xi(\zeta,B; \Gamma_1, [\zeta', B; \Gamma'_2, A, \Gamma''_2], \Gamma_3) \to C}}
$$
transforms into
$$\small
\infer[{!}C']{\Xi(\zeta,B; \Gamma_1, [[\xi,\zeta'; \Gamma'_2, \Pi, \Gamma''_2]], \Gamma_3) \to C}
{\infer[\CUT]{\Xi(\zeta,B; \Gamma_1, [\xi,\zeta',B; \Gamma'_2, \Pi, \Gamma''_2], \Gamma_3) \to C}
{\xi;\Pi \to A & \Xi(\zeta,B; \Gamma_1, [\zeta', B; \Gamma'_2, A, \Gamma''_2], \Gamma_3) \to C}}
$$

The $\sigma$ parameter decreases, with the same $\kappa$. Lambek's restriction on ${!}C$, if imposed, is kept.

{\em Subcase 4.b} (the counterpart of 3.b). The lowermost rule in $\DerR$ is a non-principal instance of $\BS L$ or $\SL L$.
We consider $\BS L$; $\SL L$ is symmetric. There are three possible situations, depending on the relative locations of the active $A$ of cut
and the active $E \BS F$ of $\BS L$.

If the active $A$ goes to the left premise of $\BS L$, then
$$\small
\infer[\CUT]{\Xi(\chi,\sigma; \Delta_1, \Phi(\xi,\zeta; \Gamma_1, \Pi, \Gamma_2), E \BS F, \Delta_2) \to C}
{\xi;\Pi \to A & 
\infer[\BS L]{\Xi(\chi,\sigma; \Delta_1, \Phi(\zeta; \Gamma_1, A, \Gamma_2), E \BS F, \Delta_2) \to C}
{\chi; \Phi(\zeta; \Gamma_1, A, \Gamma_2) \to E & \Xi(\sigma; \Delta_1, F, \Delta_2) \to C}
}
$$
transforms into
$$\small
\infer[\BS L]
{\Xi(\chi,\sigma; \Delta_1, \Phi(\xi,\zeta; \Gamma_1, \Pi, \Gamma_2), E \BS F, \Delta_2) \to C}
{\infer[\CUT]{\chi; \Phi(\xi,\zeta; \Gamma_1, \Pi, \Gamma_2) \to E}{\xi;\Pi \to A & \chi; \Phi(\zeta; \Gamma_1, A, \Gamma_2) \to E} & 
\Xi(\sigma; \Delta_1, F, \Delta_2) \to C}
$$
In particular, $\Phi$ could be just $\xi,\zeta; \Gamma_1, \Pi, \Gamma_2$, in which case $\xi,\zeta$ gets merged with $\chi,\sigma$ into one stoup.

If the active $A$ goes to the right premise of $\BS L$, there are two more cases:
$$\small
\infer[\CUT]{\Xi(\chi,\sigma; \Delta_1 (\xi,\zeta; \Gamma_1, \Pi, \Gamma_2), \Phi, E \BS F, \Delta_2) \to C}
{\xi; \Pi \to A & 
\infer[\BS L]{\Xi(\chi,\sigma; \Delta_1(\zeta; \Gamma_1, A, \Gamma_2), \Phi, E \BS F, \Delta_2) \to C}
{\chi; \Phi \to E & \Xi(\sigma; \Delta_1(\zeta; \Gamma_1, A, \Gamma_2), F, \Delta_2) \to C}}
$$
transforms into
$$\small
\infer[\BS L]{\Xi(\chi,\sigma; \Delta_1 (\xi,\zeta; \Gamma_1, \Pi, \Gamma_2), \Phi, E \BS F, \Delta_2) \to C}
{\chi; \Phi \to E & 
\infer[\CUT]{\Xi(\sigma; \Delta_1(\xi,\zeta; \Gamma_1, \Pi, \Gamma_2), F, \Delta_2) \to C}
{\xi; \Pi \to A & \Xi(\sigma; \Delta_1(\zeta; \Gamma_1, A, \Gamma_2), F, \Delta_2) \to C}}
$$
(in particular, $\Delta_1$ could be just $\xi, \zeta; \Gamma_1, \Pi, \Gamma_2$, in which case $\xi,\zeta$ gets merged with
$\chi,\sigma$ into one stoup; the $\Delta_2$ case is symmetric), and, finally,
$$\small
\infer[\CUT]{\Xi(\xi,\zeta; \Gamma_1, \Pi, \Gamma_2)(\chi, \sigma; \Delta_1, \Phi, E \BS F, \Delta_2) \to C}
{\xi; \Pi \to A & \infer[\BS L]{\Xi(\zeta; \Gamma_1, A, \Gamma_2)(\chi, \sigma; \Delta_1, \Phi, E \BS F, \Delta_2) \to C}
{\chi; \Phi \to E & \Xi(\zeta; \Gamma_1, A, \Gamma_2)(\sigma; \Delta_1, F, \Delta_2) \to C}}
$$
transforms into
$$\small
\infer[\BS L]{\Xi(\xi,\zeta; \Gamma_1, \Pi, \Gamma_2)(\chi, \sigma; \Delta_1, \Phi, E \BS F, \Delta_2) \to C}
{\chi; \Phi \to E & \infer[\CUT]
{\Xi(\xi,\zeta; \Gamma_1, \Pi, \Gamma_2)(\sigma; \Delta_1, F, \Delta_2) \to C}
{\xi; \Pi \to A & \Xi(\zeta; \Gamma_1, A, \Gamma_2)(\sigma; \Delta_1, F, \Delta_2) \to C}}
$$
The notation $\Xi(\ldots)(\ldots)$ means a meta-formula with {\em two} designated sub-meta-formulae, which are
independent (i.e.,  do not intersect).

The $\sigma$ parameter decreases, with the same $\kappa$.

{\em Subcase 4.c} (similar to the previous one). The lowermost rule of $\DerR$ is $\cdot R$:
$$\small
\infer[\CUT]{\sigma_1, \sigma_2; \Gamma_1(\xi,\zeta; \Delta_1, \Pi, \Delta_2), \Gamma_2 \to E \cdot F}
{\xi;\Pi \to A & 
\infer[\cdot R]{\sigma_1, \sigma_2; \Gamma_1(\zeta; \Delta_1, A, \Delta_2), \Gamma_2 \to E \cdot F}
{\sigma_1; \Gamma_1(\zeta; \Delta_1, A, \Delta_2) \to E & \sigma_2; \Gamma_2 \to F}}
$$
transforms into
$$\small
\infer[\cdot R]{\sigma_1, \sigma_2; \Gamma_1(\xi,\zeta; \Delta_1, \Pi, \Delta_2), \Gamma_2 \to E \cdot F}
{\infer[\CUT]{\sigma_1; \Gamma_1(\xi, \zeta; \Delta_1, \Pi, \Delta_2) \to E}
{\xi;\Pi \to A & \sigma_1; \Gamma_1(\zeta; \Delta_1, A, \Delta_2) \to E}
& \sigma_2; \Gamma_2 \to F}
$$
Again, $\Gamma_1$ could coincide with its designated part; in this situation, $\xi,\zeta$ gets merged into $\sigma_1$.
The case where $A$ appears in $\Gamma_2$ is symmetric.

Here $\sigma$ decreases, with the same $\kappa$.

{\em Subcase 4.d} (the counterpart of 3.c). The lowermost rule in $\DerR$ is $\vee L$. Here we have two cases, depending on the
mutual location of the active $A$ of cut and the active $B_1 \vee B_2$ of $\vee L$. Namely, if they are in different (non-intersecting)
meta-formulae, then
$$\small
\infer[\CUT]{\Xi(\xi,\zeta; \Gamma_1, \Pi, \Gamma_2)(\sigma; \Delta_1, B_1 \vee B_2, \Delta_2) \to C}
{\xi;\Pi \to A & 
\infer[\vee L]{\Xi(\zeta; \Gamma_1, A, \Gamma_2)(\sigma; \Delta_1, B_1 \vee B_2, \Delta_2) \to C}
{\Xi(\zeta; \Gamma_1, A, \Gamma_2)(\sigma; \Delta_1, B_1, \Delta_2) \to C & 
\Xi(\zeta; \Gamma_1, A, \Gamma_2)(\sigma; \Delta_1, B_2, \Delta_2) \to C}}
$$
transforms into two cuts ($i=1,2$) with simpler cut formulae:
$$\small
\infer[\CUT]{\Xi(\xi,\zeta; \Gamma_1, \Pi, \Gamma_2)(\sigma; \Delta_1, B_i, \Delta_2) \to C}
{\xi;\Pi \to A & \Xi(\zeta; \Gamma_1, A, \Gamma_2)(\sigma; \Delta_1, B_i, \Delta_2) \to C}
$$
whose goals are then merged by $\vee L$.
and if $B_1 \vee B_2$ is in $\Gamma_1$ (or, symmetrically, in $\Gamma_2$), then
$$\small
\infer[\CUT]{\Xi(\xi,\zeta; \Gamma_1(\sigma; \Delta_1, B_1 \vee B_2, \Delta_2), \Pi, \Gamma_2) \to C}
{\xi; \Pi \to A & \infer[\vee L]{\Xi(\zeta; \Gamma_1(\sigma; \Delta_1, B_1 \vee B_2, \Delta_2), A, \Gamma_2) \to C}
{\Xi(\zeta; \Gamma_1(\sigma; \Delta_1, B_1, \Delta_2), A, \Gamma_2) \to C
 & \Xi(\zeta; \Gamma_1(\sigma; \Delta_1, B_2, \Delta_2), A, \Gamma_2) \to C}}
$$
 transforms into two cuts ($i=1,2$), whose goals again can be merged by $\vee L$:
 $$\small
 \infer[\CUT]{\Xi(\xi,\zeta; \Gamma_1(\sigma; \Delta_1, B_2, \Delta_2), \Pi, \Gamma_2) \to C}
{\xi; \Pi \to A & \Xi(\zeta; \Gamma_1(\sigma; \Delta_1, B_2, \Delta_2), A, \Gamma_2) \to C} 
 $$
Again, $\Gamma_1$ could coincide with its designated sub-meta-formula; in this case, $\sigma$ gets merged with 
$\xi,\zeta$.

In both situations, the two new cuts are independent and have a smaller $\sigma$ parameter with the same $\kappa$.

{\em Subcase 4.e} (similar to the previous one). The lowermost rule in $\DerR$ is $\wedge R$:
$$\small
\infer[\CUT]{\Xi(\xi,\zeta; \Gamma_1, \Pi, \Gamma_2) \to C_1 \wedge C_2}{\xi;\Pi \to A &
\infer[\wedge R]{\Xi(\zeta; \Gamma_1, A, \Gamma_2) \to C_1 \wedge C_2}
{\Xi(\zeta; \Gamma_1, A, \Gamma_2) \to C_1 & \Xi(\zeta; \Gamma_1, A, \Gamma_2) \to C_2}}
$$
transforms into
$$\small
\infer[\wedge R]{\Xi(\xi,\zeta; \Gamma_1, \Pi, \Gamma_2) \to C_1 \wedge C_2}
{\infer[\CUT]{\Xi(\xi,\zeta; \Gamma_1, \Pi, \Gamma_2) \to C_1}{\xi;\Pi \to A & \Xi(\zeta; \Gamma_1, A, \Gamma_2) \to C_1}
&\infer[\CUT]{\Xi(\xi,\zeta; \Gamma_1, \Pi, \Gamma_2) \to C_2}{\xi;\Pi \to A & \Xi(\zeta; \Gamma_1, A, \Gamma_2) \to C_2}}
$$
Again, the two new cuts have a smaller $\sigma$ and the same $\kappa$.

{\bf Case 5 (axiom).} One of the sequents is an axiom of the form $A \to A$.
Then cut disappears, since the other premise coincides with the goal.
\qed
\end{proof}

\begin{theorem}\label{Th:cutelimA}
Let sequents $\xi; \Pi \to A$ and $\Xi(\zeta; \Gamma_1, A, \Gamma_2 \to C)$ be derivable
in $\sysAa$. 
Then the sequent $\Xi(\xi,\zeta; \Gamma_1, \Pi, \Gamma_2) \to C)$ is also derivable
in $\sysAa$. 
\end{theorem}

\begin{proof}
The $\sysAa$ system differs from $\sysBa$ only in two rules: ${!}R'$ and ${!}C'$. Thus, we have to
reconsider Case~1 and Subcase~4.a$'$ (in Subcase~3.a, ${!}C'$ still acts as an ``easy rule'').

{\em Case~1.} The bottom of $\DerL$ now is
$$
\infer[!R']{\xi; \Lambda \to {!}A}{\xi;\Lambda \to A}
$$
with $\xi \ne \varnothing$. We perform the same deep cut elimination procedure, as in Theorem~\ref{Th:cutelim}. Namely,
the lowermost rule application in $\DerR$ is
$$
\infer[!L]{\Xi(\zeta; \Gamma', {!}A, \Gamma'') \to C}{\Xi(\zeta, A; \Gamma', \Gamma'') \to C}
$$
and we trace the $A$ in the stoup upwards until applications of ${!}P$. Next, we replace these occurrences of $A$ with $\xi$.

The active ${!}P$ applications,
$$
\infer[{!}P]
{\Xi_i(\zeta_i, A; \Gamma'_i, \Gamma''_i) \to C_i}
{\Xi_i(\zeta_i; \Gamma'_i, A, \Gamma''_i) \to C_i}
$$
transform into cuts with a smaller $\kappa$:
$$
\infer[\CUT]{\Xi_i(\zeta_i, \xi; \Gamma'_i, \Gamma''_i) \to C_i}
{\xi; \Lambda \to A & \Xi_i(\zeta_i; \Gamma'_i, A, \Gamma''_i) \to C_i}
$$

Rule applications along the trace remain valid. Notice that the stoup non-emptiness conditions on ${!}R'$ and ${!}C'$ are maintained by
non-emptiness of $\xi$.

{\em Subcase 4.a$'$.} In this case,
$$
\infer[\CUT]{\Xi(\xi,\zeta_1,\zeta_2,\zeta'; \Gamma'_1, \Pi, \Gamma''_1, \Gamma_2, \Gamma_3) \to C}
{\xi; \Pi \to A &
\infer[{!}C']{\Xi(\zeta_1,\zeta_2,\zeta'; \Gamma'_1, A, \Gamma''_1, \Gamma_2, \Gamma_3) \to C}
{\Xi(\zeta_1, \zeta_2; \Gamma'_1, A, \Gamma''_1, [\zeta',\zeta_2; \Gamma_2], \Gamma_3) \to C}}
$$
(with $\zeta_2 \ne \varnothing$) transforms into
$$
\infer[{!}C']{\Xi(\xi,\zeta_1,\zeta_2,\zeta'; \Gamma'_1, \Pi, \Gamma''_1, \Gamma_2, \Gamma_3) \to C}
{\infer[\CUT]{\Xi(\xi,\zeta_1,\zeta_2; \Gamma'_1, \Pi, \Gamma''_1, [\zeta',\zeta_2;\Gamma_2], \Gamma_3) \to C}
{\xi; \Pi \to A & \Xi(\zeta_1, \zeta_2; \Gamma'_1, A, \Gamma''_1, [\zeta',\zeta_2; \Gamma_2], \Gamma_3) \to C}}
$$
(the case with $A$ in $\Gamma_3$ is considered symmetrically), and
$$
\infer[\CUT]{\Xi(\zeta_1,\zeta_2,\xi,\zeta'; \Gamma_1, \Gamma'_2, \Pi, \Gamma''_2, \Gamma_3) \to C}
{\xi;\Pi \to A & 
\infer[{!}C']{\Xi(\zeta_1,\zeta_2,\zeta' ; \Gamma_1, \Gamma'_2, A, \Gamma''_2, \Gamma_3) \to C}
{\Xi(\zeta_1, \zeta_2; \Gamma_1, [\zeta',\zeta_2; \Gamma'_2, A, \Gamma''_2], \Gamma_3) \to C}}
$$
(again, $\zeta_2 \ne \varnothing$) transforms into
$$
\infer[{!}C']{\Xi(\zeta_1,\zeta_2,\xi,\zeta'; \Gamma_1, \Gamma'_2, \Pi, \Gamma''_2, \Gamma_3) \to C}
{\infer[\CUT]{\Xi(\zeta_1,\zeta_2; \Gamma_1, [\xi,\zeta',\zeta_2; \Gamma'_2, \Pi, \Gamma''_2], \Gamma_3) \to C}
{\xi;\Pi \to A & \Xi(\zeta_1, \zeta_2; \Gamma_1, [\zeta',\zeta_2; \Gamma'_2, A, \Gamma''_2], \Gamma_3) \to C }}
$$\qed
\end{proof}

\section{Versions of Modified Morrill's Systems without Stoups}\label{S:nostoups}

In this section we provide alternative, in a sense more traditional formulations of $\sysAa$, $\sysBa$, and $\sysBar$ without
stoups (like in~\cite{MorrillValentin}). We denote these calculi by $\sysAfl$, $\sysBfl$, and $\sysBrfl$ respectively.

Formulae, meta-formulae, and sequents are defined as in Morrill's system (Section~\ref{S:calculi}), but all the stoups are empty now.

We start with $\sysAfl$ and $\sysBfl$. Both calculi are built on top of~$\MALCsb$, the multiplicative-additive Lambek calculus with 
brackets~\citep{Morrill1992,MoortgatMultimodal}. Axioms and rules of $\MALCsb$ are as follows:

$$
\infer[\mathrm{id}]{A \yd A}{}
$$

$$
\infer[{\SL} L]{\Xi (\Delta_1, C \SL B, \Gamma, \Delta_2) \yd D}
{\Gamma \yd B & \Xi(\Delta_1, C, \Delta_2 ) \yd D}
\qquad
\infer[{\SL} R]{\Gamma \yd C \SL B}{\Gamma, B \yd C}
$$

$$
\infer[{\BS} L]{\Xi (\Delta_1, \Gamma, A \BS C, \Delta_2) \yd D}
{\Gamma \yd A & \Xi(\Delta_1, C, \Delta_2 ) \yd D}
\qquad
\infer[{\BS} R]{\Gamma \yd A \BS C}{A, \Gamma \yd C}
$$

$$
\infer[{\mconj} L]{\Xi (\Delta_1, A \mconj B, \Delta_2) \yd D}
{\Xi (\Delta_1, A, B, \Delta_2) \yd D}
\qquad
\infer[{\mconj} R]{\Delta, \Gamma \yd A \mconj B}
{\Delta \yd A & \Gamma \yd B}
$$

$$
\infer[{\One} L]{\Xi(\Delta_1, \One, \Delta_2) \yd A}{\Xi(\Delta_1,\Delta_2) \yd A}
\qquad
\infer[{\One} R]{\Lambda \yd \One}{}
$$

$$
\infer[{\NMod} L]{\Xi(\Delta_1, [\NMod A], \Delta_2) \yd B}
{\Xi (\Delta_1, A, \Delta_2) \yd B}
\qquad
\infer[{\NMod} R]{\Xi \yd \NMod A}{[ \Xi ] \yd A}
$$

$$
\infer[{\PMod} L]{\Xi(\Delta_1, \PMod A, \Delta_2) \yd B}
{\Xi(\Delta_1, [A], \Delta_2) \yd B}
\qquad
\infer[{\PMod} R]{[\Xi] \yd \PMod A}{\Xi \yd A}
$$

The following rules for ${!}$ are the same in both systems, $\sysAfl$ and $\sysBfl$:
$$
\infer[!L]{\Xi(\Delta_1, {!}A, \Delta_2) \to C}{\Xi(\Delta_1, A, \Delta_2) \to C}
$$
$$
\infer[!P_1]{\Xi(\Delta_1, \Phi, {!}A, \Delta_2) \to C}{\Xi(\Delta_1, {!}A, \Phi, \Delta_2) \to C}
\qquad
\infer[!P_2]{\Xi(\Delta_1, {!}A, \Phi, \Delta_2) \to C}{\Xi(\Delta_1, \Phi, {!}A, \Delta_2) \to C}
$$

The difference between $\sysAfl$ and $\sysBfl$ is in the $!R$ and $!C$ rules. In $\sysAfl$, they are formulated as follows
$$
\infer[!R,\ n \geq 1]{{!}A_1, \ldots, {!}A_n \to {!}B}{{!}A_1, \ldots, {!}A_n \to B}
$$
$$
\infer[!C,\ n \geq 1]{\Xi({!}A_1, \ldots, {!}A_n, \Gamma_1, \Gamma_2, \Gamma_3) \to C}{\Xi({!}A_1, \ldots, {!}A_n, \Gamma_1, [{!}A_1, \ldots, {!}A_n, \Gamma_2], \Gamma_3) \to C}
$$
For $\sysBfl$, these rules are formulated as follows:
$$
\infer[!R]{{!}A \to {!}B}{{!}A \to B}
\qquad
\infer[!C]{\Xi({!A}, \Gamma_1, [[\Gamma_2]], \Gamma_3) \to C}{\Xi({!}A, \Gamma_1, [{!}A, \Gamma_2], \Gamma_3) \to C}
$$

The cut rule of all stoup-free calculi is formulated as follows:
$$
\infer[\CUT]{\Xi(\Gamma_1, \Pi, \Gamma_2) \to C}{\Pi \to A & \Xi(\Gamma_1, A, \Gamma_2) \to C}
$$

In order to obtain $\sysBrfl$, we impose Lambek's restriction on the rules of $\sysBfl$ in the following natural way:
\begin{itemize}
\item in $\BS R$ and $\SL R$, we require $\Gamma \ne \Lambda$;
\item in ${!}C$, we require $\Gamma_2 \ne \Lambda$.
\end{itemize}

\begin{proposition}\label{Prop:nostoupB}
Let $\Xi \to C$ be a sequent without stoups (in the context of $\sysBa$ we consider it as a sequent with empty stoups).
Then the following are equivalent:
\begin{enumerate}\itemsep=3pt
\item $\Xi \to C$ is derivable in $\sysBfl$ without cut;
\item $\Xi \to C$ is derivable in $\sysBfl$, possibly using cut;
\item $\Xi \to C$ is derivable in $\sysBa$, possibly using cut; 
\item $\Xi \to C$ is derivable in $\sysBa$ without cut.
\end{enumerate}
The same holds for the variants with Lambek's restriction, $\sysBar$ and $\sysBrfl$.
\end{proposition}

\begin{proof}
We proceed by round-robin implications: $1 \Rightarrow 2 \Rightarrow 3 \Rightarrow 4 \Rightarrow 1$.

\fbox{$1 \Rightarrow 2$} Obvious.

\fbox{$2 \Rightarrow 3$}
Consider a derivation of $\Xi \to C$ in $\sysBfl$ (possibly with cuts) and translate it into $\sysB$. The interesting case concerns
$!$-operating rules. These rules are simulated using cut, by temporarily moving the active $!$-formula to the stoup (which is
normally empty):

{\small
$$
\infer[!P_1]
{\Xi(\Gamma_1, {!}A, \Phi, \Gamma_2) \to C}
{\Xi(\Gamma_1, \Phi, {!}A, \Gamma_2) \to C}
\quad\mbox{\raisebox{.8em}{$\leadsto$}}\quad
\infer[!L]
{\Xi(\Gamma_1, {!}A, \Phi, \Gamma_2) \to C}
{\infer[\mathrm{cut}]{\Xi(A; \Gamma_1, \Phi, \Gamma_2) \to C}
{A; \Lambda \to {!}A & \Xi(\Gamma_1, \Phi, {!}A, \Gamma_2) \to C}}
$$

$$
\infer[!C]
{\Xi({!}A, \Gamma_1, [[\Gamma_2]], \Gamma_3) \to C}
{\Xi({!}A, \Gamma_1, [{!}A, \Gamma_2], \Gamma_3) \to C}
\quad\mbox{\raisebox{.8em}{$\leadsto$}}\!
\infer[!L]
{\Xi({!}A, \Gamma_1, [[\Gamma_2]], \Gamma_3) \to C}
{\infer[!C]{\Xi(A; \Gamma_1, [[\Gamma_2]], \Gamma_3) \to C}
{\infer[\mathrm{cut}]{\Xi(A; \Gamma_1, [A; \Gamma_2], \Gamma_3) \to C}
{A; \Lambda \to {!}A & 
\infer[\mathrm{cut}]{\Xi({!}A, \Gamma_1, [A; \Gamma_2], \Gamma_3) \to C}
{A; \Lambda \to {!}A & 
\Xi({!}A, \Gamma_1, [{!}A, \Gamma_2], \Gamma_3) \to C}}}}
$$

$$
\infer[!R]
{{!}A \to {!}B}
{{!}A \to B}
\quad\mbox{\raisebox{.8em}{$\leadsto$}}\quad
\infer[!L]
{{!}A \to {!}B}
{\infer[!R]{A; \Lambda \to {!}B}
{\infer[\mathrm{cut}]{A; \Lambda \to B}{A; \Lambda \to {!}A & {!}A \to B}}}
$$

$$
\infer[!L]
{\Xi(\Gamma_1, {!}A, \Gamma_2) \to C}
{\Xi(\Gamma_1, A, \Gamma_2) \to C}
\quad\mbox{\raisebox{.8em}{$\leadsto$}}\quad
\infer[!L]
{\Xi(\Gamma_1, {!}A, \Gamma_2) \to C}
{\infer[!P]{\Xi(A; \Gamma_1, \Gamma_2) \to C}{\Xi(\Gamma_1, A, \Gamma_2) \to C}}
$$
}

The sequent $A; \Lambda \to {!}A$ is derived as follows:
{\small
$$
\infer[{!}R]{A; \Lambda \to {!}A}
{\infer[{!}P]{A; \Lambda \to A}{A \to A}}
$$}
All other rules, including cut, are translated straightforwardly.

\fbox{$3 \Rightarrow 4$} This is due to cut elimination (Theorem~\ref{Th:cutelim}).

\fbox{$4 \Rightarrow 1$} Consider a cut-free proof of $\Xi \to C$ in $\sysBa$.
In the goal sequent, all stoups are empty, but this could be not the case for sequents inside the derivation.
In each sequent, we ``flatten'' the stoups, replacing each meta-formula $\zeta; \Gamma$, where
$\zeta = \{ A_1, \ldots, A_N \}$, with ${!}A_1, \ldots, {!}A_N, \Gamma$. For $\zeta = \{ A_1, \ldots, A_N \}$, let us
denote ${!}A_1, \ldots, {!}A_N$ by ${!}\zeta$.

The rules of $\sysBa$ which do not operate the stoup are mapped to those of $\sysBfl$, adding permutations for ${!}$-formulae, if necessary.
For example, this is how it is performed for $\BS R$ and $\BS L$:

{\small
$$
\infer[\BS R]{\zeta;\Gamma \to B \BS C}{\zeta;B,\Gamma \to C}
\quad\mbox{\raisebox{.8em}{$\leadsto$}}\quad
\infer[\BS R]{{!}\zeta, \Gamma \to B \BS C}
{\infer=[{!}P]{B, {!}\zeta, \Gamma \to C}{{!}\zeta, B, \Gamma \to C}}
$$
$$
\infer[\BS L]{\Xi(\zeta_1, \zeta_2; \Delta_1, \Gamma, B \BS C, \Delta_2) \to D}
{\zeta_1; \Gamma\to B & \Xi(\zeta_2; \Delta_1, C, \Delta_2) \to D} 
\quad\mbox{\raisebox{.8em}{$\leadsto$}}\quad
\infer=[{!}P]{\Xi({!}\zeta_1, {!}\zeta_2, \Delta_1, \Gamma, B \BS C, \Delta_2) \to D}
{\infer[\BS L]{\Xi({!}\zeta_2, \Delta_1, {!}\zeta_1, \Gamma, B \BS C, \Delta_2) \to C}
{{!}\zeta_1, \Gamma \to B & \Xi({!}\zeta_2, \Delta_1, C, \Delta_2) \to D}}
$$
}

Notice how Lambek's restriction is conserved in the $\BS R$ rule.

Contraction is handled as follows:

{\small $$
\infer[!C']{\Xi(\zeta, A; \Gamma_1, [[\zeta'; \Gamma_2]], \Gamma_3) \to B}
{\Xi(\zeta, A; \Gamma_1, [\zeta', A; \Gamma_2], \Gamma_3) \to B}
\quad\mbox{\raisebox{.8em}{$\leadsto$}}\quad
\infer[!P]{\Xi({!}\zeta, {!}A, \Gamma_1, [[{!}\zeta', \Gamma_2]], \Gamma_3) \to B}
{\infer[!C]{\Xi({!}A, {!}\zeta, \Gamma_1, [[{!}\zeta', \Gamma_2]], \Gamma_3) \to B}
{\infer[!P]{\Xi({!}A, {!}\zeta, \Gamma_1, [{!}A, {!}\zeta', \Gamma_2], \Gamma_3) \to B}
{\infer[!P]{\Xi({!}A, {!}\zeta, \Gamma_2, [{!}\zeta', {!}A, \Gamma_2], \Gamma_3) \to B}
{\Xi({!}\zeta, {!}A, \Gamma_2, [{!}\zeta', {!}A, \Gamma_2], \Gamma_3) \to B}}}}
$$}

Again, we notice that Lambek's restriction (${!}\zeta_2, \Gamma_2 \ne \Lambda$) is conserved.

Finally, ${!}L$ becomes just permutation, ${!}R'$ maps to ${!}R$, and the version of ${!}P$ in $\sysBa$ maps to a combination of permutation and ${!}L$ of $\sysBfl$.
\qed
\end{proof}

\begin{proposition}\label{Prop:nostoupA}
Let $\Xi \to C$ be a sequent without stoups (in the context of $\sysAa$ we consider it as a sequent with empty stoups).
Then the following are equivalent:
\begin{enumerate}\itemsep=3pt
\item $\Xi \to C$ is derivable in $\sysAfl$ without cut;
\item $\Xi \to C$ is derivable in $\sysAfl$, possibly using cut;
\item $\Xi \to C$ is derivable in $\sysAa$, possibly using cut; 
\item $\Xi \to C$ is derivable in $\sysAa$ without cut.
\end{enumerate}
\end{proposition}

\begin{proof}
\fbox{$1 \Rightarrow 2$} Obvious.

\fbox{$2 \Rightarrow 3$} The ${!}R$ and ${!}C$ rules of $\sysAfl$ are simulated in the system $\sysAa$ as follows:
{\small
$$
\infer[{!}R,\ n \ge 1]{{!}A_1, \ldots, {!}A_n \to {!}B}
{{!}A_1, \ldots, {!}A_n \to B}
\quad\mbox{\raisebox{.8em}{$\leadsto$}}\quad
\infer=[{!}L]{{!}A_1, \ldots, {!}A_n \to {!}B}
{\infer[{!}R',\ \{A_1,\ldots,A_n\} \ne \varnothing]{A_1, \ldots, A_n; \Lambda \to {!}B}
{\infer[\CUT]{A_1, \ldots, A_n; \Lambda \to B}
{A_1; \Lambda \to {!}A & \infer{A_2, \ldots, A_n; {!}A_1 \to B}
{\infer[\CUT]{\vdots}{A_n; \Lambda \to {!}A_n & {!}A_1, \ldots, {!}A_n \to B}}}}}
$$

$$
\infer[{!}C,\ n \ge 1]{\Xi({!}A_1, \ldots, {!}A_n, \Gamma_1, \Gamma_2, \Gamma_3) \to C}
{\Xi({!}A_1, \ldots, {!}A_n, \Gamma_1, [{!}A_1, \ldots, {!}A_n, \Gamma_2], \Gamma_3) \to C}
$$
$$ \mbox{\rotatebox{270}{$\leadsto$}} $$
$$
\infer=[{!}L]{\Xi({!}A_1, \ldots, {!}A_n, \Gamma_1, \Gamma_2, \Gamma_3) \to C}
{\infer[{!}C',\ \{A_1,\ldots,A_n\} \ne \varnothing]{\Xi(A_1, \ldots, A_n; \Gamma_1, \Gamma_2, \Gamma_3) \to C}
{\infer[\CUT]{\Xi(A_1, \ldots, A_n; \Gamma_1, [A_1, \ldots, A_n; \Gamma_2], \Gamma_3) \to C}
{A_1; \Lambda \to {!}A_1 & \infer{\Xi(A_2, \ldots, A_n; {!}A_1, \Gamma_1, [A_1, \ldots, A_n; \Gamma_2], \Gamma_3) \to C}
{\infer[\CUT]{\vdots}{A_n; \Lambda \to {!}A_n & \Xi({!}A_1, \ldots, {!}A_n, \Gamma_1, [{!}A_1, \ldots, {!}A_n, \Gamma_2], \Gamma_3) \to C }}}}}
$$
}
where $A_i; \Lambda \to {!}A_i$ ($i = 1, \ldots, n$) is derived as follows:
{\small
$$
\infer[{!}R,\ \{A_i\} \ne \varnothing]{A_i; \Lambda \to {!}A_i}
{\infer[{!}P]{A_i; \Lambda \to A_i}{A_i \to A_i}}
$$
}
All other rules are translated exactly as in the proof of the $2 \Rightarrow 3$ implication of Proposition~\ref{Prop:nostoupB}.

\fbox{$3 \Rightarrow 4$} This is due to cut eliminaton, Theorem~\ref{Th:cutelimA}.

\fbox{$4 \Rightarrow 1$}
The ${!}R'$ rule of $\sysAa$ maps directly onto the ${!}R$ rule of $\sysAfl$. Contraction is handled as follows:
{\small
$$
\infer[{!}C',\ \zeta_2 \ne \varnothing]
{\Xi(\zeta_1, \zeta_2, \zeta'; \Gamma_1, \Gamma_2, \Gamma_3) \to C}
{\Xi(\zeta_1, \zeta_2; \Gamma_1, [\zeta', \zeta_2; \Gamma_2], \Gamma_3) \to C}
$$
}
transforms into
{\small
$$
\infer=[{!}P]
{\Xi({!}\zeta_1, {!}\zeta_2, {!}\zeta', \Gamma_1, \Gamma_2, \Gamma_3) \to C}
{\infer[{!}C,\ {!}\zeta_2 \ne\Lambda]{\Xi({!}\zeta_2, {!}\zeta_1, \Gamma_1, {!}\zeta', \Gamma_2, \Gamma_3) \to C}
{\infer=[{!}P]{\Xi({!}\zeta_2, {!}\zeta_1, \Gamma_1, [{!}\zeta_2, {!}\zeta', \Gamma_2], \Gamma_3) \to C}
{\Xi({!}\zeta_1, {!}\zeta_2, \Gamma_1, [{!}\zeta', {!}\zeta_2, \Gamma_2], \Gamma_3) \to C}}}
$$
}
All other rules are translated exactly as in the proof of the $4 \Rightarrow 1$ implication of Proposition~\ref{Prop:nostoupB}.
\qed
\end{proof}

Finally, for $\sysB$ we prove only one implication, since the other one does not hold (for counter-example see Section~\ref{S:issues}).

\begin{proposition}\label{Prop:nostoupBx}
If a sequent without stoups is derivable in $\sysB$, then it is also derivable in $\sysBfl$.
\end{proposition}

\begin{proof}
Recall that derivations in $\sysB$ are cut-free by definition. Thus, this proposition can be proved by modifying the $4 \Rightarrow 1$ implication
of Proposition~\ref{Prop:nostoupB}. The rules that are different in $\sysB$, if compared with
$\sysBa$, are ${!}R$ and ${!}C$. Thus, we need to reconsider these rules.
Fortunately, after ``flattening'' the stoups, ${!}R$ and ${!}C$ become identical to ${!}R'$ and ${!}C'$ respectively.
\qed
\end{proof}

\section{The $\pi$ and $\pi_q$ Projections}\label{S:nobrackets}

The calculi presented above are related to the bracket-free system $\EL$ by means of so-called
{\em bracket-forgetting projections (BFP).} 
The BFPs are going to be used in the undecidability proofs (Section~\ref{S:undec} below), in order to make
use of the standard undecidability proof for $\EL$ in the undecidability proofs for more sophisticated sequents
with brackets. We define two versions of BFP, and for simplicity we
do this for the syntax without stoups.

\begin{definition}
The {\em $\pi$-projection} of a formula is defined in the following recursive way:
\begin{align*}
& \pi(p) = p \mbox{ for $p \in \Var$} && \pi(A \cdot B) = \pi(A) \cdot \pi(B)\\
& \pi(A \BS B) = \pi(A) \BS \pi(B) && \pi(B \SL A) = \pi(B) \SL \pi(A) \\
& \pi(A_1 \wedge A_2) = \pi(A_1) \wedge \pi(A_2) && \pi(A_1 \vee A_2) = \pi(A_1) \vee \pi(A_2) \\
& \pi({!}A) = {!} \pi(A) && \pi(\U) = \U\\
& \pi(\NMod A) = \pi(A) && \pi(\PMod A) = \pi(A)
\end{align*}
For meta-formulae (tree terms) and sequents without stoups the $\pi$-projection is defined as follows:
\begin{align*}
& \pi(\Gamma_1, \ldots, \Gamma_k) = \pi(\Gamma_1), \ldots, \pi(\Gamma_k) && \pi(\Lambda) = \Lambda\\
& \pi([\Xi]) = \pi(\Xi) && \pi(\Xi \to C) = (\pi(\Xi) \to \pi(C))
\end{align*}
\end{definition}

For technical reasons we shall also need the following modification of $\pi$-projection.
\begin{definition}
Let $q$ be a designated variable. Then the {\em $\pi_q$-projection} is defined on variables as follows:
$$
\pi_q(p) = \left\{ 
\begin{aligned}
& p, && \mbox{if $p \ne q$}\\
& \U, && \mbox{if $p = q$}
\end{aligned}
\right.
$$
and then propagated to formulae, meta-formulae, and sequents exactly as $\pi$.
\end{definition}

The $\pi$-projection erases all bracket information from a sequent. Since brackets block some unwanted derivabilities, this projection is only one-way sound,
as formulated in the following definition. The $\pi_q$-projection additionally makes the special variable $q$ behave as a unit (this is going to be necessary when
handling $\sysBr$, the system with Lambek's restriction which does not include a unit constant).

\begin{definition}
A calculus $\Lc$ is {\em $\pi$-sound ($\pi_q$-sound)} in $\EL$ if for any sequent derivable in $\Lc$ its $\pi$-projection (resp., $\pi_q$-projection) is
derivable in $\EL$.
\end{definition}

Since the rules of all Morrill's systems (in the version without stoups), after applying the $\pi$-projection, map to rules of $\EL$,
we automatically get $\pi$-soundness for Morrill's systems. For $\pi_q$-soundness, we additionally notice that the axiom $q \to q$ maps to
a derivable sequent $\U \to \U$ (everything else remains as for the $\pi$-projection).
\begin{proposition}
The following calculi  are $\pi$-sound and $\pi_q$-sound in $\EL$: $\sysAfl$, $\sysBfl$, and $\sysBrfl$.
\end{proposition}

Notice that $\pi$ and $\pi_q$ lose essential information about bracketing, so the reverse implication, ``$\pi$-completeness'' or ``$\pi_q$-completeness,''
does not (and is not intended to) hold. For example, $\PMod p \to p$ is not derivable in any of the Morrill's systems, 
while its $\pi$-projection (and also $\pi_q$-projection),  $p \to p$, is an axiom of $\EL$. More interesting examples arise from the linguistic usage
of brackets (see Linguistic Introduction): {\em e.g.,} {\sl *``the girl whom John loves Mary and Pete loves''} is not assigned type $NP$
(parsed as a valid noun phrase) in bracketed calculi, but after forgetting the brackets and bracket modalities the corresponding sequent becomes
derivable in $\EL$.

For fragments without additives, there are analogous notions of $\pi$-sound\-ness and $\pi_q$-sound\-ness of a given calculus in $\ELM$.

\section{Undecidability}\label{S:undec}

In this section we prove algorithmic undecidability of the derivability problems for systems defined above. 
In order to make our results stronger, we confine ourselves to fragments of these systems which include only 
multiplicative connectives (product and divisions), brackets and bracket modalities, and the subexponential, but
not additive connectives. For the full $\MALC$-variants of the calculi, the corresponding undecidability results follow as corollaries, by
conservativity.

Recall the well-known proof of undecidability for $\ELM$, the Lambek calculus (without brackets) enriched with a full-power exponential 
modality~\citep{LMSS,Kanazawa,KanKuzNigSce2018Dale}, via encoding of {\em semi-Thue systems}.

A {\em semi-Thue system}~\citep{Thue} is a pair $\STS = \langle \Af, P \rangle$, where $\Af$ is a finite alphabet and
$P$ is a finite set of rewriting rules of the form $\alpha \Rightarrow \beta$, where $\alpha$ and $\beta$
are words (possibly empty) over $\Af$. A rule from $P$ can be {\em applied} as follows:
$\eta\,\alpha\,\theta \Rightarrow_\STS \eta\,\beta\,\theta$, where $(\alpha \Rightarrow \beta) \in P$ and
$\eta$ and $\theta$ are arbitrary (possibly empty) words over $\Af$. By $\Rightarrow^*_\STS$ (rewritability
relation in $\STS$) we denote the
reflexive-transitive closure of $\Rightarrow_\STS$.

We use the following classical result by~\citet{markov47dan} and~\citet{post47jsl}: 

\begin{theorem}\label{Th:Markov}
There exists a semi-Thue system $\STS$ for which the $\Rightarrow^*_\STS$ relation is algorithmically undecidable,
i.e., there exists no algorithm that, given words $\gamma$ and $\delta$, decides whether $\gamma \Rightarrow^*_\STS \delta$.
\end{theorem}

Before going into the details of our undecidability proof, we sketch the ideas. As noticed by~\citet{BuszkoZML,Buszkowski2005nonlog}, semi-Thue systems can be encoded as finite {\em theories} (that is, sets of sequents considered as additional axioms) over the Lambek calculus.
Following~\citet{Buszkowski2005nonlog}, such an encoding can 
performed in a very natural way, by taking a new axiom $y_1, \ldots, y_m \to x_1 \cdot\ldots\cdot x_k$ for each rule
$x_1 \ldots x_k \Rightarrow y_1 \ldots y_m$ of the semi-Thue system. (Notice how the arrows change their direction here.) 

\begin{proposition}\label{Pr:STStoTheory}
The sequent $a_1, \ldots, a_n \to b_1 \cdot\ldots\cdot b_k$ is derivable in $\MALC$ extended with the set
of new axioms produced from rules of a semi-Thue system $\STS$ (as explained above) if and only if 
$b_1 \ldots b_k \Rightarrow^*_{\STS} a_1 \ldots a_n$.
\end{proposition}

By Theorem~\ref{Th:Markov}, this yields undecidability for the problem of derivability from finite sets of hypotheses in the Lambek calculus.

In his earlier paper, \citet{BuszkoZML} provides a much more sophisticated encoding of semi-Thue derivations using only one division operation
(the encoding above also uses product). We discuss this encoding later, in Section~\ref{S:Buszko}.

Reducing derivability from finite theories to ``pure'' derivability requires a sort of deduction theorem, which allows to {\em internalise}
additional axioms (hypotheses) into the sequent being derived. For classical or intuitionistic logic, for example, this could be implemented
as follows: formula $\varphi$ is derivable from hypotheses $\psi_1, \ldots, \psi_n$ if and only if the formula $\psi_1 \to (\psi_2 \to \ldots \to (\psi_n \to \varphi)
\ldots)$ is derivable without hypotheses. In the Lambek calculus without (sub)exponentials, however, such a theorem is impossible, due to the substructural nature
of the system. This is due to the fact that in derivations from hypotheses each hypothesis can be used several (but also maybe zero) times, while in the absence
of weakening and contraction the ``pure'' calculus treats each formula as a ``resource'' which should be used exactly once.

The full-power exponential, as in $\EL$, enables the structural rules of weakening, contraction, and permutation, and 
enjoys internalisation of extra axiom, in the following form.
\begin{proposition}\label{Pr:dedEL}
A sequent $\Delta \to D$ is derivable in $\EL$ from a set of sequents $\{ \Gamma_1 \to C_1, \ldots, \Gamma_N \to C_N \}$ (possibly using cut) if and only if
the sequent ${!}(C_1 \SL \prod\Gamma_1), \ldots, {!}(C_N \SL \prod\Gamma_N), \Delta \to D$ is derivable in $\EL$.
\end{proposition}
By $\prod \Gamma$, for $\Gamma = E_1, \ldots, E_k$, we here and further denote the product
$E_1 \cdot \ldots \cdot E_k$.

Theorem~\ref{Th:Markov} and Propositions~\ref{Pr:STStoTheory} and~\ref{Pr:dedEL} together yield undecidability for $\ELM$.

Subexponentials, with restricted sets of structural rules, also allow internalisation, but the
$!$-formulae used in order to obtain it are more complicated. For $\LLsM$, this is performed by~\cite{KanKuzNigSce2018Dale}
(\cite{KanKuzSceFG} use a slightly different strategy). 
We perform internalisation of finite sets of hypotheses in Morill's systems, where ${!}$ interacts with brackets.

For our undecidability proofs, it will be convenient to use Chomsky's type-0 (unrestricted) grammars~\citep{Chomsky}, a formalism closely
related to semi-Thue systems. A type-0 grammar $\Gc$ can be defined as a semi-Thue system $\STS$ with the following additional features:
\begin{itemize}
\item a designated symbol $s \in \Af$, called the {\em starting symbol};
\item a designated subset $\Sigma \subset \Af$, called the {\em terminal alphabet};
\item left-hand sides of rewriting rules are required to be non-empty.
\end{itemize}
The {\em language} generated by the type-0 grammar $\Gc$ is defined as the set of all words $w$ over the {\em terminal} alphabet $\Sigma$, 
such that $s \Rightarrow^*_\STS w$, where $\STS$ is the rewriting (semi-Thue) system of $\Gc$. Further we use the notation $\Rightarrow^*_\Gc$ instead
of $\Rightarrow^*_\STS$.

For type-0 grammars, there is the following form of Theorem~\ref{Th:Markov}:
\begin{theorem}\label{Th:Markov_Gc}
There exists a type-0 grammar $\Gc$ such that the language generated by $\Gc$ is algorithmically undecidable, 
i.e., there exists no algorithm that, given a word $w$ over $\Sigma$, decides whether $s \Rightarrow^*_\Gc w$.
\end{theorem}

Wishing to prove undecidability for several closely related calculi, 
we first introduce an abstract notion
of internalisation of finite theories in a deduction-theorem style. Then we prove undecidability for an arbitrary calculus $\Lc$ which enjoys
this property and is $\pi_q$-sound in $\ELM$. The internalisation property facilitates the ``forward'' direction  of the encoding,
from a type-0 grammar (semi-Thue system) to $\Lc$, and $\pi_q$-soundness is used for the ``backwards'' direction, from derivations
in $\Lc$, via $\ELM$, back to derivations in the grammar.

For simplicity, we are going to internalise formulae rather than sequents: recall that $\Gamma \to C$ corresponds to $C \SL \prod\Gamma$.
Further we shall designate two specific variables, $s$ and $q$. The $s$ variable is going to be the starting symbol of the grammar, and
$q$ is a fresh variable which should not appear in the grammar. The $q$ variable will be used {\em in lieu} of $\U$, since the latter could
be unavailable in presence of Lambek's restriction.

\begin{definition}\label{Df:intern}
Let $\Ac = \{ A_1, \ldots, A_N \}$ be a finite set of formulae.
A meta-formula $\Phi$ {\em internalises} $\Ac$ in the calculus $\Lc$, if the following holds:
\begin{enumerate}
\item \label{It:internS} the sequent $\Phi, s \to s$ is derivable in $\Lc$;
\item \label{It:internLand} the following `landing' rule is admissible in $\Lc$:
$$
\infer[\mathrm{land},\ A_i \in \Ac]
{\Phi, \Delta_1, \Delta_2 \to C}
{\Phi, \Delta_1, A_i, \Delta_2 \to C}
$$
\item \label{It:internBack} 
the sequent ${!}A_1, \ldots, {!}A_N \to \prod\pi_q(\Phi)$ is derivable in $\ELM$.
\end{enumerate}
If for any finite set $\Ac$ of Lambek formulae there exists a meta-formula $\Phi$ which internalises $\Ac$ in $\Lc$, then
we say that $\Lc$ {\em internalises finite sets of Lambek formulae.}
\end{definition}

Let $\Gc$ be a type-0 grammar and suppose that all letters of its alphabet are
variables ($\Af \subset \Var$). We also introduce an extra fresh variable $q \notin \Af$. Let
$$
\Ac_{\Gc} = \{ (x_1 \cdot\ldots\cdot x_k) \SL (y_1 \cdot\ldots\cdot y_m) \mid
x_1 \ldots x_k \Rightarrow y_1 \ldots y_m \mbox{ is a rewriting rule of $\Gc$}\}.
$$
By definition of a type-0 grammar, $x_1 \ldots x_k$ is always non-empty. On the other hand,
$y_1 \ldots y_m$ could be empty ($m = 0$), and in this case we include just $x_1 \cdot\ldots\cdot x_k$ into $\Bc_\Gc$.
Such a rule, with an empty right-hand side, is called an {\em $\varepsilon$-rule} and written as
$x_1 \ldots x_k \Rightarrow \varepsilon$ ($\varepsilon$ stands for the empty word).

Now we are ready to formulate and prove the key lemma. 

\begin{lemma}\label{Lm:undecRR}
Let $\Lc$ be $\pi_q$-sound in $\ELM$ and admit the $\cdot L$, $\cdot R$, and $\SL L$ rules.
Let $\Phi_{\Gc}$ internalise $\Ac_{\Gc} = \{ A_1, \ldots, A_N \}$ in $\Lc$ in the sense of Definition~\ref{Df:intern}.
Then the following are equivalent:
\begin{enumerate}
\item $s \Rightarrow_{\Gc}^* a_1 \ldots a_n$;
\item the sequent $\Phi_{\Gc}, a_1, \ldots, a_n \to s$ is derivable in $\Lc$;
\item there exists such a bracketing $\Delta$ of $a_1, \ldots, a_n$ that the sequent
$\Phi_{\Gc}, \Delta \to s$ is derivable in $\Lc$;
\item the sequent ${!}A_1, \ldots, {!}A_N, a_1, \ldots, a_n \to s$ is derivable in $\ELM$.
\end{enumerate}
\end{lemma}

\begin{proof}
We prove round-robin implications: $1 \Rightarrow 2 \Rightarrow 3 \Rightarrow 4 \Rightarrow 1$.

\fbox{$1 \Rightarrow 2$}
Proceed by induction on the number of rewriting steps in $\Rightarrow_{\Gc}^*$. The base case is no rewriting steps, 
$s \Rightarrow^*_{\Gc} s$, and $\Phi_{\Gc}, s \to s$ is derivable by Definition~\ref{Df:intern}, item~\ref{It:internS}.

Next, let 
$a_1 \ldots a_i x_1 \ldots x_m a_r \ldots a_n \Rightarrow_{\Gc} a_1 \ldots a_i y_1  \ldots y_k a_r \ldots a_n$ be the last rewriting 
step. By induction hypothesis, since $s \Rightarrow_{\Gc}^* a_1 \ldots a_i x_1 \ldots x_m a_r \ldots a_n$ in fewer rewriting steps,
we have $\Phi_{\Gc}, a_1, \ldots, a_i, x_1, \ldots, x_m, a_r, \ldots, a_n$. Since $x_1 \ldots x_m \Rightarrow y_1 \ldots y_k$ is a rule of
$\Gc$, we have $(x_1 \cdot \ldots \cdot x_m) \SL (y_1 \cdot \ldots \cdot y_m) \in \Ac_{\Gc}$.
Now the needed sequent $\Phi_{\Gc}, a_1, \ldots, a_i, y_1, \ldots, y_m, a_r, \ldots, a_n \to s$ is derived as follows,
using the `landing' rule provided by Definition~\ref{Df:intern}, item~\ref{It:internLand}:
$$
\infer[\mathrm{land}]{\Phi_{\Gc}, a_1, \ldots, a_i, y_1, \ldots, y_k, a_r, \ldots, a_n \to s}
{\infer[\SL L]{\Phi_{\Gc}, a_1, \ldots, a_i, (x_1 \cdot \ldots \cdot x_m) \SL (y_1 \cdot \ldots \cdot y_k), y_1, \ldots, y_k, a_r, \ldots, a_n \to s}
{\infer=[\cdot R]{\vphantom{\Phi} y_1, \ldots, y_k \to y_1 \cdot \ldots \cdot y_k}{\mbox{axioms}} & 
\infer=[\cdot L]{\Phi_{\Gc}, a_1, \ldots, a_i, x_1 \cdot \ldots \cdot x_m, a_r, \ldots, a_n \to s}
{\Phi_{\Gc}, a_1, \ldots, a_i, x_1, \ldots, x_m, a_r, \ldots, a_n \to s}}}
$$
(Here and further double horizontal line means several applications of the rule.)

\fbox{$2 \Rightarrow 3$} is obvious, since one just takes the trivial bracketing $\Delta = a_1, \ldots, a_n$.

\fbox{$3 \Rightarrow 4$}
Since $a_1, \ldots, a_n = \pi_q(\Delta)$ and $\Lc$ is $\pi_q$-sound in $\ELM$, the sequent 
$\pi_q(\Phi_{\Gc}), a_1, \ldots, a_n \to s$ is derivable in $\ELM$. Next, proceed as follows, using item~\ref{It:internBack} of 
Definition~\ref{Df:intern}:
$$
\infer[\CUT]{{!}B_1, \ldots, {!}B_N, a_1, \ldots, a_n \to s}{{!}B_1, \ldots, {!}B_N \to \prod \pi_q(\Phi_{\Gc}) & 
\infer=[\cdot L]{\prod \pi_q(\Phi_{\Gc}), a_1, \ldots, a_n \to s}{\pi_q(\Phi_{\Gc}), a_1, \ldots, a_n \to s}}
$$

\fbox{$4 \Rightarrow 1$} 
This part comes directly from the standard undecidability proof for $\ELM$, see~\cite{KanKuzNigSce2018Dale}.
 Consider 
the derivation of the sequent ${!}A_1, \ldots, {!}A_N, a_1, \ldots, a_n \to s$ in $\ELM$. 
The cut rule in $\ELM$ is eliminable~\citep{KanKuzNigSce2018Dale}, 
so we can suppose that
this derivation is cut-free. All formulae in this
derivation are subformulae of the goal sequent, and the only applicable rules
are ${\mconj} L$, ${\mconj} R$, ${\SL} L$, and rules operating ${!}$ in the antecedent:
${!}L$, ${!}C_{1,2}$, ${!}W$. 

Now let us hide all the formulae which include $\SL$ or ${!}$. 
This trivialises all ${!}$-operating rules. Next, let us replace
all $\mconj$'s in the antecedents with commas. 
This, in its turn, trivialises ${\mconj} L$.  
All sequents in our derivation
are now of the form $b_1, \ldots, b_\ell \Rightarrow C$, where $\ell \ge 0$ and $C = c_1 \mconj \ldots \mconj c_r$ ($r \ge 1$) or $C = \One$.
For the sake of uniformity, we also write $C = \One$ as $C = c_1 \mconj \ldots \mconj c_r$ with $r = 0$.
The $\SL L$ rule reduces to 
$$
\infer{b_1, \ldots, b_i, b_{i+1}, \ldots, b_j, b_{j+1}, \ldots, b_s \yd C}
{b_{i+1}, \ldots, b_j \to y_1 \mconj \ldots \mconj y_k & 
b_1, \ldots, b_i, x_1, \ldots, x_m, b_{j+1}, \ldots, b_s \yd C}
$$
where $x_1 \ldots x_m \Rightarrow y_1 \ldots y_k$ is a rewriting rule of $\Gc$.
For each $\varepsilon$-rule $x_1 \ldots x_m \Rightarrow \varepsilon$ in $\Gc$ we get the
rule
$$
\infer{b_1, \ldots, b_i, b_{i+1}, \ldots, b_s \yd C}
{b_1, \ldots, b_i, x_1, \ldots, x_m, b_{i+1}, \ldots, b_s \yd C}
$$
(which is the reduction of ${!}L$ for ${!}(x_1 \cdot \ldots \cdot x_m)$).
Finally, $\cdot R$ transforms into
 $$
\infer{b_1, \ldots, b_i, b_{i+1}, \ldots, b_\ell \yd c_1 \mconj \ldots \mconj c_j \mconj c_{j+1} \mconj \ldots \mconj c_r}
{b_1, \ldots, b_i \yd c_1 \mconj \ldots \mconj c_j & b_{i+1}, \ldots, b_\ell \yd c_{j+1} \mconj \ldots \mconj c_r}
$$
and axioms are of the form $a \yd a$. 

Now straightforward induction on derivation establishes the following fact: if $b_1, \ldots, b_\ell \yd c_1 \mconj \ldots \mconj c_r$ is derivable
in the simplified calculus presented above, then $b_1 \ldots b_\ell$ is derivable from $c_1 \ldots c_r$ in the type-0 grammar $|Gc$. This finishes
our proof.
\qed
\end{proof}

Theorem~\ref{Th:Markov_Gc} and Lemma~\ref{Lm:undecRR} immediately yield the following generic undecidability result (``meta-theorem'').

\begin{theorem}\label{Th:undec_generic}
Let $\Lc$ be $\pi_q$-sound in $\ELM$, admit the $\cdot L$, $\cdot R$, and $\SL L$ rules, and internalise
finite sets of Lambek formulae. Then the derivability problem in $\Lc$ is undecidable.
\end{theorem}

Now, in order to prove undecidability, it is sufficient to show that the calculi considered in this paper internalise finite sets of Lambek formulae.
The easiest example is $\ELM$, the multiplicative-additive Lambek calculus with a full-power exponential modality. 
A set $\Ac = \{ A_1, \ldots, A_N \}$ is internalised in $\ELM$ by $\Phi = {!}A_1, \ldots, {!}A_N$ (cf. Proposition~\ref{Pr:dedEL}).
For $\LLsM$, the situation is a bit trickier, since in the absence of the weakening rule
${!}A_1, \ldots, {!}A_N, s \to s$ (item~\ref{It:internS} of Definition~\ref{Df:intern}) is not derivable. 
This issue is handled by extending $\Phi$ with extra formulae  which neutralise ${!}A_i$. Namely, $\Phi = \U \SL {!}A_1, {!}A_1, \ldots
\U \SL {!}A_N, {!} A_N$ internalises $\Ac$ in $\LLsM$~\citep{KanKuzNigSce2018Dale}. 
Actually, the $\U$ constant here can be replaced by $s \SL s$.

For the calculi with brackets, which interact with the contraction rule, we go further along this line. We still have to add formulae like
$\U \SL {!}A_i$, for item~\ref{It:internS} of Definition~\ref{Df:intern}; but now we also need to neutralise the changes which the contraction rule
makes on the bracketing structures. We carry this strategy out in Propositions~\ref{Pr:internA} and~\ref{Pr:internB} below.

We start with systems without stoups: $\LsysAfl$, $\LsysBfl$, and $\LsysBrfl$.

\begin{proposition}\label{Pr:internA}
The meta-formula
$$
\Phi = {!} ((s \SL s) \SL {!} \NMod A_1), {!}\NMod A_1, \ldots,
{!} ((s \SL s) \SL {!} \NMod A_N), {!} \NMod A_N
$$
internalises $\{A_1, \ldots, A_N\}$ in~$\LsysAfl$. 
\end{proposition}

\begin{proof}

1. The sequent $\Phi, s \to s$ is derived in $\LsysAfl$ as follows:

{\small
$$
\infer=[!L]{{!} ((s \SL s) \SL {!} \NMod A_1), {!}\NMod A_1, \ldots,
{!} ((s \SL s) \SL {!} \NMod A_N), {!} \NMod A_N, s \to s}
{\infer=[\SL L]{(s \SL s) \SL {!} \NMod A_1, {!} \NMod A_1, \ldots,
(s \SL s) \SL {!} \NMod A_N, {!} \NMod A_N, s \to s}
{{!} \NMod A_1 \to {!} \NMod A_1 & \ldots & {!}\NMod A_N \to {!}\NMod A_N & 
\infer[\SL L]{s \SL s, \ldots, s \SL s, s \SL s, s \to s}{s \to s & \infer[\SL L]{s \SL s, \ldots, s \SL s, s \to s}{s \to s & \infer{\vdots}{s \to s}}}}}
$$
}

2. The `landing' rule is derived in $\LsysAfl$ as follows:

{\small
$$
\infer[!C\mbox{ applied to ${!}\NMod A_i$}]
{\Phi, \Delta_1, \Delta_2 \to C}
{\infer[{!}L]{\Phi, \Delta_1, [{!}\NMod A_i], \Delta_2 \to C}
{\infer[\NMod L]{\Phi, \Delta_1, [\NMod A_i], \Delta_2 \to C}{\Phi, \Delta_1, A_i, \Delta_2 \to C}}}
$$
}

3. Notice that
$$
\pi_q(\Phi) = {!} ((s \SL s) \SL {!}A_1), {!}A_1, \ldots, {!}((s \SL s) \SL {!}A_N), {!}A_N.
$$
Next, by applying $\cdot R$ to $\Lambda \to {!} ((s \SL s) \SL {!}A_1)$; ${!}A_1 \to {!}A_1$; \ldots
$\Lambda \to {!} ((s \SL s) \SL {!}A_N)$; ${!}A_1 \to {!}A_N$, we get the necessary sequent
${!}B_1, \ldots, {!}B_N \to \prod \pi_q(\Phi)$.

The sequents $\Lambda \to {!} ((s \SL s) \SL {!}A_i)$ are derived in $\ELM$ as follows:

{\small
$$
\infer{\Lambda \to {!}((s \SL s) \SL {!}B_1)}
{\infer{\Lambda \to (s \SL s) \SL {!}A_1}
{\infer{{!}A_1 \to s \SL s}
{\infer{\Lambda \to s \SL s}{s \to s}}}}
$$}
\qed
\end{proof}

\begin{proposition}\label{Pr:internB}
The meta-formula 
$$
\Phi = {!} ((s \SL s) \SL {!} Z_1), {!}Z_1, \ldots, {!}((s \SL s) \SL {!} Z_n), {!}Z_n, 
{!}((s \SL s) \SL \PMod\PMod q), [[ q ]],
$$
where
$$
Z_i = (\NMod ({!}A_i \cdot \PMod\PMod q)) \SL q,
$$
internalises $\{A_1, \ldots, A_n\}$ in~$\LsysBfl$ and~$\LsysBrfl$.
\end{proposition}

Notice that in $\LsysBrfl$ we do not need to impose any additional non-emptiness restriction on the `landing' rule,
since Lambek's restriction is automatically satisfied by non-emptiness of $\Phi$. Thus, for our undecidability proof
$\LsysBrfl$ is capable of encoding arbitrary semi-Thue systems, even the ones which include $\varepsilon$-rules.

For categorial grammars, however, Lambek's restriction will make a difference in the expressive power of $\LsysBrfl$ as opposed
to $\LsysBfl$, as explained in the next section.

\begin{proof}

1. The sequent $\Phi, s \to s$ is derived in $\LsysBrfl$ (and therefore in $\LsysBfl$) as follows:
{\small $$
\infer=[!L]{{!} ((s \SL s) \SL {!}Z_1), {!}Z_1, \ldots, {!}((s \SL s) \SL {!}Z_n), {!}Z_n, 
{!}((s \SL s) \SL \PMod\PMod q), [[ q ]], s \to s}
{\infer=[\SL L]{(s \SL s) \SL {!}Z_1, {!}Z_1, \ldots, (s \SL s) \SL {!}Z_n, {!}Z_n, 
(s \SL s) \SL \PMod\PMod q, [[ q ]], s \to s}
{{!}Z_1 \to {!}Z_1 & \ldots & {!}Z_n \to {!}Z_n & \infer[\PMod R]{[[q]] \to \PMod\PMod q}{\infer[\PMod R]{[q] \to \PMod q}{q \to q}} & 
\infer[\SL L]{s \SL s, \ldots, s \SL s, s \SL s, s \to s}{s \to s & \infer[\SL L]{s \SL s, \ldots, s \SL s, s \to s}{s \to s & \infer{\vdots}{s \to s}}}}}
$$}
Notice that all antecedents in this derivation are non-empty, so Lambek's restriction is observed.

2.  Let $\Phi'$ be 
$(s \SL s) \SL {!}Z_1), {!}Z_1, \ldots, {!}((s \SL s) \SL {!}Z_n), {!}Z_n, {!}((s \SL s) \SL \PMod\PMod q)$, that is, $\Phi = \Phi', [[q]]$, 
and recall that ${!}Z_i = {!} ((\NMod ({!}A_i \cdot \PMod\PMod q)) \SL q)$. 
Now the `landing' rule is derived in $\LsysBrfl$ (and therefore also in $\LsysBfl$) as follows.

{\small $$
\infer[!C\mbox{ applied to ${!}Z_i$}]{\Phi', [[q]], \Delta_1, \Delta_2 \to C}
{\infer[!L]{\Phi', [{!} ((\NMod ({!}A_i \cdot \PMod\PMod q)) \SL q), q], \Delta_1, \Delta_2 \to C}
{\infer[\SL L]{\Phi', [ (\NMod ({!}A_i \cdot \PMod\PMod q)) \SL q, q ], \Delta_1, \Delta_2 \to C}
{q \to q & \infer[\NMod L]{\Phi', [ \NMod ({!}A_i \cdot \PMod \PMod q) ], \Delta_1, \Delta_2 \to C}
{\infer[\cdot L]{\Phi', {!}A_i \cdot \PMod\PMod q, \Delta_1, \Delta_2 \to C}
{\infer[!P_2]{\Phi', {!}A_i, \PMod\PMod q, \Delta_1, \Delta_2 \to C}
{\infer[!L]{\Phi', \PMod\PMod q, \Delta_1, {!}A_i, \Delta_2 \to C}
{\infer[\PMod L]{\Phi', \PMod\PMod q, \Delta_1, A_i, \Delta_2 \to C}
{\infer[\PMod L]{\Phi', [\PMod q], \Delta_1, A_i, \Delta_2 \to C}
{\Phi', [[q]], \Delta_1, A_i, \Delta_2 \to C}}}}}}}}}
$$}

Again, notice how Lambek's restriction is maintained by non-emptiness of $\Phi'$ and by the $q$ in the brackets
(when applying $!C$).

3. 
Notice that
\begin{multline*}
\pi_q(\Phi) = {!} ((s \SL s) \SL {!}({!}A_1 \cdot \U)) \SL \U, 
{!}({!}A_1 \cdot \U), \ldots,\\
{!} ((s \SL s) \SL {!}({!}A_n \cdot \U)) \SL \U, 
{!}({!}A_n \cdot \U), {!}((s \SL s) \SL \U), \U.
\end{multline*}
Next, we enjoy the following derivations in $\ELM$:
{\small $$
\infer[\SL R]{\Lambda \to {!} ((s \SL s) \SL {!}({!}A_i \cdot \U)) \SL \U}
{\infer[\U L]{\U \to {!} ((s \SL s) \SL {!}({!}A_i \cdot \U))}
{\infer[{!} R]{\Lambda \to {!} ((s \SL s) \SL {!}({!}A_i \cdot \U))}
{\infer[\SL R]{\Lambda \to (s \SL s) \SL {!}({!}A_i \cdot \U)}
{\infer[{!} W]{{!}({!}A_i \cdot \U) \to s \SL s}
{\infer[\SL R]{\Lambda \to s \SL s}{s \to s}}}}}}
\qquad
\infer[{!} R]{{!}A_i \to {!}({!} A_i \cdot \U)}
{\infer[\cdot R]{{!}A_i \to {!}A_i \cdot \U}
{{!}A_i \to {!}A_i & \Lambda \to \U}}
\qquad
\infer[{!} R]{\Lambda \to {!}((s \SL s) \SL \U)}
{\infer[\SL R]{\Lambda \to (s \SL s) \SL \U}
{\infer[\U L]{\U \to s \SL s}
{\infer[\SL R]{\Lambda \to s \SL s}{s \to s}}}}
$$}
Also recall that $\Lambda \vdash \U$ is an axiom of $\ELM$.
Now several applications of $\cdot L$ yield  ${!}A_1, \ldots, {!}A_n \to \prod \pi_q(\Phi)$.
\qed
\end{proof}

Now by Theorem~\ref{Th:undec_generic} we get undecidability for the systems without stoups.

\begin{theorem}
The derivability problems for $\LsysAfl$, $\LsysBfl$, and $\LsysBrfl$ are undecidable.
\end{theorem}

For systems with stoups, the situation is as follows. 
First, concerning derivability of sequents with empty stoups, $\LsysAa$, $\LsysBa$, and $\LsysBar$  just equivalent to $\LsysAfl$,
$\LsysBfl$, and $\LsysBrfl$ respectively (Propositions~\ref{Prop:nostoupB} and~\ref{Prop:nostoupA}). 
This gives undecidability for these systems.

\begin{theorem}
The derivability problems for $\LsysAa$, $\LsysBa$, and $\LsysBar$ are undecidable.
\end{theorem}

Proving undecidability for original Morrill's systems, $\LsysA$ and $\LsysB$ (which we altered in order to gain cut elimination, see 
Section~\ref{S:issues}), requires some extra work. For $\LsysA$ it is easy. 
The right rule for ${!}$ is never used in the proof of item~\ref{It:internS} in Proposition~\ref{Pr:internA}.
Also, in the landing rule (item~\ref{It:internLand}) there are no other ${!}$-formulae moved into the newly created
bracketed island. Thus, from the point of view of Proposition~\ref{Pr:internA} $\LsysA$ is indistinguishable from $\LsysAa$, and
this yields the necessary undecidability result.

\begin{theorem}
The derivability problem for $\LsysA$ is undecidable.
\end{theorem}

The case of $\LsysB$ is trickier. The issue is that the ${!}L$ rule in this calculus is not invertible, {\em i.e.,} 
derivability of $\Xi(\zeta; \Delta_1, {!}A, \Delta_2) \to C$ is
not always equivalent to derivability of $\Xi(\zeta,A; \Delta_1, \Delta_2) \to C$. In particular,
${!}p \to {!}p$ is derivable, while $p; \Lambda \to {!}p$ is not. For our construction this issue
is crucial. Namely, when proving the `landing' rule (item~\ref{It:internLand} of Definition~\ref{Df:intern}), in
$\LsysB$ we would have to move the ${!}$-formulae from $\Phi$ to the stoup. Otherwise, we could not operate the contraction
rule. This, however, would cause problems with item~\ref{It:internS}: the left premises of the derivation
(see proof of Proposition~\ref{Pr:internB}) become $Z_i; \Lambda \to {!}Z_i$ (instead of ${!}Z_i \to {!}Z_i$), and in general
are not derivable. 

Fortunately, this issue is easily resolved by adding an extra ${!}$ on $Z_i$~\citep{KanKuzSceFG19}. Indeed,
${!}Z_i; \Lambda \to {!}Z_i$ {\em is} derivable from ${!}Z_i \to {!}Z_i$ by application of ${!}P$. This yields the
following internalisation property, where we essentially use the stoup for ${!}$-formulae in $\Phi$ (the rightmost
$[[q]]$ is kept outside the stoup).

\begin{proposition}\label{Pr:internBx}
The meta-formula 
$$
\Phi = (s \SL s) \SL {!} Z_1, {!}Z_1, \ldots, (s \SL s) \SL {!} Z_N, {!}Z_N, 
(s \SL s) \SL \PMod\PMod q; [[ q ]],
$$
where
$$
Z_i = (\NMod ({!}A_i \cdot \PMod\PMod q)) \SL q,
$$
satisfies items~\ref{It:internS} and~\ref{It:internLand} of Definition~\ref{Df:intern} for  internalisation of 
$\{A_1, \ldots, A_N\}$ in~$\LsysB$. 
\end{proposition}

\begin{proof}
1. The sequent $\Phi,s \to s$ is derived in $\LsysB$ as follows:

{\small
$$
\infer=[!P]{(s \SL s) \SL {!} Z_1, {!}Z_1, \ldots, (s \SL s) \SL {!} Z_N, {!}Z_N, 
(s \SL s) \SL \PMod\PMod q; [[ q ]], s \to s}
{\infer=[\SL L]{(s \SL s) \SL {!} Z_1, {!}Z_1, \ldots, (s \SL s) \SL {!} Z_N, {!}Z_N, 
(s \SL s) \SL \PMod\PMod q, [[ q ]], s \to s}
{{!}Z_1 \to {!}Z_1 & \ldots & {!}Z_N \to {!}Z_N & [[q]] \to \PMod\PMod q & 
s \SL s, \ldots, s \SL s, s \SL s, s \to s}}
$$
}

2. Let $\zeta$ be the stoup of $\Phi$, that is, $\Phi = \zeta; [[q]]$. Then the ``landing'' rule is established as follows:

{\small
$$
\infer[{!}C\mbox{ applied to ${!}Z_i$ in $\zeta$}]
{\zeta; [[q]], \Delta_1, \Delta_2 \to C}
{\infer[{!}P]{\zeta; [{!}Z_i; q], \Delta_1, \Delta_2 \to C}
{\infer[{!}L]{\zeta; [{!}Z_i, q], \Delta_1, \Delta_2 \to C}
{\infer[{!}P,\ Z_i = (\NMod ({!}A_i \cdot \PMod\PMod q)) \SL q]{\zeta; [Z_i; q], \Delta_1, \Delta_2 \to C}
{\infer[\SL L]{\zeta; [(\NMod ({!}A_i \cdot \PMod\PMod q)) \SL q, q], \Delta_1, \Delta_2 \to C}
{q \to q & \infer[\NMod L]{\zeta; [ \NMod ({!}A_i \cdot \PMod\PMod q) ], \Delta_1, \Delta_2 \to C}
{\infer[\cdot L]{\zeta; {!}A_i \cdot \PMod\PMod q, \Delta_1, \Delta_2 \to C}
{\infer[{!}L]{\zeta; {!}A_i, \PMod\PMod q, \Delta_1, \Delta_2 \to C}
{\infer[{!}P]{\zeta, A_i; \PMod\PMod q, \Delta_1, \Delta_2 \to C}
{\infer=[\PMod L]{\zeta; \PMod\PMod q, \Delta_1, A_i, \Delta_2 \to C}
{\zeta; [[ q ]], \Delta_1, A_i, \Delta_2 \to C}}}}}}}}}}
$$
}
\qed
\end{proof}

Now we can finalise our undecidability proof for $\LsysB$.

\begin{lemma}\label{Lm:keyBx}
Let $\Ac_{\Gc} = \{ A_1, \ldots, A_N \}$ and let $Z_i = (\NMod ({!}A_i \cdot \PMod\PMod q)) \SL q$ ($i = 1, \ldots, N$). Then the
sequent 
$${!}((s \SL s) \SL {!} Z_1), {!}{!}Z_1, \ldots, {!}((s \SL s) \SL {!} Z_n), {!}{!}Z_n, 
{!}((s \SL s) \SL \PMod\PMod q), [[ q ]], a_1, \ldots, a_n \to s
$$
is derivable in $\LsysB$ if and only if $s \Rightarrow^*_{\Gc} a_1 \ldots a_n$.
\end{lemma}

\begin{proof}
For the ``if'' direction, we use Proposition~\ref{Pr:internBx} and proceed exactly as in the $1 \Rightarrow 2$ implication from the proof
of Lemma~\ref{Lm:undecRR}. This gives derivability of
$$
(s \SL s) \SL {!} Z_1, {!}Z_1, \ldots, (s \SL s) \SL {!} Z_n, {!}Z_n, 
(s \SL s) \SL \PMod\PMod q; [[ q ]], a_1, \ldots, a_n \to s,
$$
which yields the necessary sequent
$${!}((s \SL s) \SL {!} Z_1), {!}{!}Z_1, \ldots, {!}((s \SL s) \SL {!} Z_n), {!}{!}Z_n, 
{!}((s \SL s) \SL \PMod\PMod q), [[ q ]], a_1, \ldots, a_n \to s
$$
by several applications of ${!}L$.

For the ``only if'' direction, since the sequent in question is derivable in $\LsysB$, it is also derivable in 
$\LsysBfl$ (Proposition~\ref{Prop:nostoupBx}). By cut with ${!}Z_i \to {!}{!}Z_i$ (which is derivable in $\sysBfl$) we get derivability
of
$${!}((s \SL s) \SL {!} Z_1), {!}Z_1, \ldots, {!}((s \SL s) \SL {!} Z_n), {!}Z_n, 
{!}((s \SL s) \SL \PMod\PMod q), [[ q ]], a_1, \ldots, a_n \to s.
$$
Recall that cut is admissible in $\LsysBfl$ by Proposition~\ref{Prop:nostoupB}. Now we use Proposition~\ref{Pr:internB} and
the $4 \Rightarrow 1$ implication of Lemma~\ref{Lm:undecRR} and conclude that $s \Rightarrow^*_\Gc a_1 \ldots a_n$.
\qed
\end{proof}

\begin{theorem}
The derivability problem in $\LsysB$ is undecidable.
\end{theorem}

\begin{proof}
Immediately by Lemma~\ref{Lm:keyBx} and Theorem~\ref{Th:Markov_Gc}.\qed
\end{proof}

\section{Generative Power of Categorial Grammars}\label{S:grammar}

Besides  undecidability results,  
Lemma~\ref{Lm:undecRR} has another corollary:
{\em categorial grammars} based on the calculi with subexponentials considered throughout this paper
generate exactly the class of all recursively enumerable languages.

The definitions of categorial grammars are given in Section~\ref{S:MALC} for calculi without brackets and in Section~\ref{S:calculi} for bracketed systems.
Recall that for the latter we distinguish two notions of recognition, namely s-recognition and t-recognition.

Let us briefly survey known results characterising classes of languages generated by categorial grammars over different extensions of the Lambek calculus.%

For the pure (multiplicative-only) Lambek calculus this class coincides with the class of context-free languages.
The hard part, Lambek to context-free, was done by~\citet{PentusCF}. 
For the easier direction, context-free to Lambek, we refer to Gaifman~\citep{BGS1960}, as the first one who obtained the result,
and~\citet{Buszko1985ZML}, for a modern way of proving it using Greibach normal form~\citep{Greibach}.\footnote{Methods of Gaifman and Buszkowski work only for context-free languages without the empty word. The empty word case
was handled by~\cite{Kuz2012IGPL}.}

Adding the unit constant, $\U$, does not extend the class of languages generated by Lambek grammars---all $\LU$-languages are
still context-free~\citep{Kuz2012FG}.

Grammars based on the Lambek calculus with brackets, $\Lbx$, also generate only context-free languages, both in the
sense of s- and t-recognition. This was initially claimed by~\citet{Jaeger2003}, but, as noticed by~\citet{FaddaMorrill2005},
his proof for t-recognition relied on an incorrect lemma by~\citet{Versmissen1996}. A correct proof was given by~\citet{Kanazawa2018}.

Additive connectives, $\vee$ and $\wedge$, increase the generative power
of Lambek grammars. Namely, as noticed by~\citet{KanazawaJoLLI}, 
$\MALC$-grammars an generate finite intersections of context-free languages and, moreover,
images of such intersections under symbol-to-symbol homomorphisms
(that is, homomorphisms $h \colon \Sigma^* \to \Sigma^*$ 
that map $\Sigma$ to $\Sigma$). Furthermore, as shown by~\cite{Kuz2013}
and~\cite{KuznetsovOkhotin}, the class of $\MALC$-languages
includes the class of languages generated by conjunctive grammars~\cite{OkhotinSurvey}. This latter class is strictly greater than the
class of intersections of context-free languages, but, unless $\mathsf{P} = 
\mathsf{NP}$, is still not closed under symbol-to-symbol homomorphisms~\citep{KuznetsovOkhotin}.

In this section we show that adding the (sub)exponential modality, even constrained by brackets, increases
the power of Lambek categorial grammars to the highest possible level---all
recursively enumerable (r.e.) languages.

From many definitions of the class of r.e. languages we choose the one based on type-0 grammars
(see Section~\ref{S:undec} above for definition). A language $M$ over alphabet $\Sigma$ is r.e. if and only if there
exists a type-0 grammar $\Gc$ such that $M = \{ w \in \Sigma^* \mid s \Rightarrow_\Gc^* w \}$, {\em i.e.,} $M$ is the
language generated by $\Gc$.

Obviously, all languages generated by categorial grammars based on calculi considered in this paper are r.e.
The converse statement is non-trivial, and extends undecidability results of the previous section. For bracketed
calculi, moreover, we have two notions of recognition (s-recognition and t-recognition).

As in the previous section, we provide a generic result. Notice how Lambek's restriction comes into play here.
\begin{theorem}\label{Th:gram_generic}
\begin{enumerate}
\item Let $\Lc$ satisfy the conditions of Theorem~\ref{Th:undec_generic} and additionally admit
$\PMod L$ and $\BS R$ without Lambek's restriction. Then any r.e. language can be generated by  a categorial
grammar based on $\Lc$.
\item Let $\Lc$ satisfy the conditions of Theorem~\ref{Th:undec_generic} and admit
$\PMod L$ and $\BS R$ with Lambek's restriction. Then any r.e. language without the empty word can be generated by  a categorial
grammar based on $\Lc$.
\end{enumerate}
\end{theorem}

\begin{proof}
Consider the type-0 grammar $\Gc$ which generates the given r.e. language.
Let $\Phi_\Gc$ internalise $\Gc$ in $\Lc$. 
We are going to prove that the following three statements are equivalent:
\begin{enumerate}
\item a word $a_1 \ldots a_n$ belongs to the language generated
by $\Gc$ (that is, $s \Rightarrow^*_\Gc a_1 \ldots a_n$);
\item the sequent $a_1, \ldots, a_n \to \prod \Phi_\Gc \BS s$ is derivable in $\Lc$;
\item the sequent $\Delta \to \prod \Phi_\Gc \BS s$ is derivable in $\Lc$ for some bracketing $\Delta$ of
$a_1, \ldots, a_n$.
\end{enumerate}
The meta-formula $\Phi_\Gc$ could contain brackets; for brackets, we define $\prod$ as follows:
$\prod [ \Gamma ] = \PMod (\prod \Gamma)$.

In order to establish \fbox{$1 \Rightarrow 2$}, apply Lemma~\ref{Lm:undecRR}. Let $s \Rightarrow^*_\Gc a_1 \ldots a_n$. We get
$\Phi_\Gc, a_1, \ldots, a_n \to s$ derivable in $\Lc$. Applying $\cdot L$ and $\PMod L$ several times,
we obtain $\prod \Phi_\Gc, a_1, \ldots, a_n \to s$. Now we apply $\SL R$. In Case~1, without Lambek's restriction,
we always get $a_1, \ldots, a_n \to \prod \Phi_\Gc \BS s$. In Case~2, this is possible only for $n > 0$, and in 
this case we consider only type-0 grammars which do not generate the empty word.

The \fbox{$2 \Rightarrow 3$} implication is established by taking the trivial bracketing $\Delta = a_1, \ldots, a_n$. Finally,
for the backwards \fbox{$3 \Rightarrow 1$} implication we use $\pi_q$-soundness of $\Lc$ in $\ELM$ and derive 
$a_1, \ldots, a_n \to \prod \pi_q(\Phi_\Gc) \BS s$ in $\ELM$. 

In $\ELM$, the $\BS R$ rule is invertible using cut:
$$
\infer[\CUT]{F, \Pi \to E}{\Pi \to F \BS E  & \infer[\BS L]{F, F  \BS E \to E}{F \to F & E \to E}}
$$
Thus, we get $\prod \pi_q(\Phi_\Gc), a_1, \ldots, a_n \to s$, and by cut with item~\ref{It:internBack} of Definition~\ref{Df:intern}
we get ${!}A_1, \ldots, {!}A_N, a_1, \ldots, a_n \to s$. The $4 \Rightarrow 1$ implication of Lemma~\ref{Lm:undecRR} finishes the job.

Now the necessary categorial grammar is constructed as follows. The lexicon $\rhd$ is trivial, just the identity relation:
$\rhd = \{ \langle a,a \rangle \mid a \in \Sigma \}$. All the information is kept in the goal formula $H = \prod \pi_q(\Phi_\Gc) \BS s$.
By definition, this grammar generates the same language as $\Gc$, both in the sense of s-recognition and t-recognition. \qed
\end{proof}

Notice that the grammars constructed in Theorem~\ref{Th:gram_generic} have the property of {\em unique type assignment:} for each
letter $a$ there exists exactly one formula $A$ such that $a \rhd A$. For the pure Lambek calculus, constructing such grammars is
much harder. However, for each context-free language without the empty word there exists a Lambek grammar with unique type assignment,
as shown by~\citet{Safiullin}, see also~\citet{Kuzn2017WoLLIC}.

Internalisation properties (Propositions~\ref{Pr:internA}, \ref{Pr:internB}, \ref{Pr:internBx}) now yield the following theorems.

\begin{theorem} 
Let $M$ be a recursively enumerable language. Then
\begin{enumerate}
\item there exists a $\ELM$-grammar which generates $M$;
\item there exists a $\LLsM$-grammar which generates $M$;
\item for each of the calculi $\LsysA$, $\LsysAa$, $\LsysAfl$, $\LsysB$, $\LsysBa$, $\LsysBfl$ there exist
a grammar which generates $M$, both in the sense of s-recognition and t-recognition.
\end{enumerate}
\end{theorem}

\begin{theorem}
Let $M$ be a recursively enumerable language without the empty word. Then
for each of the calculi $\LsysBar$ and $\LsysBrfl$ there exists a grammar 
which generates $M$, both in the sense of s-recognition and t-recognition.
\end{theorem}

\section{Undecidability for One-Division Fragments}\label{S:Buszko}

In this section we strengthen our complexity lower bounds by restricting ourselves to the smallest non-trivial 
fragment of the Lambek calculus with only one division operation, 
extended by ${!}$ and, in the bracketed case, bracket modalities $\PMod$ and $\NMod$. 
For these one-division systems, we obtain the same complexity results. Thus, the situation is different from
the pure Lambek:  while checking derivability in the Lambek calculus (without additives, brackets, and exponentials)
is an NP-complete problem~\citep{PentusNP}, for the one-division fragment there exists a polynomial time
 algorithm~\citep{Savateev2010,KuznetsovTrMIAN2016}. 
In contrast, in the presence of ${!}$ the one-division fragment is as powerful, as the whole system, both with and without brackets.

Our construction is based on Buszkowski's method of encoding type-0 grammars in the one-division fragment of the Lambek calculus
extended with extra (non-logical) axioms~\citep{BuszkoZML}. Then these non-logical axioms are internalised using ${!}$, in the same way 
as in Section~\ref{S:undec}.
 We present a new, simplified version of Buszkowski's construction. This version does not require the type-0 grammar to be
 translated into binary normal form, and, in particular, works also for languages with the empty word.

Recall that the one-division fragment of the Lambek calculus uses formulae constructed from variables using only one division operation, $\SL$.
We denote this fragment by $\Ld$ (the asterisk means that we do not impose Lambek's restriction). Axioms of $\Ld$ are of the form $A \to A$,
and its rules of inference are as follows:
$$
\infer[\SL L]{\Delta_1, B \SL A, \Pi, \Delta_2 \to C}{\Pi \to A & \Delta_1, B, \Delta_2 \to C}
\qquad
\infer[\SL R]{\Pi \to B \SL A}{\Pi, A \to B}
$$
$$
\infer[\CUT]{\Delta_1, \Pi, \Delta_2 \to C}{\Pi \to A & \Delta_1, A, \Delta_2 \to C}
$$

Let us call a {\em B-rule} (after Buszkowski) an inference rule of the following form:
$$
\infer[B_{q_1,\ldots,q_m,r;p_1,\ldots,p_k,t}]{p_1,\ldots, p_k, \Delta \to t}{\Delta, q_1, \ldots, q_m \to r}
$$
In this rule, $q_1,\ldots,q_m,p_1,\ldots,p_k,r$, and $t$ are {\em concrete variables} from $\Var$, not meta-variables.
Each tuple of variables produces an independent B-rule.
On the other hand, $\Delta$ stands for an arbitrary sequence
of Lambek formulae (in the one-division language).

With each rule, we associate the following Gentzen-style rule, which we call a {\em $\mathrm{B}'$-rule:}
$$
\infer[B'_{q_1,\ldots,q_m,r;p_1,\ldots,p_k,t}]{\Pi_1, \ldots, \Pi_k, \Delta \to t}{\Pi_1 \to p_1 & \ldots & \Pi_k \to p_k & \Delta, q_1, \ldots, q_m \to r}
$$
and the following {\em B-formula:}\footnote{Here $E \SL F_1 \ldots F_n$ is used as a shortcut 
for $(E \SL F_n) \SL \ldots \SL F_1$, thus, $B$ is a one-division formula.}
$$
B = (t \SL (r \SL q_1 \ldots q_m)) \SL p_1 \ldots p_k.
$$
For a B-formula $B$, we introduce the corresponding {\em B-axiom} $\Lambda \to B$.

B-rules, B-axioms, and $\mathrm{B}'$-rules will be used as extensions of the Lambek calculus with one division, and each of the three variants has its benefits:
\begin{enumerate}
\item with B-rules, a derivation of a sequent without Lambek connectives ({\em i.e.,} of the form $z_1, \ldots, z_m \to s$) is
non-branching, which will allow us to encode it using a $!$ modality with a restricted version of contraction rule;
\item B-formulae are convenient for incorporating into the main sequent using a ``deduction theorem'' with $!$;
\item finally, the calculus with $\mathrm{B'}$-rules, unlike the first two variants, admits cut elimination (see Lemma~\ref{Lm:LdBcut} below),
which facilitates analysis of derivations.
\end{enumerate}

B-rules, B-axioms, and $\mathrm{B'}$-rules yield equivalent extensions of the Lambek calculus:
\begin{lemma}\label{Lm:eqBB}
In the presence of cut, the extensions of $\Ld$ with: (1) a set of B-rules; (2) the corresponding set of $\mathrm{B}'$-rules;
(3) the corresponding set of B-axioms are equivalent, i.e., derive the same set of sequents.
\end{lemma}

\begin{proof}
\fbox{$(1) \Rightarrow (2)$} (modelling B-rules using $\mathrm{B}'$-rules)
{\small
$$
\infer[B'_{q_1,\ldots,q_m,r;p_1,\ldots,p_k,t}]
{p_1, \ldots, p_k, \Delta \to t}
{p_1 \to p_1 & \ldots & p_k \to p_k & \Delta, q_1, \ldots, q_m \to r}
$$
}

\fbox{$(2) \Rightarrow (3)$} (modelling $\mathrm{B}'$-rules using B-axioms and cut)
{\small
$$
\infer=[\SL L]{(t \SL (r \SL q_1 \ldots q_m)) \SL p_1 \ldots p_k, \Pi_1, \ldots, \Pi_k, \Delta \to t}
{\Pi_1 \to p_1 & \ldots & \Pi_k \to p_k & \infer[\SL L]{t \SL (r \SL q_1 \ldots q_m), \Delta \to t}
{\infer=[\SL R]{\Delta \to r \SL q_1 \ldots q_m}{\Delta, q_1, \ldots, q_m \to r} & t \to t}}
$$
}
Then a cut with the B-axiom $\Lambda \to (t \SL (r \SL q_1 \ldots q_m)) \SL p_1 \ldots p_k$ yields the goal sequent
$\Pi_1, \ldots, \Pi_k, \Delta \to t$.

\fbox{$(3) \Rightarrow (1)$} (deriving B-axioms using B-rules)
{\small
$$
\infer=[\SL R]{\Lambda \to (t \SL (r \SL q_1 \ldots q_m)) \SL p_1 \ldots p_k}
{\infer[\SL R]{p_1, \ldots, p_k \to t \SL (r \SL q_1 \ldots q_m)}
{\infer[B_{q_1,\ldots,q_m,r;p_1,\ldots,p_k,t}]{p_1, \ldots, p_k, r \SL q_1 \ldots q_m \to t}
{\infer=[\SL L]{r \SL q_1 \ldots q_m, q_1, \ldots, q_m \to r}
{q_1 \to q_1 & \ldots & q_m \to q_m & r \to r}}}}
$$
}
\qed
\end{proof}

\begin{lemma}\label{Lm:LdBcut}
The one-division Lambek calculus $\Ld$, extended with an arbitrary finite set of $\mathrm{B}'$-rules, admits cut elimination.
\end{lemma}

\begin{proof}
The proof goes via a standard argument, exactly as for the Lambek calculus itself~\citep{Lambek58}. The global induction is by the number
of cuts in a derivation. Each cut is eliminated by nested induction, where the outer parameter is the complexity of the formula being cut,
and the inner one is the summary derivation depth of the premises of the cut.

The base case is cut with axiom ($A \to A$), which just disappears. For the induction step, one distinguishes principal and non-principal cut premises.
A premise is called principal, if the last rule in its derivation introduces the formula being cut. Thus, if the left premise is introduced by a
$\mathrm{B'}$-rule, it is always principal (since this rule introduces the succedent $t$). The key trick, however, is that if the {\em right} premise of a cut is introduced by a $\mathrm{B}'$-rule, then it is 
{\em never} principal. This is due to the fact that a $\mathrm{B'}$-rule introduces nothing to the antecedent: the antecedent of its goal, $\Pi_1, \ldots, \Pi_k, \Delta$, is composed from antecedents of the premises.

Thus, there are three possible cases.

{\em Case 1:} the left premise is non-principal. Cut can be exchanged with the last rule in the left premise derivation, and propagates upward. The
inner induction parameter gets smaller, while the outer one is intact.
Propagation of cut through non-principal $(\to\SL)$ and $(\SL\to)$ is standard and is performed exactly as in the
cut elimination proof of~\citet{Lambek58}. As for $\mathrm{B}'$-rules, such a rule cannot yield a non-principal left premise of cut.

{\em Case 2:} the right premise in non-principal. Cut propagates to the right.

{This is how cut gets propagated through a $\mathrm{B}'$-rule:

{\small $$
\infer[\CUT]
{\Pi_1, \ldots, \Pi_k, \Delta', \Psi, \Delta'' \to t}{\Psi \to A & \infer[B']{\Pi_1, \ldots, \Pi_k, \Delta', A, \Delta'' \to t}
{\Pi_1 \to p_1 & \ldots & \Pi_k \to p_k & \Delta', A, \Delta'', q_1, \ldots, q_m \to t}}
$$}
transforms into
{\small
$$
\infer[B']
{\Pi_1, \ldots, \Pi_k, \Delta', \Psi, \Delta'' \to t}
{\Pi_1 \to p_1 & \ldots & \Pi_k \to p_k & \infer[\CUT]{\Delta', \Psi, \Delta'', q_1, \ldots, q_m \to t}
{\Psi \to A & \Delta', A, \Delta'', q_1, \ldots, q_m \to t}}
$$}
and
{\small
$$
\infer[\CUT]
{\Pi_1, \ldots, \Pi'_i, \Psi, \Pi''_i, \ldots, \Pi_k, \Delta \to t}{\Psi \to A & 
\infer[B']{\Pi_1, \ldots, \Pi'_i, A, \Pi''_i, \ldots, \Pi_k, \Delta \to t}
{\Pi_1 \to p_1 & \ldots & \Pi'_i, A, \Pi''_i \to p_i & \ldots & \Pi_k \to p_k & \Delta, q_1, \ldots, q_m \to t}}
$$}
transforms into
{\small
$$
\infer[B']
{\Pi_1, \ldots, \Pi'_i, \Psi, \Pi''_i, \ldots, \Pi_k, \Delta \to t}
{\Pi_1 \to p_1 & \ldots & 
\infer[\CUT]{\Pi'_i, \Psi, \Pi''_i \to p_i}{\Psi \to A & \Pi'_i, A, \Pi''_i \to p_i} & \ldots & \Pi_k \to p_k &
\Delta, q_1, \ldots, q_m \to t}
$$}

\normalsize

}

Propagation of cut to the right through non-principal $\SL R$ and $\SL L$ is again due to Lambek.

{\em Case 3:} both left and right premises are principal, being introduced by $\SL R$ and $\SL L$ respectively. In this case
cut transforms into two cuts of lower complexity:
$$
\infer[\CUT]{\Gamma, \Psi, \Pi, \Delta \to C}
{\infer[\SL R]{\Psi \to E \SL F}{\Psi, F \to E} & 
\infer[\SL L]{\Gamma, E \SL F,  \Pi, \Delta \to C}
{\Pi \to F & \Gamma, E, \Delta \to C}}
$$
becomes
$$
\infer[\CUT]{\Gamma, \Psi, \Pi, \Delta \to C}
{\Pi \to F & \infer[\CUT]{\Gamma, \Psi, F, \Delta \to C}
{\Psi, F \to E & \Gamma, E, \Delta \to C}}
$$
(This transformation, again, comes from the original Lambek's proof.)
\qed
\end{proof}

Now we are ready to present the encoding of type-0 grammars. Consider a grammar $\Gc = \langle N, \Sigma, P, s \rangle$.
For each production $\pf = (v_1 \ldots v_m \Rightarrow w_1 \ldots w_k) \in P$ add $\bap$, $\bbp$, $\bcp$, $\bdp$, $\bep$, $\bfp$, and
$\typ$, for each $y \in N \cup \Sigma$, as distinct variables to $\Var$, and consider the following seven B-rules:
\begin{align*}
& \infer[(1_{\pf})]{\bep, \Delta \to \bap}{\Delta \to s} && \infer[(4_{\pf})]{\typ, \Delta \to \bbp}{\Delta, y \to \bbp} \\
& \infer[(2_{\pf})]{\typ, \Delta \to \bap}{\Delta, y \to \bap} && \infer[(5_{\pf})]{\bfp, \Delta \to \bcp}{\Delta, \bep \to \bbp} \\
& \infer[(3_{\pf})]{\twp_1, \ldots, \twp_k, \Delta \to \bbp}{\Delta, v_1, \ldots, v_m \to \bap}
           && \infer[(6_{\pf})]{y, \Delta \to \bcp}{\Delta, \typ \to \bcp} \\
& \infer[(7_{\pf})]{\Delta \to s}{\Delta, \bfp \to \bcp}
\end{align*}

By $\mathbf{B}_\Gc$ we denote the set of all B-rules obtained from production rules of $\Gc$ as shown above;
let $\mathcal{B}_\Gc$ be the corresponding set of B-formulae and $\mathbf{B}'_\Gc$ be the corresponding set of
$\mathrm{B}'$-rules.

Before going further, let us comment a bit on these B-rules. In the language without product, we cannot directly implement the `landing' rule which replaces one subword with another (that is, applies a semi-Thue transition) at an arbitrary place of the antecedent. However, if we manage to move the subword to the right-hand side of the antecedent, it can be indeed replaced by another one (and moved to the left-hand side) by a B-rule, which is our main rule $(3_\pf)$. Other rules do the necessary preparations.
This idea is essentially due to Buszkowski; here we present it more straightforwardly. First, $(1_\pf)$ starts the replacement procedure. Second, several applications of $(2_\pf)$ rotate the antecedent so that the necessary subword is on the right-hand side of the antecedent. The usage of an alternative alphabet ($\typ$ instead of $y$) and special variables in the succedent ($\bap$, ...) here ensures that this process cannot be aborted, and other rules cannot be applied until we finish. Third, as said above, $(3_\pf)$ performs the actually semi-Thue transition. Finally, $(4_\pf)$--$(7_\pf)$ perform the backwards rotation and quit the procedure. This strategy is formalized in the proof of the key Lemma~\ref{Lm:undecRRBuszko} below, which is the 
version of Lemma~\ref{Lm:undecRR} for encoding Busz\-kow\-ski's rules.

\begin{lemma}\label{Lm:undecRRBuszko}
Let $\Lc$ be $\pi_q$-sound in $\ELM$ and admit the $\SL L$ and $\SL R$ rules (maybe with Lambek's restriction for the latter).
Let $\Psi_{\Gc}$ internalise $\Bc_{\Gc} = \{ B_1, \ldots, B_N \}$ in $\Lc$ (see Definition~\ref{Df:intern}).
Also let all formulae in $\Psi_{\Gc}$ be ${!}$-formulae, for which permutation rules are allowed in $\Lc$.
Then the following are equivalent:
\begin{enumerate}
\item $s \Rightarrow^*_{\Gc} z_1 \ldots z_n$;
\item $z_1, \ldots, z_n \to s$ is derivable from axiom $s \to s$, using only rules from $\mathbf{B}_{\Gc}$, without cut;
\item $\Psi_{\Gc}, z_1, \ldots, z_n \to s$ is derivable in~$\Lc$.
\item there exists such a bracketing $\Delta$ of $z_1, \ldots, z_n$ that the sequent $\Psi_{\Gc}$ is derivable in~$\Lc$;
\item the sequent ${!}B_1, \ldots, {!}B_N, z_1, \ldots, z_n \to s$ is derivable in $\EL$;
\item $z_1, \ldots, z_n \to s$ is derivable in $\Ld$ extended with rules from $\mathbf{B}'_\Gc$.
\end{enumerate}
\end{lemma}

\begin{proof}
This proof shares much with the proof of Lemma~\ref{Lm:undecRR}. 

\fbox{$1 \Rightarrow 2$} 
Proceed by induction on the derivation of $z_1 \ldots z_n$ from $s$ in $\Gc$.
The base case, $s \Rightarrow^*_\Gc s$, corresponds to the $s \to s$ axiom. For the induction step, consider the last
production rule $\pf = (v_1 \ldots v_m  \Rightarrow w_1 \ldots w_k)$
 applied in the derivation: $$s \Rightarrow^*_\Gc z_1 \ldots z_i v_1 \ldots v_m z_j \ldots z_n \Rightarrow_\Gc
z_1 \ldots z_i w_1 \ldots w_k z_j \ldots z_n.$$ By induction hypothesis, the sequent
$z_1, \ldots, z_i, v_1, \ldots, v_m, z_j, \ldots, z_n \to s$ is derivable.  The necessary sequent
$z_1, \ldots, z_i, w_1, \ldots, w_k, z_j, \ldots, z_n \to s$ is now derived as follows:
$$
\infer[(7_\pf)]{z_1, \ldots, z_i, w_1, \ldots, w_k, z_j, \ldots, z_n \to s}
{\infer=[(6_\pf)]{z_1, \ldots, z_i, w_1, \ldots, w_k, z_j, \ldots, z_n, \bfp \to \bcp}
{\infer[(5_\pf)]{\bfp, \tzp_1, \ldots, \tzp_i, \twp_1, \ldots, \twp_k, \tzp_j, \ldots, \tzp_n \to \bcp}
{\infer=[(4_\pf)]{\tzp_1, \ldots, \tzp_i, \twp_1, \ldots, \twp_k, \tzp_j, \ldots, \tzp_n, \bep \to \bbp}
{\infer[(3_\pf)]{\twp_1, \ldots, \twp_k, \tzp_j, \ldots, \tzp_n, \bep, z_1, \ldots, z_i \to \bbp}
{\infer=[(2_\pf)]{\tzp_j, \ldots, \tzp_n, \bep, z_1, \ldots, z_i, v_1, \ldots, v_m \to \bap}
{\infer[(1_\pf)]{\bep, z_1, \ldots, z_i, v_1, \ldots, v_m, z_j, \ldots, z_n \to \bap}
{z_1, \ldots, z_i, v_1, \ldots, v_m, z_j, \ldots, z_n \to s}}}}}}}
$$

\fbox{$2 \Rightarrow 3$}
Proceed by induction on derivation. 
The base case, $\Phi, s \to s$, is derivable by item~1 of Definition~\ref{Df:intern}. The induction step,
{\em i.e.,} application of a B-rule of the form
$$
\infer[B]{p_1, \ldots, p_k, \Delta \to t}{\Delta, q_1, \ldots, q_m \to r}
$$
is handled using the `landing' rule (item~2 of Definition~\ref{Df:intern}) as follows:

{\small
$$
\infer[\mathrm{land}]{\Psi_{\Gc}, p_1, \ldots, p_k, \Delta \to t}
{\infer=[\SL L]{\Psi_{\Gc}, (t \SL (r \SL q_1 \ldots q_m)) \SL p_1 \ldots p_k, p_1, \ldots, p_k, \Delta \to t}
{p_1 \to p_1 & \ldots & p_k \to p_k & 
\infer=[!P_2\mbox{ applied to formulae of $\Psi_{\Gc}$}]{\Psi_{\Gc}, t \SL (r \SL q_1 \ldots q_m), \Delta \to t}
{\infer[\SL L]{t \SL (r \SL q_1 \ldots q_m), \Psi_{\Gc}, \Delta \to t}
{\infer=[\SL R]{\Psi_{\Gc}, \Delta \to r \SL q_1 \ldots q_m}{\Psi_{\Gc}, \Delta, q_1, \ldots, q_m \to r} & t \to t}}}}
$$}
Notice that here we essentially used the permutation rules for formulae of $\Psi_{\Gc}$.

\fbox{$3 \Rightarrow 4$} Obvious: take the trivial (empty) bracketing $\Delta = z_1, \ldots, z_n$.

\fbox{$4 \Rightarrow 5$} is handled exactly as in Lemma~\ref{Lm:undecRR}. 

\fbox{$5 \Rightarrow 6$} Consider a cut-free derivation of ${!}B_1, \ldots, {!}B_N, z_1, \ldots, z_n$ in $\EL$
and erase all $!$-formulae from it. Then applications of structural rules for ${!}$ become trivial, and $!L$ transforms into
$$
\infer{\Delta_1, \Delta_2 \to A}{\Delta_1, B, \Delta_2 \to A}
$$
where $B$ is a B-formula from $\Bc_\Gc$. This is equivalent to cut with the B-axiom $\Lambda \to B$. 

Thus, $z_1, \ldots, z_n$ is derivable in $\Ld$ extended by the set of B-axioms obtained from $\Gc$ and, by Lemma~\ref{Lm:eqBB}, in the corresponding extension by $\mathrm{B}'$-rules,
$\mathbf{B}'_\Gc$.

\fbox{$6 \Rightarrow 1$}
The  extension of the Lambek calculus with $\mathbf{B}'_\Gc$, admits
cut elimination (Lemma~\ref{Lm:LdBcut}), and in the cut-free derivation the only rules that can be applied
are $\mathrm{B}'$-rules.

Proceed by induction on this derivation. The base case is the $s \to s$ axiom, and we have $s \Rightarrow_\Gc^* s$.
For the induction step, let us go upwards along the derivation, turning right at each application of a $\mathrm{B}'$-rule, and trace the succedent:
$$
\xymatrix{
s \ar[r]_{(7'_\pf)} &
\ar@(ul,ur)^{(6'_\pf)} \bcp \ar[r]_{(5'_\pf)} &
\ar@(ul,ur)^{(4'_\pf)} \bbp \ar[r]_{(3'_\pf)} &
\ar@(ul,ur)^{(2'_\pf)} \bap \ar[r]_{(1'_\pf)} &
s 
}
$$
(Since variables $\bap$, $\bbp$, and $\bcp$ could never appear in antecedents,
the derivation cannot stop at an axiom of the form $\bap \to \bap$ or alike.)

Essentially, as we shall see below, once we started with $(7'_\pf)$, we fix the production rule $\pf$
and perform, as a whole, the block of $\mathrm{B}'$-rules which emulates application of $\pf$
(as the last production rule in the derivation). Then we return to a sequent of the form $\Delta \to s$,
ready to perform our backtracking further.

Variables $\bdp$, $\bep$, $\bfp$, and $\typ$ ($y \in N \cup \Sigma$) are never
 succedents of conclusions of rules from $\mathbf{B}'_\Gc$. Therefore, left premises of the rules
$(1'_\pf)$--$(5'_\pf)$, which are of the form $\Pi_i \to p_i$, where $p_i$ is one of the
aforementioned variables, could only be axioms $p_i \to p_i$. This means that
$(1'_\pf)$--$(5'_\pf)$ actually transform into the corresponding B-rules,
$(1_\pf)$--$(5_\pf)$. As for $(7'_\pf)$, it already coincides with $(7_\pf)$.

Thus, the bottom of our derivation looks as follows, where $\Delta_1, \ldots, \Delta_{n'} = z_1, \ldots, z_n$:

$$
\infer[(7_\pf)]
{\Delta_1, \ldots, \Delta_{n'} \to s}
{\infer=[(6'_\pf)]{\Delta_1, \ldots, \Delta_{n'}, \bfp \to \bcp}
{\Delta_1 \to y_1 & \ldots & \Delta_{n'} \to y_{n'} &
\infer[(5_\pf)]{\bfp, \typ_1, \ldots, \typ_i, \twp_1, \ldots, \twp_k, \typ_j, \ldots, \typ_{n'} \to \bcp}
{\infer=[(4_\pf)]{\typ_1, \ldots, \typ_i, \twp_1, \ldots, \twp_k, \typ_j, \ldots, \typ_{n'}, \bep \to \bbp}
{\infer[(3_\pf)]{\twp_1, \ldots, \twp_k, \typ_j, \ldots, \typ_{n'}, \bep, y_1, \ldots, y_i \to \bbp}
{\infer=[(2_\pf)]{\typ_j, \ldots, \typ_{n'}, \bep, y_1, \ldots, y_i, v_1, \ldots, v_m \to \bap}
{\infer[(1_\pf)]{\bep, y_1, \ldots, y_i, v_1, \ldots, v_m, y_j, \ldots, y_{n'} \to \bap}
{y_1, \ldots, y_i, v_1, \ldots, v_m, y_j, \ldots, y_{n'} \to s}}}}}}}
$$

Consider sequents of the form $\Delta_i \to y_i$ (left premises); $w_1, \ldots, w_k$ are also $y$'s. If $y_i \ne s$, then it could not be the succedent
of the conclusion of a rule from $\mathbf{B}'_\Gc$, therefore $\Delta_i = y_i$ and this is just an axiom. If $y_i = s$,
then by induction hypothesis we have $\Delta_i$ derivable from $s$ in $\Gc$. Thus, in both cases\footnote{This part can
be simplified a bit by modifying $\Gc$. Namely, we could introduce a new starting symbol $s'$ with a rule
$s' \Rightarrow s$. The language generated by $\Gc$ will not change. After this transformation, the starting symbol $s'$ will never
appear in the derivation, except for its start, and therefore there would be always $y_i \ne s'$, and $\Delta_i = y_i$.}
 $\Gc$ derives
$\Delta_i$ from $y_i$.

By induction hypothesis we have $s \Rightarrow_\Gc^* y_1 \ldots y_i v_1 \ldots v_m y_j \ldots y_{n'}$, and since
$\pf = (v_1 \ldots v_m \Rightarrow w_1 \ldots w_k)$ is a production rule of $\Gc$ (the form of $\pf$ is taken from
$(3_\pf)$), $s \Rightarrow_\Gc^* y_1 \ldots y_i w_1 \ldots w_k y_j \ldots y_{n'}$. Finally, we recall a well-known property of derivations in type-0 grammars: if $s \Rightarrow_\Gc^* y_1 \dots y_{n'}$ ($w_i$'s are also part of $y_j$'s) and $y_i \Rightarrow_\Gc^* \Delta_i$ for each $i$, then $s \Rightarrow^*_\Gc \Delta_1 \ldots \Delta_{n'} = z_1 \ldots z_n$.
\qed
\end{proof}

This lemma yields results on complexity and generative power of categorial grammars for one-division fragments, exactly as Lemma~\ref{Lm:undecRR} does in the
general case. 

\begin{theorem}
The derivability problems for one-division fragments (that is, fragments including $\SL$, ${!}$, brackets and bracket modalities) of
$\ELM$, $\LLsM$, $\LsysAfl$, $\LsysAa$, and $\LsysA$ are undecidable.
\end{theorem}

\begin{theorem}
For any r.e. language $M$ and for each of the calculi mentioned in the previous theorem there exists a categorial grammar for $M$ based on the
given calculus. For bracketed systems, such a grammar both s-recognises and t-recognises $M$.
\end{theorem}

For $\sysBb$ and $\sysBr$, however, we cannot directly use the internalisation given by Proposition~\ref{Pr:internB},
since the meta-formula $\Phi$ used there includes the product connective. Also, $\Phi$ includes $[[q]]$, which is not
a $!$-formula and does not allow permutation, as required in Lemma~\ref{Lm:undecRRBuszko}.

We overcome this issue by slightly modifying the notion of internalisation and proving a new version of Proposition~\ref{Pr:internB}.

\begin{definition}\label{Df:Vintern}
Let $\Bc = \{ B_1, \ldots, B_N \}$ be a finite set of formulae and $\mathcal{V} \subseteq \Var$ be a finite set of
variables. 
A meta-formula $\Psi$ {\em $\mathcal{V}$-internalises} $\Bc$ in the calculus $\Lc$, if the following holds:
\begin{enumerate}
\item the sequent $\Psi, s \to s$ is derivable in $\Lc$;
\item the following `$t$-landing' rule is admissible in $\Lc$ for any $t \in \mathcal{V}$:
$$
\infer[\mathrm{land}_t,\ B_i \in \Bc]{\Psi, \Delta_1, B_i, \Delta_2 \to t}{\Psi, \Delta_1, \Delta_2 \to t}
$$
\item the sequent ${!}B_1, \ldots, {!}B_N \to \prod \pi_q(\Psi)$ is derivable in $\ELM$.
\end{enumerate}
\end{definition}

The new notion of $\mathcal{V}$-internalisation differs from the original notion of internalisation (Definition~\ref{Df:intern}) in item~2.
This item is formulated in a weaker form: we restrict the antecedents of sequents in the `landing' rule by a finite set
$\mathcal{V}$ of variables. The key observation is that $\mathcal{V}$-internalisation, where $\mathcal{V}$ is the set of all variables
used in $\Bc$, is already sufficient for Lemma~\ref{Lm:undecRRBuszko}. Thus, now we only have to prove the $\mathcal{V}$-internalisation property for
$\LsysBa$.

\begin{proposition}
Let $\mathcal{V} = \{ t_1, \ldots, t_m \}$ be a finite set of variables and $\Bc = \{ B_1, \ldots, B_N \}$ be
a finite set of Lambek formulae. Then the following meta-formula
$$
\Psi_{\Bc,\mathcal{V}} = {!}((s \SL s) \SL {!}Z_{1,1}), {!}Z_{1,1}, \ldots, {!}((s \SL s) \SL {!}Z_{m,N}), {!}Z_{m,N},
{!}((s \SL s) \SL \PMod\PMod q),  {!}\PMod\PMod q,
$$
where 
$$
Z_{i,j} = (\NMod (t_j \SL ((t_j \SL {!}B_i) \SL {!}\PMod\PMod q))) \SL q,
$$
$\mathcal{V}$-internalises $\Bc$ in $\LsysBfl$ and in $\LsysBrfl$.
\end{proposition}

\begin{proof}

For short, denote $\Psi_{\Bc,\mathcal{V}}$ by just $\Psi$. Item~1 of Definition~\ref{Df:Vintern} is checked exactly as in Proposition~\ref{Pr:internB}.

For item~2, let us check the $t$-landing rule for $t = t_j \in \mathcal{V}$ and $B_i \in \Bc$.
Let $\Psi = \Psi', {!}\PMod\PMod q$.

{\small
$$
\infer[!L]{\Psi', {!}\PMod\PMod q, \Delta_1, \Delta_2 \to t_j}
{\infer[\PMod L]{\Psi', \PMod\PMod q, \Delta_1, \Delta_2 \to t_j}
{\infer[\PMod L]{\Psi', [\PMod q], \Delta_1, \Delta_2 \to t_j}
{\infer[!C\mbox{ applied to ${!}Z_{i,j}$}]
{\Psi', [[q]], \Delta_1, \Delta_2 \to t_j}
{\infer[!L]{\Psi', [!((\NMod (t_j \SL ((t_j \SL {!}B_i) \SL ({!}\PMod\PMod q)))) \SL q), q], \Delta_1, \Delta_2 \to t_j}
{\infer[\SL L]{\Psi', [(\NMod (t_j \SL ((t_j \SL {!}B_i) \SL ({!}\PMod\PMod q)))) \SL q, q], \Delta_1, \Delta_2 \to t_j}
{q \to q & \infer[\NMod L]{\Psi', [\NMod (t_j \SL ((t_j \SL {!}B_i) \SL ({!}\PMod\PMod q))))], \Delta_1, \Delta_2 \to t_j}
{\infer=[!P_1\mbox{ applied to formulae of $\Psi'$}]{\Psi', t_j \SL ((t_j \SL {!}B_i) \SL ({!}\PMod\PMod q))), \Delta_1, \Delta_2 \to t_j}
{\infer[\SL L]{t_j \SL ((t_j \SL {!}B_i) \SL ({!}\PMod\PMod q))), \Psi', \Delta_1, \Delta_2 \to t_j}
{\infer[\BS R]{\Psi', \Delta_1, \Delta_2 \to (t_j \SL {!}B_i) \SL ({!}\PMod\PMod q)}
{\infer[\BS R]{\Psi', \Delta_1, \Delta_2, {!} \PMod\PMod q \to t_j \SL {!}B_i}
{\infer[!P_2]{\Psi', \Delta_1, \Delta_2, {!}\PMod\PMod q, {!}B_i \to t_j}
{\infer[!L]{\Psi', \Delta_1, {!}B_i, \Delta_2, {!}\PMod\PMod q \to t_j}
{\infer[!P_2]{\Psi', \Delta_1, B_i, \Delta_2, {!}\PMod\PMod q \to t_j}
{\Psi', {!}\PMod\PMod q, \Delta_1, B_i, \Delta_2 \to t_j
}}} }}
& t_j \to t_j}}}}}}
}}}
$$
}

Finally, let us check item~3. The $\pi_q$-projection of $\Psi$ includes the following formulae (the order does not matter due to permutation rules):
\begin{enumerate}
\item ${!}((s \SL s) \SL {!} \pi_q(Z_{i,j}))$;
\item ${!} Z_{i,j} = {!}((t_j \SL ((t_j \SL {!}B_i) \SL {!}\U))) \SL\U)$;
\item ${!} ((s \SL s) \SL \U)$ and ${!}\U$.
\end{enumerate}

We enjoy the following derivations in $\ELM$:
$$
\infer[{!} R]{\Lambda \to {!}((s \SL s) \SL {!} \pi_q(Z_{i,j}))}
{\infer[\SL R]{\Lambda \to (s \SL s) \SL {!} \pi_q(Z_{i,j})}
{\infer[{!} W]{{!} \pi_q(Z_{i,j}) \to s \SL s}{\infer[\SL R]{\Lambda \to s \SL s}{s \to s}}}}
\qquad
\infer [{!} R]{\Lambda \to {!}((s \SL s) \SL \U)}
{\infer[\SL R]{\Lambda \to (s \SL s) \SL \U}{\infer[\U L]{\U \to s \SL s}
{\infer[\SL R]{\Lambda \to s \SL s}{s \to s}}}}
$$
$$
\infer[{!} R]{{!}B_i \to {!}((t_j \SL ((t_j \SL {!}B_i) \SL {!}\U)) \SL\U)}
{\infer[\SL R]{{!}B_i \to (t_j \SL ((t_j \SL {!}B_i) \SL {!}\U)) \SL\U}
{\infer[\U L]{{!}B_i, \U \to t_j \SL ((t_j \SL {!}B_i) \SL {!}\U)}
{\infer[\SL R]{{!}B_i \to t_j \SL ((t_j \SL {!}B_i) \SL {!}\U)}
{\infer[\SL L]{{!}B_i, (t_j \SL {!}B_i) \SL {!}\U \to t_j}
{\infer[{!} R]{\Lambda \to {!}\U}{\Lambda\to\U} & \infer[{!}P]{{!}B_i, t_j \SL {!}B_i \to t_j}{\infer[\SL L]{t_j \SL {!}B_i, {!}B_i \to t_j}
{{!}B_i \to {!}B_i & t_j \to t_j}}}}}}}
$$
By $\cdot R$, we derive ${!}B_1, \ldots, {!}B_N, \ldots, {!}B_1, \ldots, {!}B_N \to \prod \pi_q(\Psi)$
(here ${!}B_1, \ldots, {!}B_N$) is repeated $m$ times. Permutations and contractions yield  the needed sequent
 ${!}B_1, \ldots, {!}B_N \to \prod \pi_q(\Psi)$. \qed
\end{proof}

By Proposition~\ref{Prop:nostoupB}, we propagate this construction to $\LsysBa$ and $\LsysBar$.
Finally, for the original  Morrill's system $\LsysB$ we use the same trick as in Section~\ref{S:undec} (Proposition~\ref{Pr:internBx}),
adding an extra ${!}$ over $Z_{i,j}$. In a whole, this yields the necessary results.

\begin{theorem}
The derivability problems for one-division fragments (that is, fragments including $\SL$, ${!}$, brackets and bracket modalities) of
$\LsysBfl$, $\LsysBrfl$, $\LsysBa$, $\LsysBar$, and $\LsysB$ are undecidable.
\end{theorem}

\begin{theorem}
For any r.e. language $M$ and for each of the calculi mentioned in the previous theorem there exists a categorial grammar based on the
given calculus which both s-recognises and t-recognises $M$, for calculi without Lambek's restriction ($\LsysBfl$, $\LsysBa$, and $\LsysB$), and
$M - \{\varepsilon\}$, in the case with Lambek's restriction ($\LsysBrfl$ and $\LsysBar$).
\end{theorem}

\section{Conclusions and Future Work}\label{S:conclusion}

In this article, we have performed the logical analysis of two systems introduced by Morrill as a base for the CatLog categorial grammar
parser, $\sysA$ and $\sysB$. We have pointed out issues with cut elimination in these systems, and provided necessary modification for which cut elimination is proved.
We also discussed how Lambek's non-empti\-ness restriction can be imposed on $\sysB$. From the algorithmic point of view, we have proved undecidability for each of
Morrill's systems, even in the smallest possible language with only one division, brackets and bracket modalities, and the subexponential. Moreover, we have
shown that categorial grammars based on Morrill's calculi can generate arbitrary recursively enumerable languages (in the case with Lambek's restriction---arbitrary
r.e. languages without the empty word).

One of the most interesting questions for future research is as follows. The undecidability results presented in this article look unfortunate, since
the calculi it is applied to are intended to be used in natural language parsing software. 

Thus, it is an important task to explore fragments of the calculi,
guarded by certain syntactic conditions on applying $!$, for which the derivability problem is decidable. For systems without brackets, in particular,
$\LLs$, such an algorithm exists under the condition that $!$ is applied only to variables~\citep{KanKuzSceFG}. The complexity of this algorithm
is the same as for the calculus without ${!}$: NP for $\LLsM$ and PSPACE for $\LLs$. Moreover, for $\ELM$  decidability is known for a broader class
of formulae allowed under ${!}$, namely, formulae of implication depth 1, that is, of the form $p_1 \ldots p_k \BS q \SL r_1 \ldots r_m$~\citep{Fofanova2018}.
Extending this result to $\EL$ and $\LLs$ is still an open question.

For systems with brackets, the class of formulae for which the derivability problem becomes decidable appears to be much broader. For Morrill's first 
system, $\sysA$, this class is guarded by so-called bracket non-negative condition (BNNC) imposed on ${!}$-formulae. Under this condition, ${!}$ can be applied
to any formula which does not include negative occurrences of $\PMod$ and does not include positive occurrences of $\NMod$. In particular, ${!}$ is allowed
to be applied to any formula which does not include bracket modalities at all, no matter how complex this formula is. \citet{MorrillValentin} show that
the derivability problem in $\sysA$ for sequents obeying BNNC is decidable; \citet{KanKuzSceFCT} establish an NP upper complexity bound for its fragment without additives, also with BNNC imposed.
We conjecture that for the full system, including additives, the complexity bound is PSPACE. These complexity boundaries are tight, since the multiplicative-only Lambek calculus is already NP-complete~\citep{PentusNP} and $\MALC$ is PSPACE-complete~\citep{KanovichKazimierz}. 
For Morrill's  second system, $\sysB$, formulating
the corresponding version of BNNC and proving decidability for the fragment guarded by this new condition is a problem for further investigation.

Another, potentially simpler but more technical question left for further research is the question of extending our cut elimination proof to calculi with discontinuous operations~\citep{MorValDispl}. We conjecture that the proof could be obtained as a combination of our proof (using ``deep cut elimination'' for ${!}$-formulae) presented here and the proof by~\citet{MorValDispl} for displacement calculus. The notations, however, would become extremely complicated---thus, a digestable presentation of such a proof becomes a separate challenge.

\paragraph*{Acknowledgements.}\
We are grateful to Glyn Morrill for a number of very helpful interactions we benefited from at various stages of this work. 
The work of Max Kanovich was partially supported by
EPSRC Programme Grant EP/R006865/1: ``Interface Reasoning for Interacting
Systems (IRIS).'' The work of Andre Scedrov and Stepan Kuznetsov was
prepared within the framework of the HSE University Basic Research Program
and partially funded by the Russian Academic Excellence Project `5--100.' The
work of Stepan Kuznetsov was also partially supported by the Council of the President of Russia for Support of Young Russian Researchers and Leading Research Schools of the Russian Federation, by the Young Russian Mathematics Award, and by
the Russian Foundation for Basic Research grant 20-01-00435.


\bibliographystyle{bbs}
\bibliography{subexponential.bib}


\end{document}